\documentclass[12pt,peerreview,onecolumn]{IEEEtranTCOM}
\pdfoutput=1
\usepackage{amsmath,epsfig,amssymb,verbatim,amsopn,subfigure,color}
\usepackage{cite,xspace,exscale}
\usepackage{array,algorithm,algorithmic}
\usepackage{bbm}
\usepackage{array}
\usepackage{multirow}
\usepackage{multicol}
\usepackage{rotating}
\usepackage{dsfont,float}
\usepackage{stfloats}
\usepackage{enumerate,makecell}
\normalsize

\ifCLASSOPTIONpeerreview
\renewcommand{\baselinestretch}{1.77}
\fi

\ifCLASSOPTIONonecolumn
    
\else
    
\fi

\makeatletter
\let\l@ENGLISH\l@english
\makeatother

\makeatletter
\renewcommand*{\@opargbegintheorem}[3]{\trivlist
  \item[\hskip \labelsep{\itshape #1\ #2}] {\itshape (#3):} {\normalfont}}
\makeatother
\usepackage{relsize}

\graphicspath{{figs/},{Figures/}}

\newtheorem{lemma}{Lemma}
\newtheorem{remark}{Remark}

\newtheorem{proposition}{Proposition}

\newtheorem{define}{Definition}
\graphicspath{{figs/},{Figures/}}

\ifCLASSOPTIONpeerreview
\else
\fi

\newcommand{\suc}{\textrm{suc}}

\newcommand{\out}{\textrm{out}}
\newcommand{\near}{\textrm{near}}
\newcommand{\far}{\textrm{far}}
\newcommand{\AuthorOne}{Jing~Guo, {\em{Member, IEEE}}}
\newcommand{\AuthorTwo}{Salman~Durrani, {\em{Senior Member, IEEE}}}
\newcommand{\AuthorThree}{Xiangyun~Zhou, {\em{Senior Member, IEEE}}}
\newcommand{\AuthorFour}{Halim~Yanikomeroglu, {\em{Fellow, IEEE}}}
\newcommand{\ThankOne}{J. Guo, X. Zhou and S. Durrani are with the Research School of Engineering, The Australian National University, Canberra, ACT 2601, Australia (Emails: \{jing.guo, xiangyun.zhou, salman.durrani\}@anu.edu.au). H. Yanikomeroglu is with the Department of Systems and Computer Engineering, Carleton University, Ottawa, ON K1S 5B6, Canada (E-mail: halim@sce.carleton.ca).}

\begin{document}
\clearpage
This work has been submitted to the IEEE for possible publication.
Copyright may be transferred without notice, after which this version
may no longer be accessible.
  \newpage
\setcounter{page}{1}

\title{Design of Non-orthogonal Multiple Access Enhanced Backscatter Communication}
\author{\IEEEauthorblockN{\AuthorOne,~\AuthorThree,~\AuthorTwo,~and~\AuthorFour\thanks{\ThankOne}}}
\maketitle
%

%
\begin{abstract}
Backscatter communication (BackCom), which allows a backscatter node (BN) to communicate with the reader by modulating and reflecting the incident continuous wave from the reader, is considered as a promising solution to power the future Internet-of-Things. In this paper, we consider a single BackCom system, where multiple BNs are served by a reader. We propose to use the power-domain non-orthogonal multiple access (NOMA), i.e., multiplexing the BNs in different regions or with different backscattered power levels, to enhance the spectrum efficiency of the BackCom system. To better exploit power-domain NOMA, we propose to set the reflection coefficients for multiplexed BNs to be different. Based on this considered model, we develop the reflection coefficient selection criteria. To illustrate the enhanced system with the proposed criteria, we analyze the performance of BackCom system in terms of the average number of bits that can be successfully decoded by the reader for two-node pairing case and the average number of successful BNs for the general multiplexing case. Our results shows that NOMA achieves much better performance gain in the BackCom system as compared to its performance gain in the conventional system, which highlights the importance of applying NOMA to the BackCom system.
\end{abstract}
%

\ifCLASSOPTIONpeerreview
    \newpage
\fi

\section{Introduction}
To provide ubiquitous connectivity among tens of billions devices, the internet-of-things (IoT) is envisaged as one of the key technology trends for the fifth generation (5G) system~\cite{8030485}. Under the IoT paradigm, the low-cost devices can automatically communication with each other without human intervention. Nonetheless, with the development of IoT technology, there are currently many research challenges needed to be addressed, one of them being the energy issue~\cite{Dawy-2017,8030504}. For those devices where the battery replacement can be very costly, the energy harvesting becomes a desirable approach to maintain the functionality of devices for a long period. It is worthy to note that the energy harvesting can be very compatible with most IoT devices, because these devices only consume a small amount of energy~\cite{Dawy-2017,Jayakody-2017}.

One of the promising energy harvesting techniques is the backscatter communication (BackCom)~\cite{Lium-2017}. A BackCom system generally has two main components, a reader and a backscatter node (BN). The BN does not have any active radio frequency (RF) component, and it reflects and modulates the incident single-tone sinusoidal continuous wave (CW) from the reader for the uplink communication. The reflection is achieved by intentionally mismatching the antennas input impedance and the signal encoding is achieved by varying the antenna impedance~\cite{Boyer-2014}. The BN can also harvests the energy from the CW signal. These energy-saving features make the BackCom system become a prospective candidate for IoT.

The backscatter technique is commonly used in the radio frequency identification systems (RFID), which usually accommodates the short range communication (i.e., several meters)~\cite{Vannucci-2008,Boyer-2014}. Recently, the BackCom system has been proposed for providing longer range communications, e.g., by installation of battery units and supporting low-bit rate communications~\cite{Vannucci-2008,Bletsas-2009}, or exploiting the bistatic architectures~\cite{6742719}. Such extended-range BackCom systems have been considered for point-to-point communication~\cite{6836141,Liu-2017,Vincent-2013,Yang-2017,Han-2017,Mudasar-2017} and one-to-many communication~\cite{Vannucci-2008,Bletsas-2009,Yang-2015,Psomas-2017,Zhu-2017}. For the \textit{point-to-point communication}, the physical layer security mechanism was developed in~\cite{6836141}, where the reader interferes with the eavesdropper by injecting a randomly generated noise signal which is added to the CW sent to the tag. In~\cite{Liu-2017}, for a BackCom system consisting of multiple reader-tag pairs, a multiple access scheme, named as time-hopping full-duplex BackCom, was proposed to avoid interference and enable full-duplex communication. Other works have considered BackCom systems with BNs powered by the ambient RF signal~\cite{Vincent-2013,Yang-2017} or power beacons~\cite{Han-2017,Mudasar-2017}. For the \textit{one-to-many communication}, a set of signal and data extraction techniques for the backscatter sensors' information was proposed in~\cite{Vannucci-2008}, where the sensors operate in different subcarrier frequencies. In~\cite{Bletsas-2009}, the authors used the beamforming and frequency-shift keying modulation to minimize the collision in a backscatter sensor network and studied the sensor collision (interference) performance. In~\cite{Yang-2015}, an energy beamforming scheme was proposed based on the backscatter-channel state information and the optimal resource allocation schemes were also obtained to maximize the total utility of harvested energy. In~\cite{Psomas-2017}, the decoding probability for a certain sensor was derived using stochastic geometry, where three collision resolution techniques (i.e., directional antennas, ultra-narrow band transmissions and successive interference cancellation (SIC)) were incorporated. For an ALOHA-type random access, by applying the machine learning to implement intelligent sensing, the work in~\cite{Zhu-2017} presented a framework of backscatter sensing with random encoding at the BNs and the statistic inference at the reader.

In this work, we focus on the uplink communication in a one-to-many BackCom system. To handle the multiple access, non-orthogonal multiple access (NOMA) is employed. By allowing multiple users to be served in the same resource block, NOMA can greatly improve the spectrum efficiency and it is also envisaged as an essential technology for 5G systems~\cite{Ding-2017}. In general, the NOMA technique can be divided into power-domain NOMA and code-domain NOMA. The code-domain NOMA utilizes user-specific spreading sequences for concurrently using the same resource, while the power-domain NOMA exploits the difference in the channel gain among users for multiplexing. The power-domain NOMA has the advantages of low latency and high spectral efficiency~\cite{Shin-2017} and it will be considered in our work. For the conventional communication system, the implementation of power-domain NOMA in the uplink communication has been well investigated in the literature, e.g.,~\cite{Imari-2014,Ding-2014,Diamantoulakis-2016,Mohammad-2017}. Very recently, the authors in~\cite{Lyu-2017} investigated NOMA in the context of a power station-powered BackCom system. To implement NOMA, the time spent on energy harvesting for each BN is different, where the optimal time allocation policy was obtained.

\textit{Paper contributions:} In this paper, we consider a single BackCom system, where one reader serves multiple randomly deployed BNs. We adopt a hybrid power-domain NOMA and time division multiple access (TDMA) to enhance the BackCom system performance. Specifically, we multiplex the BNs in different spatial regions (namely the region division approach) or with different backscattered power levels (namely the power division approach) to implement NOMA. Different from the conventional wireless devices that can actively adjust the transmit power, we set the reflection coefficients for the multiplexed BNs to be different in order to better exploit power-domain NOMA. We make the following major contributions in this paper:
\begin{itemize}
  \item We propose a NOMA-enhanced BackCom system, where the reflection coefficients for the multiplexed BNs from different groups are set to different values to utilize the power-domain NOMA. Based on the considered system model, we develop criteria for choosing the reflection coefficients for the different groups. To the best of our knowledge, such guidelines have not yet been proposed in the literature.
  \item We adopt a metric, named as the average number of successfully decoded bits (i.e., the average number of bits can be successfully decoded by the reader in one time slot), to evaluate the system performance. For the most practical case of two-node pairing, we derive the exact analytical closed-form results for the fading-free scenario and semi closed-form results for the fading scenario (cf. Table~\ref{tb:1}). For analytical tractability, under the fading-free and general multiple-node multiplexing case, we analyze a metric, the average number of successful BNs given $N$ multiplexing BNs, which has similar performance trend as the average number of successfully decoded bits. The derived expressions allow us to verify the proposed selection criteria and investigate the impact of system parameters.
  \item Our numerical results show that, NOMA generally can achieve much better performance gain in the BackCom system compared to its performance gain in the conventional system. This highlight the importance of incorporating NOMA with the BackCom system.
\end{itemize}

The remainder of the paper is organized as follows. Section~\ref{sec:system} presents the detailed system model, including the developed NOMA scheme. The proposed reflection coefficient selection criterion is presented in Section~\ref{sec:designtuition}. The definition and the analysis of the considered performance metrics for the fading-free and fading scenarios are given in Sections~\ref{sec:ana1} and~\ref{sec:fading}, respectively. Section~\ref{sec:result} presents the numerical and simulation results to study the NOMA-enhanced BackCom system. Finally, conclusions are presented in Section~\ref{sec:summary}.

\section{System Model}\label{sec:system}
\subsection{Spatial Model}\label{sec:spatialmodel}
We consider a BackCom system consisting of a single reader and $M$ BNs (sensors), as illustrated in Fig.~\ref{fig_systemmodel1}. The coverage zone $\mathcal{S}$ for the reader is assumed to be an annular region specified by the inner and outer radii $R_{1}$ and $R$, where the reader is located at the origin~\cite{Bletsas-2009,Psomas-2017}. The $M$ BNs are randomly independently and uniformly distributed inside $\mathcal{S}$, i.e., the location of BNs is modelled as the binomial point process. Consequently, the distribution of the random distance between a BN and the reader, $r$, is $f_r(r)=\frac{2r}{R^2-R_1^2}$~\cite{Zubair-2013}.
\ifCLASSOPTIONpeerreview
\begin{figure}
\centering
\subfigure[Spatial model (${\color{red}\blacktriangle}=$ reader, ${\color{blue}\bullet}=$ BNs).]{\label{fig_systemmodel1}\includegraphics[width=0.3 \textwidth]{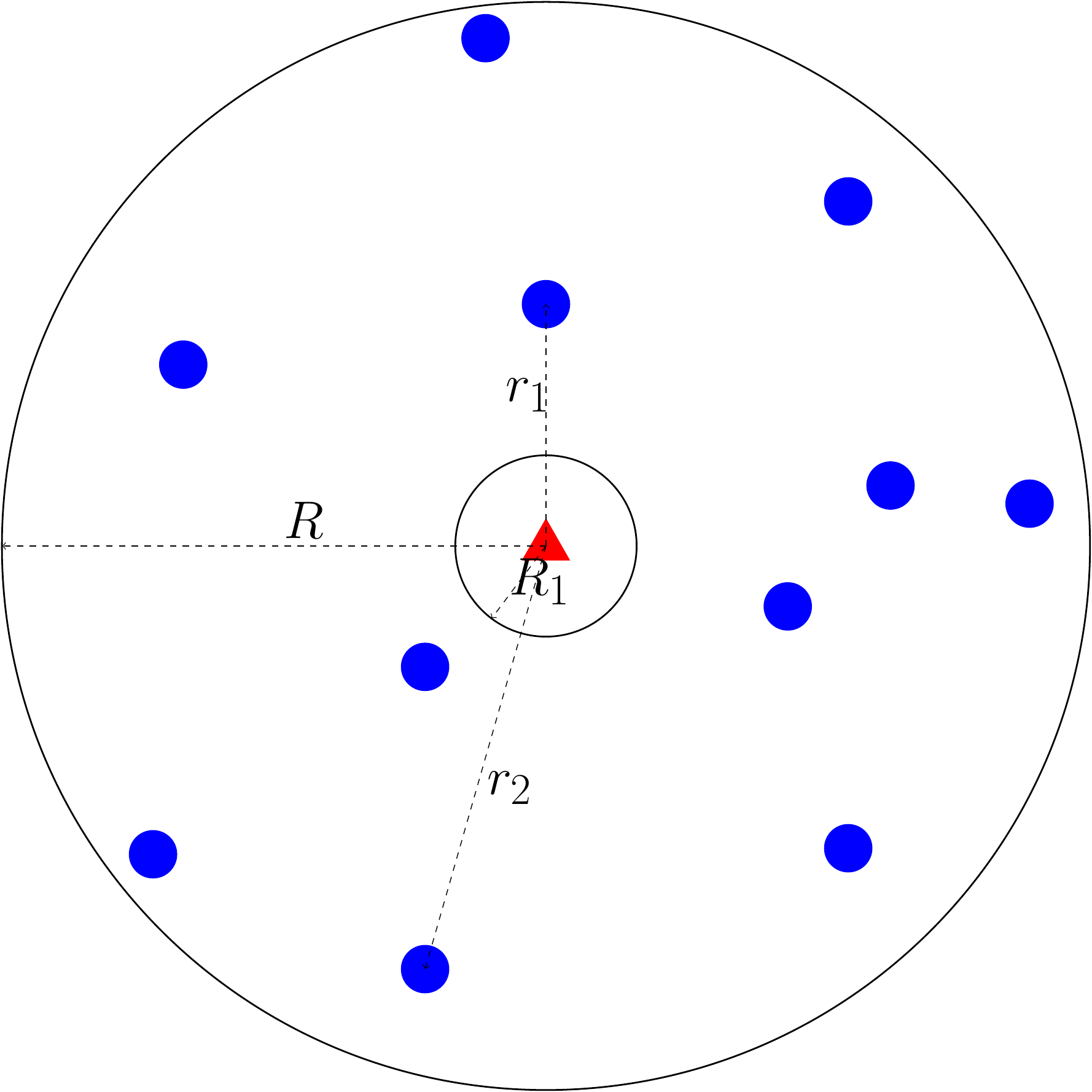}}
\mbox{\hspace{3cm}}
\subfigure[Time slot structure (${\color{green}\blacksquare}=$ mini-slot on NOMA, ${\color{black}\square}=$ mini-slot on single access).]{\label{fig_systemmodel3}\includegraphics[width=0.3\textwidth]{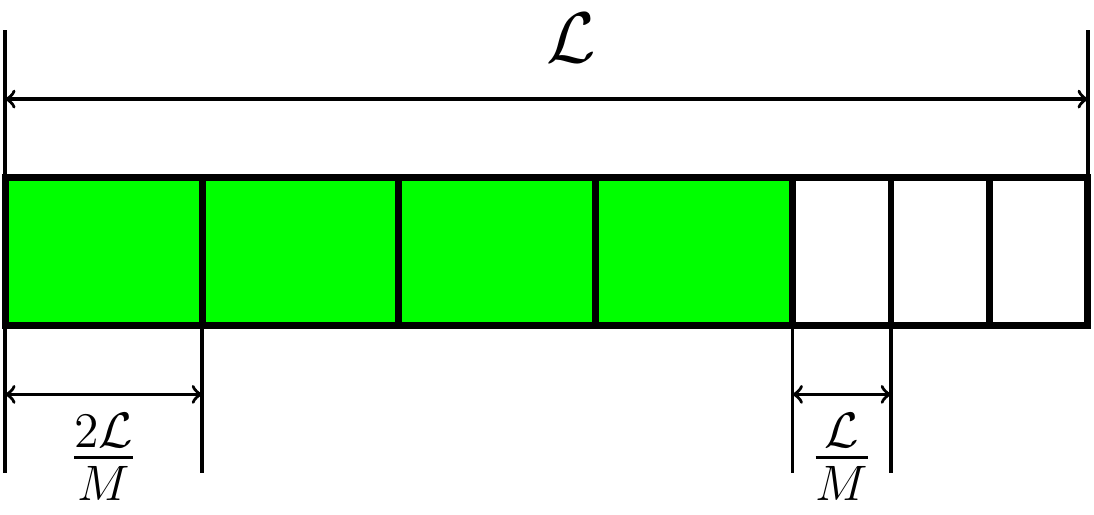}}

        \caption{Illustration of the system model for two-node pairing case.}

\end{figure}
\else
\begin{figure}
\centering
\subfigure[Spatial model (${\color{red}\blacktriangle}=$ reader, ${\color{blue}\bullet}=$ BNs).]{\label{fig_systemmodel1}\includegraphics[width=0.25 \textwidth]{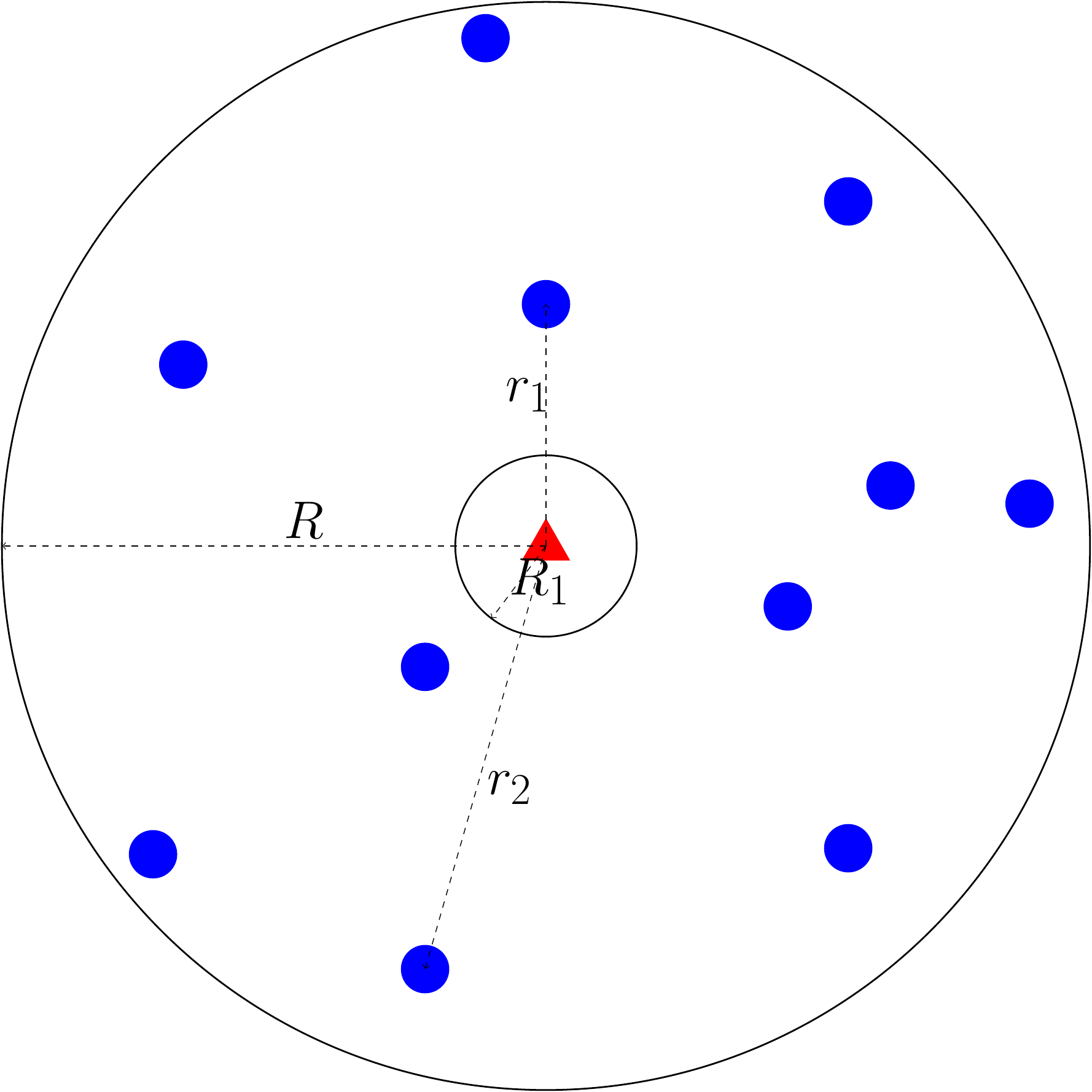}}\\
\vspace{+0.1in}
\subfigure[Time slot structure (${\color{green}\blacksquare}=$ mini-slot on NOMA, ${\color{black}\square}=$ mini-slot on single access).]{\label{fig_systemmodel3}\includegraphics[width=0.3\textwidth]{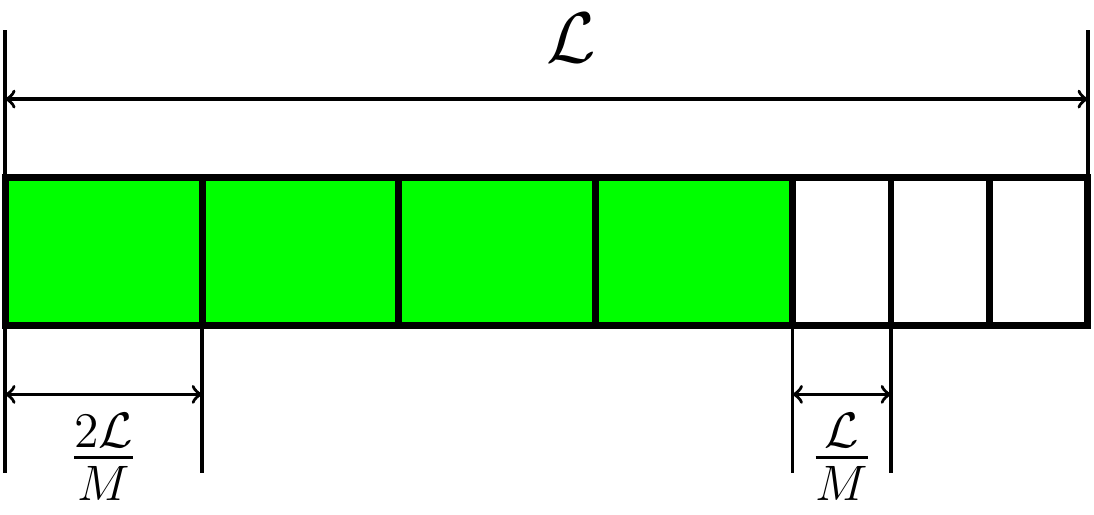}}

        \caption{Illustration of the system model for two-node pairing case.}

\end{figure}
\fi

\subsection{Channel Model}\label{sec:channelmodel}
In this work, we first consider the fading-free channel model, i.e., we use the path-loss to model the wireless communication channel. Thus, for a receiver, its received power from a transmitter is given by $p_t r^{-\alpha}$, where $p_t$ is the transmitter's transmit power, $\alpha$ is the path-loss exponent, and $r$ is the distance between the transmitter and receiver pair, respectively. This fading-free channel model is a reasonable assumption for the BackCom system with strong line-of-sight (LOS) links~\cite{Psomas-2017}. This can be justified as follows. The coverage zone for a reader is generally relatively small, especially compared to the cell's coverage region, and the BNs are close to the reader; hence, the communication link is very likely to experience strong LOS fading. In Section~\ref{sec:fading}, we will extend the system model to include the fading. Under the fading case, we assume that the fading on the communication link is identically and independently distributed (i.i.d.) Nakagami-$m$ fading. Also we will show that the design intuition gained from the fading-free scenario can provide a good guideline for LOS fading scenario. The additive white Gaussian noise (AWGN) with noise power $\mathcal{N}$ is also included in the system.

\subsection{Backscatter Communication Model}
In general, the BNs do not actively transmit any radio signal. Instead, the communication from a BN to the reader is achieved by reflecting the incident CW signal from the reader. In this work, the reader is assumed to transmit a CW signal for most of the time, while each BN has two states, namely the backscattering state and the waiting state. Fig.~\ref{fig_systemmodel2} depicts the structure of the considered BN; it is mainly composed of the transmitter, receiver, energy harvester, information decoder, micro-controller and variable impedance.

In the \emph{backscattering state}, the BN's transmitter is active and is backscattering the modulated signal via a variable impedance. We consider the binary phase shift keying modulation in this work. To modulate the signal, the in-built micro-controller switches impedances between the two impedance states. These two impedance are assumed to generate two reflection coefficients with the same magnitude (denoted as $\xi$) but with different phase shift (i.e., zero degree and 180 degree). Combining with our channel model, given that the transmit power of the reader is $P_T$, the backscattered power at a BN is $\xi P_T r^{-\alpha}$.

In the \emph{waiting state}, the BN stops backscattering and only harvests the energy from the CW signal. The harvested energy is used to power the circuit and sensing functions. We assume that each BN has a relatively large energy storage. The storage battery allows the accumulation of energy with random arrivals and the stored energy can be used to maintain the normal operation of BNs in the long run.
\ifCLASSOPTIONpeerreview
\begin{figure}
\centering
\includegraphics[width=0.9 \textwidth]{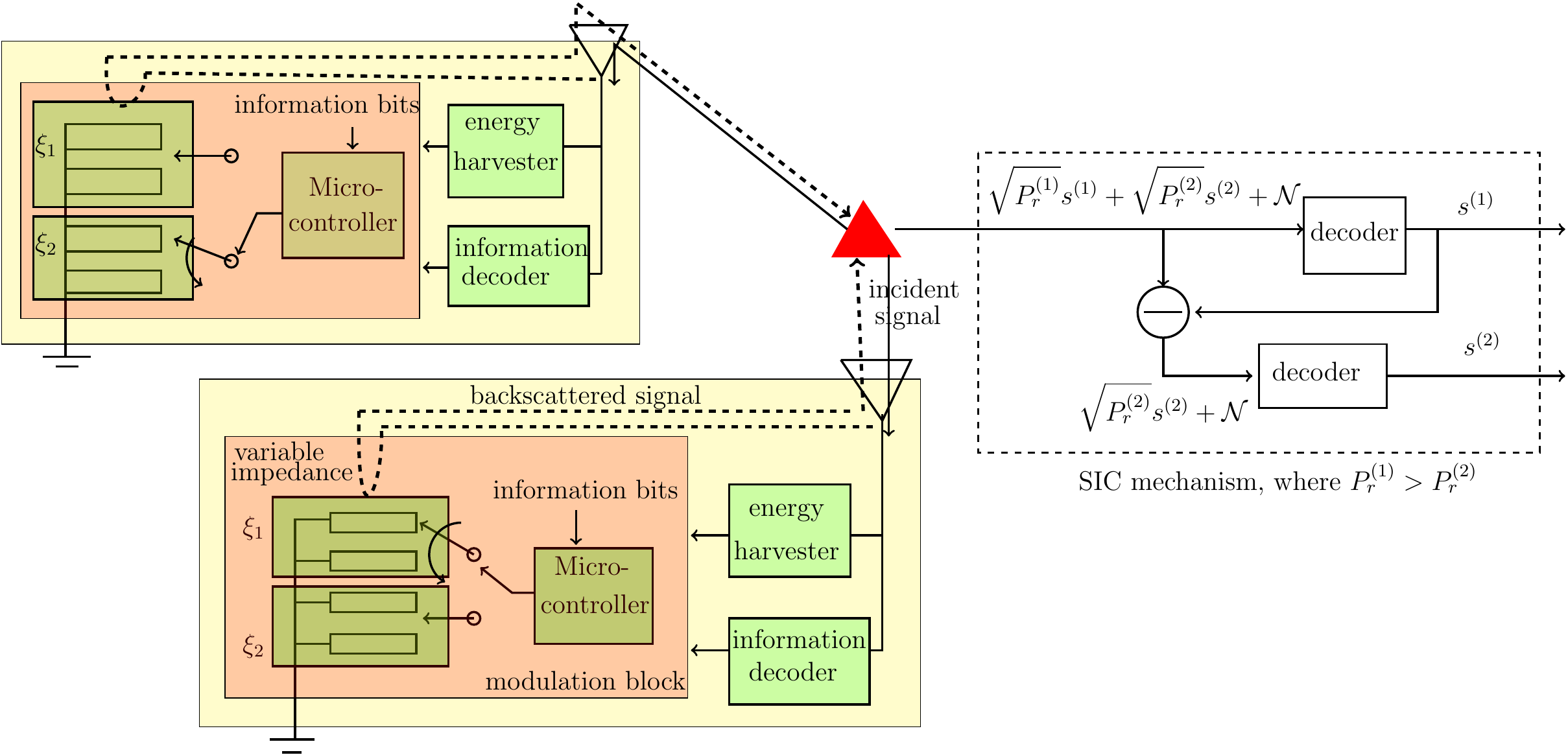}
        \caption{Illustration of the BackCom with NOMA for $N=2$ scenario. $P_r^{(1)}$ and $P_r^{(2)}$ denote the stronger signal power and the weaker signal power at the reader, respectively. $s^{(\cdot)}$ denotes the corresponding normalized information signal.}
        \label{fig_systemmodel2}
\end{figure}
\else
   \begin{figure*}[!t]
\centering
\includegraphics[width=0.8 \textwidth]{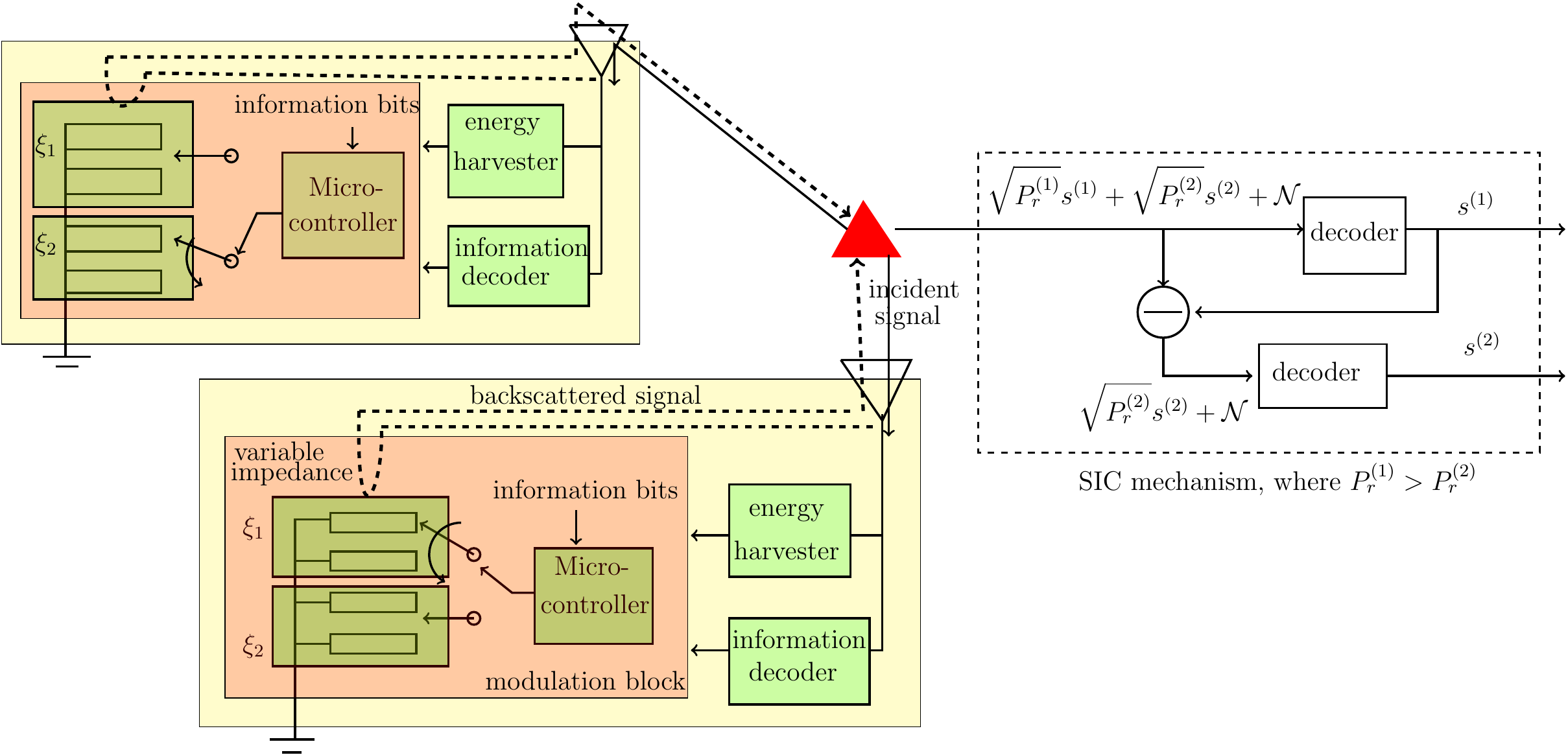}
        \caption{Illustration of the BackCom with NOMA for $N=2$ scenario. $P_r^{(1)}$ and $P_r^{(2)}$ denote the stronger signal power and the weaker signal power at the reader, respectively. $s^{(\cdot)}$ denotes the corresponding normalized information signal.}
        \label{fig_systemmodel2}
\end{figure*}
\fi

\subsection{Proposed NOMA Scheme}\label{sec:NOMAmodel}
In this section, we describe the proposed NOMA scheme for the BackCom system, which is a contribution of this work. We focus on the uplink communication and employ a hybrid of power-domain NOMA and TDMA. Each time slot lasts $\mathcal{L}$ seconds and the data rate for each BN is $\mathcal{R}$ bits/secs. Each time slot $\mathcal{L}$ is further divided into multiple mini-slots depending on the multiplexing situation, which will be explained later in Section~\ref{sec:NOMA:frame}.

\subsubsection{Region division for multiplexing}It is widely known that the fundamental principle of implementing power-domain NOMA is to multiplex (group) users with the relatively large channel gain difference on the same spectrum resource~\cite{Shin-2017}. Hence, we utilize the BNs residing in separate regions to implement power-domain NOMA, which is named as the region division approach. Specifically, the reader ``virtually'' divides the coverage zone into $N$ subregions\footnote{In this work, we mainly focus on the $N=2$ case (i.e., NOMA with two-node pairing case), which is widely considered in the literature. The analysis for the general $N$ case will be presented in Section~\ref{sec:generalcase}.} and the $i$-th subregion is an annular region specified by the radii $R_{i}$ and $R_{i+1}$, where $i\in [1,N]$, $R_{i}<R_{i+1}$ and $R_{N+1}=R$. The reader randomly picks one BN from each subregion to implement NOMA. Since the BNs are randomly deployed in $\mathcal{S}$, it is possible that the number of BNs in each subregion is not equal. For this unequal number of BNs scenario, the reader will first multiplex $N$ BNs. If the reader cannot further multiplex $N$ BNs, it will then multiplex $N-1$ BNs, $N-2$ BNs and so on and so forth.

\subsubsection{Time slot structure}\label{sec:NOMA:frame} Each time slot $\mathcal{L}$ is divided into multiple mini-slots. For the mini-slot used to multiplex $n$ BNs, the time allocated to this mini-slot is assumed to be $n\frac{\mathcal{L}}{M}$. Let us consider $N=2$ as an example and assume that there are $t$ BNs residing in the first subregion (namely the near subregion) and $M-t$ BNs in the second subregion (namely the far subregion), where $t\leq M/2$ is considered. In the first $t$ mini-slots, where each mini-slot lasts $\frac{2\mathcal{L}}{M}$ seconds, the reader will randomly select one BN from the near subregion and another BN from the far subregion to implement NOMA for each mini-slot. As for the remaining $M-2t$ BNs in the far subregion, since there are no available BNs in the near subregion to pair them, they can only communicate with the reader in a TDMA fashion in the following $M-2t$ mini-slots, i.e., each BN is allocated $\frac{\mathcal{L}}{M}$ seconds to backscatter the signal alone. Note that the BNs which are not selected by the reader to backscatter signal on a certain mini-slot are in waiting state. The time slot structure for two-node pairing case is illustrated in Fig.~\ref{fig_systemmodel3}.

\subsubsection{Reflection coefficient differentiation and its implementation} To make the difference of channel gains for multiplexing nodes more significant, we consider that the reflection coefficient for the BN belonging to different subregion is different. Let $\xi_i$ denote the reflection coefficient for the BN in the $i$-th subregion and we set $1\geq \xi_{1}\geq...\xi_{i-1}\geq\xi_{i}...\geq  \xi_{N}>0$. The reflection coefficient $\xi$ is of importance for the BackCom system with NOMA. In Section~\ref{sec:designtuition}, we will provide design guidelines on how to choose the reflection coefficient for each subregion to improve the system performance.

In order to know which BNs belong to which subregions, the following approach is adopted in this work. We assume that each BN has a unique ID, which is known by the reader~\cite{Bletsas-2009}. The reader broadcasts the training signal to all BNs and each node then backscatters this signal in its corresponding assigned slot~\cite{Yang-2015}. By receiving the backscattered signal, the reader can categorize the BNs into different subregions based on the different power levels. At the same time, each BN can decide which subregion it belongs to according to the received training signal power from the reader, and then switches its impedance pair to the corresponding subregion's impedance pair for the NOMA implementation. Note that, we assume that each BN has $N$ impedance pairs corresponding to the $N$ reflection coefficients for each subregion, from which the micro-controller can select\footnote{Note that $N$ is a pre-defined system parameter. Once $N$ is chosen, the hardware (e.g., the impedance pairs) is fixed.}. Additionally, during the training period, all BNs switch to the first impedance pair (e.g., the reflection coefficient is $\xi_1$).

\subsubsection{SIC mechanism} NOMA is carried out via the SIC technique at the reader. We assume that the decoding order is always from the strongest signal to the weakest signal\footnote{Under the fading-free scenario, the decoding order is from the nearest BN to the farthest BN. Under the fading case, the signal here implies the instantaneous backscattered signal received at the reader and the strongest signal may not come from the nearest BN.} and error propagation is also included. For example, the reader firstly detects and decodes the strongest signal, and treats the weaker signal as the interference. If the signal-to-interference-plus-noise ratio (SINR) at the reader is greater than a threshold $\gamma$, the strongest signal can be successfully decoded and extracted from the received signal. The reader then decodes the second strongest signal and so on and so forth. If the SINR is below the threshold, the strongest signal cannot be decoded and the reader will not continue to decode the weaker signals, which implies that the remaining weaker signals fail to be decoded as well~\cite{6954404}. Fig.~\ref{fig_systemmodel2} illustrates the basic structure for the SIC technique.

\section{Design Guideline for the Reflection Coefficients}\label{sec:designtuition}
For the conventional communication system implementing power-domain NOMA, the multiplexed devices transmit with different powers in order to gain the benefits from NOMA. Unfortunately, actively updating the transmit power is impossible for BNs, since they are passive devices. Instead, the reflection coefficient is an adjustable system parameter for BNs to enhance the system performance. It is intuitive to set the reflection coefficients for the near subregions as large as possible and set the reflection coefficients for the far subregions as small as possible. Then the question is how small (or large) should the reflection coefficients be for the near (or far) subregions. In this section, we provide a simple design guideline for choosing the reflection coefficients for the subregions, which is presented in the following proposition.

\begin{proposition}
Based on our system model considered in Section~\ref{sec:system}, to achieve the best system performance, the reflection coefficient for each subregion should satisfy the following conditions
\ifCLASSOPTIONpeerreview
\begin{align}
  &\xi_N\geq \gamma\frac{\mathcal{N}R^{2\alpha}}{P_T},\label{eq:designguide1}\\
  &\xi_i\geq \max\left\{\xi_{i+1},\gamma\left (\sum_{j=i+1}^{N}\xi_j\frac{R_{i+1}^{2\alpha}}{R_j^{2\alpha}}+\frac{\mathcal{N}R_{i+1}^{2\alpha}}{P_T} \right)\right\}, \quad {i\leq N-1.}\label{eq:designguide}
\end{align}
\else
\begin{align}
  &\xi_N\geq \gamma\frac{\mathcal{N}R^{2\alpha}}{P_T},\label{eq:designguide1}\\
  &\xi_i\geq \max\left\{\xi_{i+1},\gamma\left (\sum_{j=i+1}^{N}\xi_j\frac{R_{i+1}^{2\alpha}}{R_j^{2\alpha}}+\frac{\mathcal{N}R_{i+1}^{2\alpha}}{P_T} \right)\right\}, \nonumber\\
  &\quad\quad\quad\quad\quad\quad\quad\quad\quad\quad\quad\quad\quad\quad\quad\quad\quad {i\leq N-1.}\label{eq:designguide}
\end{align}
\fi

For the simplest case where $N=2$, we have $\xi_2\geq \gamma\frac{\mathcal{N}R^{2\alpha}}{P_T}$ and $\xi_1\geq \max\left\{\xi_{2},\gamma\left (\xi_2+\frac{\mathcal{N}R_{2}^{2\alpha}}{P_T} \right)\right\}$.
\end{proposition}

\begin{proof}
We consider the case of $N$ multiplexing nodes and the design guideline obtained for this scenario also holds for the case of $n$ multiplexing nodes, where $n<N$, since the decoding signal receives the most severe interference for the $N$ multiplexing nodes case. The best performance that can be achieved by the BackCom system is that the signals from all the multiplexed BNs are successfully decoded. In other words, the SINR for the $i$-th strongest signal, denoted $\textsf{SINR}_i$, is greater than the channel threshold $\gamma$, where $i\in [1,N-1]$, and the signal-to-interference (SNR) for the weakest signal, denoted as $\textsf{SNR}_N$ is also higher than $\gamma$.

Let us start from the strongest signal and its SINR is given by $\textsf{SINR}_1=\frac{P_T\xi_1r_1^{-2\alpha}}{\sum_{j=2}^{N}P_T\xi_jr_j^{-2\alpha}+\mathcal{N}}$, where $r_j$ represents the random distance between the reader and the BN from the $j$-th subregion and its conditional probability density function (PDF) is $f_{r_j}(r_j)=\frac{2r_j}{R_{j+1}^2-R_{j}^2}$ with $r_j\in[R_{j},R_{j+1}]$. In order to ensure that the strongest signal will always be successfully decoded, the worst case of $\textsf{SINR}_1$ should always be greater than $\gamma$. The worst case for $\textsf{SINR}_1$ is that $r_1=R_2$ and $r_j=R_j$; hence, we can write the condition that the strongest signal is always successfully decoded as $\frac{P_T\xi_1R_2^{-2\alpha}}{\sum_{j=2}^{N}P_T\xi_jR_j^{-2\alpha}+\mathcal{N}}\geq \gamma$. After rearranging the inequality, we obtain $\xi_1\geq \gamma\left( \sum_{j=2}^{N}\xi_j\frac{R_2^{2\alpha}}{R_j^{2\alpha}}+\frac{\mathcal{N}R_2^{2\alpha}}{P_T}\right)$. Adopting the same procedure, we can find the value of $\xi_i$ for the other signals.
\end{proof}

\begin{remark}
Under the proposed selection criterion, every BN can be successfully decoded for the fading-free scenario. Clearly, when more BNs can be multiplexed (i.e., $N$ is a relatively large value), the network performance can be greatly improved. From~\eqref{eq:designguide}, we can see that $\xi_i$ is the summation of $\xi_j$, where $j>i$, and also depends on $\gamma$. When $\gamma$ is large, the obtained $\xi_i$ can be greater than one, which is impractical. In order to meet the condition in~\eqref{eq:designguide}, we have to set $\xi_N$ as small as possible. Correspondingly, when $N$ is large, the transmit power of the reader $P_T$ should be increased in order to satisfy condition in~\eqref{eq:designguide1}. Hence, there is a tradeoff between the BackCom system performance and the reader's transmit power together with the SIC implementation complexity.
\end{remark}

\section{Analysis of the Proposed BackCom System with NOMA}\label{sec:ana1}
In this section, we present the analysis of the performance metrics for our considered BackCom system with NOMA, under the fading-free scenario.
\subsection{Performance Metrics}
The \emph{average number of successfully decoded bits}, $\bar{\mathcal{C}}_{\suc}$, is the main metric considered in this work. It is defined as the average number of bits that can be successfully decoded at the reader in one time slot. For the system where the coverage region is divided into $N$ subregions, this metric depends on: (i) the average number of successful BNs given that $n$ (where $n\in[1,N]$) BNs are multiplexed, denoted as $\bar{\mathcal{M}}_n$; and (ii) all possible multiplexing scenarios (i.e., the number of BNs in each separate subregion).

For $N=2$, we investigate the average number of successfully decoded bits, $\bar{\mathcal{C}}_{\suc}$. When $N\geq 3$, there is no general expression for $\bar{\mathcal{C}}_{\suc}$, because the second condition corresponds to the classical balls into bins problem and currently the general form listing all possible allocation cases is not possible~\cite{Raab-1998}. In this work, for $N\geq 3$ scenario, we consider the metric $\bar{\mathcal{M}}_N$, i.e., the average number of successful BNs given $N$ multiplexing nodes. As will be shown in Section~\ref{sec:result}, $\bar{\mathcal{M}}_N$ has similar trends as $\bar{\mathcal{C}}_{\suc}$ for general $N$ case.

\subsection{Two-Node Pairing Case ($N=2$)}\label{sec:nofading}
We first consider the two-node pairing case, which is widely adopted and considered in the NOMA literature due to its feasibility in practical implementation. The definition and the essential expression of $\bar{\mathcal{C}}_{\suc}$ are given below, where the factors used to calculate this metric for different scenarios are summarized in Table~\ref{tb:1} (cf. Section~\ref{sec:fading:summary}).

\begin{define}
Based on our NOMA-enhanced BackCom system in Section~\ref{sec:NOMAmodel}, the average number of successfully decoded bits $\bar{\mathcal{C}}_{\suc}$ is
\ifCLASSOPTIONpeerreview
\begin{align}\label{eq:general:C}
\bar{\mathcal{C}}_{\suc}=&\sum_{t=0}^{M/2}\binom{t}{M}p_{\near}^t(1-p_{\near})^{M-t}\left(t\frac{2\mathcal{L}\mathcal{R}}{M}\bar{\mathcal{M}}_2+(M-2t)\frac{\mathcal{L}\mathcal{R}}{M}\bar{\mathcal{M}}_{1\near} \right)\nonumber\\
&+\sum_{t=M/2+1}^{M}\binom{t}{M}p_{\near}^t(1-p_{\near})^{M-t}\left((M-t)\frac{2\mathcal{L}\mathcal{R}}{M}\bar{\mathcal{M}}_2+(2t-M)\frac{\mathcal{L}\mathcal{R}}{M}\bar{\mathcal{M}}_{1\far}\right),
\end{align}
\else
\begin{align}\label{eq:general:C}
\bar{\mathcal{C}}_{\suc}=&\sum_{t=0}^{M/2}\binom{t}{M}p_{\near}^t(1-p_{\near})^{M-t}\nonumber\\
&\times\left(t\frac{2\mathcal{L}\mathcal{R}}{M}\bar{\mathcal{M}}_2+(M-2t)\frac{\mathcal{L}\mathcal{R}}{M}\bar{\mathcal{M}}_{1\near} \right)\nonumber\\
&+\sum_{t=M/2+1}^{M}\binom{t}{M}p_{\near}^t(1-p_{\near})^{M-t}\nonumber\\
&\times\left((M-t)\frac{2\mathcal{L}\mathcal{R}}{M}\bar{\mathcal{M}}_2+(2t-M)\frac{\mathcal{L}\mathcal{R}}{M}\bar{\mathcal{M}}_{1\far}\right),
\end{align}
\fi
\noindent where $p_{\near}$ is the average probability that a BN is residing in the near subregion (i.e., the first subregion) and it equals to $p_{\near}=\frac{R_2^2-R_1^2}{R^2-R_1^2}$. $\bar{\mathcal{M}}_{1\near}$ ($\bar{\mathcal{M}}_{1\far}$) denotes the average number of successful BNs coming from the near (far) subregion, given that it accesses the reader alone. $\bar{\mathcal{M}}_2$ is the average number of successful BNs when two BNs are paired, and it can be expressed as $\bar{\mathcal{M}}_2=p_1+2p_2$, where $p_2$ is the average probability that signals for the paired BNs are successfully decoded and $p_1$ is the probability that only the stronger signal is successfully decoded.
\end{define}

The key elements that determine $\bar{\mathcal{C}}_{\suc}$ are presented in the following lemmas.
\ifCLASSOPTIONpeerreview
\begin{lemma}
Based on our system model in Section~\ref{sec:system}, given that two BNs are paired, the probability that the signals from the two BNs are successfully decoded and the probability that the signal from only one BN is successfully decoded are given by
\begin{align}\label{eq:nofading:p2}
p_2=\left\{ \begin{array}{ll}
0, \quad\quad\quad\quad\quad\quad\quad\quad\quad\quad\quad\,\,\,\,\quad\quad\quad\quad\quad\quad\quad\quad\quad\quad\quad\quad\quad\quad\quad\quad\quad\quad\quad\quad{\gamma\geq \frac{P_T\xi_2R_2^{-2\alpha}}{\mathcal{N}};}\\
\frac{\left(\textrm{max}\left\{\frac{1}{R^{2\alpha}},\frac{\mathcal{N}\gamma}{\xi_2P_T} \right\}
\right)^{-\frac{1}{\alpha}}-R_2^2}{R^2-R_2^2}, \quad\quad\quad\quad\quad\quad\quad\quad\quad\quad\quad\quad{\left(\gamma< \frac{P_T\xi_2R_2^{-2\alpha}}{\mathcal{N}}\right)\&\&\left(\gamma\leq\frac{R_2^{-2\alpha}}{\kappa R_2^{-2\alpha}+\frac{\mathcal{N}}{\xi_1P_T}} \right);}\\
0,
\quad\quad\quad\quad\quad\quad\quad\quad\quad\quad\,\,\,\quad\quad{\left(\gamma< \frac{P_T\xi_2R_2^{-2\alpha}}{\mathcal{N}}\right)\&\&\left(R_1^{-2\alpha}\leq\frac{\mathcal{N}\gamma}{\xi_1P_T}+\textrm{max}\left\{ \frac{1}{R^{2\alpha}},\frac{\mathcal{N}\gamma}{\xi_2P_T} \right\}\gamma\kappa \right);}\\
\Omega\left(\textrm{max}\left\{\frac{1}{R^{2\alpha}},\frac{R_2^{-2\alpha}}{\gamma\kappa}-\frac{\mathcal{N}}{\xi_2P_T},\frac{\mathcal{N}\gamma}{\xi_2P_T}\right\}
, \textrm{min}\left\{\frac{1}{R_2^{2\alpha}},\frac{R_1^{-2\alpha}}{\gamma\kappa}-\frac{\mathcal{N}}{\xi_2P_T}\right\} ,\textrm{max}\left\{ \frac{1}{R^{2\alpha}},\frac{\mathcal{N}\gamma}{\xi_2P_T} \right\}
\right), \quad{\textrm{otherwise};}\\
\end{array} \right.
\end{align}

\begin{align}\label{eq:nofading:p1}
p_1=\left\{ \begin{array}{ll}
0, \quad\quad\quad\quad\quad\quad\quad\quad\quad\quad\quad\,\,\,\,\quad\quad\quad\quad\quad\quad\quad\quad\quad\quad\quad\quad\quad\quad\quad\quad\quad\quad\quad\quad{\gamma\leq \frac{P_T\xi_2R^{-2\alpha}}{\mathcal{N}};}\\
\frac{R^2-\left(\textrm{min}\left\{ \frac{1}{R_2^{2\alpha}},\frac{\mathcal{N}\gamma}{\xi_2P_T} \right\}\right)^{-\frac{1}{\alpha}} }{R^2-R_2^2}, \quad\quad\quad\quad{\left(\gamma> \frac{P_T\xi_2R^{-2\alpha}}{\mathcal{N}}\right)\&\&\left(R_2^{-2\alpha}\geq\frac{\mathcal{N}\gamma}{\xi_1P_T}+\textrm{min}\left\{ \frac{1}{R_2^{2\alpha}},\frac{\mathcal{N}\gamma}{\xi_2P_T} \right\}
\gamma\kappa \right);}\\
0,
\quad\quad\quad\quad\quad\quad\quad\quad\quad\quad\quad\quad\quad\quad\quad\quad\,\,\quad\quad\quad\quad{\left(\gamma> \frac{P_T\xi_2R^{-2\alpha}}{\mathcal{N}}\right)\&\&\left(\gamma\leq\frac{R_1^{-2\alpha}}{\kappa R^{-2\alpha}+\frac{\mathcal{N}}{\xi_1P_T}} \right);}\\
\Omega\left(\textrm{max}\left\{\frac{1}{R^{2\alpha}},\frac{R_2^{-2\alpha}}{\gamma\kappa}-\frac{\mathcal{N}}{\xi_2P_T}\right\}
,
\textrm{min}\left\{\frac{1}{R_2^{2\alpha}},\frac{R_1^{-2\alpha}}{\gamma\kappa}-\frac{\mathcal{N}}{\xi_2P_T},\frac{\mathcal{N}\gamma}{\xi_2P_T}\right\}
,R^{-2\alpha}\right), \quad\quad\quad\quad\quad\,\,\,{\textrm{otherwise};}\\
\end{array} \right.
\end{align}
\noindent respectively, where $\Omega\left(p,q,w\right)\triangleq \frac{q^{-\frac{1}{\alpha}}R_1^2-\left(\frac{P_T\xi_1}{\mathcal{N}\gamma q}\right)^{\frac{1}{\alpha}}\,_2F_1\left[-\frac{1}{\alpha},\frac{1}{\alpha},\frac{\alpha-1}{\alpha},-\frac{P_T\xi_2}{\mathcal{N}}q\right] -
p^{-\frac{1}{\alpha}}R_1^2+\left(\frac{P_T\xi_1}{\mathcal{N}\gamma p}\right)^{\frac{1}{\alpha}}\,_2F_1\left[-\frac{1}{\alpha},\frac{1}{\alpha},\frac{\alpha-1}{\alpha},-\frac{P_T\xi_2}{\mathcal{N}}p\right]}{(R_2^2-R_1^2)(R^2-R_2^2)}-\frac{w^{-\frac{1}{\alpha}}-p^{-\frac{1}{\alpha}}}{R^2-R_2^2}$ and $\kappa\triangleq\frac{\xi_2}{\xi_1}$.
\end{lemma}
\else
\begin{lemma}
Based on our system model in Section~\ref{sec:system}, given that two BNs are paired, the probability that the signals from the two BNs are successfully decoded and the probability that the signal from only one BN is successfully decoded are given by~\eqref{eq:nofading:p2} and~\eqref{eq:nofading:p1}, respectively, as shown at the top of next page, where $\Omega\left(p,q,w\right)\triangleq \frac{q^{-\frac{1}{\alpha}}R_1^2-\left(\frac{P_T\xi_1}{\mathcal{N}\gamma q}\right)^{\frac{1}{\alpha}}\,_2F_1\left[-\frac{1}{\alpha},\frac{1}{\alpha},\frac{\alpha-1}{\alpha},-\frac{P_T\xi_2}{\mathcal{N}}q\right] }{(R_2^2-R_1^2)(R^2-R_2^2)}
-\frac{p^{-\frac{1}{\alpha}}R_1^2+\left(\frac{P_T\xi_1}{\mathcal{N}\gamma p}\right)^{\frac{1}{\alpha}}\,_2F_1\left[-\frac{1}{\alpha},\frac{1}{\alpha},\frac{\alpha-1}{\alpha},-\frac{P_T\xi_2}{\mathcal{N}}p\right]}{(R_2^2-R_1^2)(R^2-R_2^2)}-\frac{w^{-\frac{1}{\alpha}}-p^{-\frac{1}{\alpha}}}{R^2-R_2^2}$ and $\kappa\triangleq\frac{\xi_2}{\xi_1}$.
   \begin{figure*}[!t]
\normalsize
\begin{align}\label{eq:nofading:p2}
p_2=\left\{ \begin{array}{ll}
0, \quad\quad\quad\quad\quad\quad\quad\quad\quad\quad\quad\,\,\,\,\quad\quad\quad\quad\quad\quad\quad\quad\quad\quad\quad\quad\quad\quad\quad\quad\quad\quad\quad\quad{\gamma\geq \frac{P_T\xi_2R_2^{-2\alpha}}{\mathcal{N}};}\\
\frac{\left(\textrm{max}\left\{\frac{1}{R^{2\alpha}},\frac{\mathcal{N}\gamma}{\xi_2P_T} \right\}
\right)^{-\frac{1}{\alpha}}-R_2^2}{R^2-R_2^2}, \quad\quad\quad\quad\quad\quad\quad\quad\quad\quad\quad\quad{\left(\gamma< \frac{P_T\xi_2R_2^{-2\alpha}}{\mathcal{N}}\right)\&\&\left(\gamma\leq\frac{R_2^{-2\alpha}}{\kappa R_2^{-2\alpha}+\frac{\mathcal{N}}{\xi_1P_T}} \right);}\\
0,
\quad\quad\quad\quad\quad\quad\quad\quad\quad\quad\,\,\,\quad\quad{\left(\gamma< \frac{P_T\xi_2R_2^{-2\alpha}}{\mathcal{N}}\right)\&\&\left(R_1^{-2\alpha}\leq\frac{\mathcal{N}\gamma}{\xi_1P_T}+\textrm{max}\left\{ \frac{1}{R^{2\alpha}},\frac{\mathcal{N}\gamma}{\xi_2P_T} \right\}\gamma\kappa \right);}\\
\Omega\left(\textrm{max}\left\{\frac{1}{R^{2\alpha}},\frac{R_2^{-2\alpha}}{\gamma\kappa}-\frac{\mathcal{N}}{\xi_2P_T},\frac{\mathcal{N}\gamma}{\xi_2P_T}\right\}
, \textrm{min}\left\{\frac{1}{R_2^{2\alpha}},\frac{R_1^{-2\alpha}}{\gamma\kappa}-\frac{\mathcal{N}}{\xi_2P_T}\right\} ,\textrm{max}\left\{ \frac{1}{R^{2\alpha}},\frac{\mathcal{N}\gamma}{\xi_2P_T} \right\}
\right), \quad{\textrm{otherwise};}\\
\end{array} \right.
\end{align}

\begin{align}\label{eq:nofading:p1}
p_1=\left\{ \begin{array}{ll}
0, \quad\quad\quad\quad\quad\quad\quad\quad\quad\quad\quad\,\,\,\,\quad\quad\quad\quad\quad\quad\quad\quad\quad\quad\quad\quad\quad\quad\quad\quad\quad\quad\quad\quad{\gamma\leq \frac{P_T\xi_2R^{-2\alpha}}{\mathcal{N}};}\\
\frac{R^2-\left(\textrm{min}\left\{ \frac{1}{R_2^{2\alpha}},\frac{\mathcal{N}\gamma}{\xi_2P_T} \right\}\right)^{-\frac{1}{\alpha}} }{R^2-R_2^2}, \quad\quad\quad\quad{\left(\gamma> \frac{P_T\xi_2R^{-2\alpha}}{\mathcal{N}}\right)\&\&\left(R_2^{-2\alpha}\geq\frac{\mathcal{N}\gamma}{\xi_1P_T}+\textrm{min}\left\{ \frac{1}{R_2^{2\alpha}},\frac{\mathcal{N}\gamma}{\xi_2P_T} \right\}
\gamma\kappa \right);}\\
0,
\quad\quad\quad\quad\quad\quad\quad\quad\quad\quad\quad\quad\quad\quad\quad\quad\,\,\quad\quad\quad\quad{\left(\gamma> \frac{P_T\xi_2R^{-2\alpha}}{\mathcal{N}}\right)\&\&\left(\gamma\leq\frac{R_1^{-2\alpha}}{\kappa R^{-2\alpha}+\frac{\mathcal{N}}{\xi_1P_T}} \right);}\\
\Omega\left(\textrm{max}\left\{\frac{1}{R^{2\alpha}},\frac{R_2^{-2\alpha}}{\gamma\kappa}-\frac{\mathcal{N}}{\xi_2P_T}\right\}
,
\textrm{min}\left\{\frac{1}{R_2^{2\alpha}},\frac{R_1^{-2\alpha}}{\gamma\kappa}-\frac{\mathcal{N}}{\xi_2P_T},\frac{\mathcal{N}\gamma}{\xi_2P_T}\right\}
,R^{-2\alpha}\right), \quad\quad\quad\quad\quad\,\,\,{\textrm{otherwise}.}\\
\end{array} \right.
\end{align}
 \hrulefill
\vspace*{4pt}
\vspace{-0.05 in}
\end{figure*}
\end{lemma}
\fi
\textit{Proof:} See Appendix~A.

\begin{lemma}
Based on our system model in Section~\ref{sec:system}, given that only one BN from the near subregion accesses the reader, the average number of successful BNs is
\begin{align}\label{eq:N1case}
\bar{\mathcal{M}}_{1\near}=\frac{\min\left\{R_2^2,\left(\frac{P_T\xi_1}{\gamma\mathcal{N}} \right)^{\frac{1}{\alpha}}\right\}-R_1^2}{R_2^2-R_1^2}\textbf{1}\left(\frac{P_T\xi_1R_1^{-2\alpha}}{\mathcal{N}}\geq\gamma \right),
\end{align}
and the average number of successful BNs when only one BN from the far subregion accesses the reader is
\begin{align}\label{eq:N1case1}
\bar{\mathcal{M}}_{1\far}=\frac{\min\left\{R^2,\left(\frac{P_T\xi_2}{\gamma\mathcal{N}} \right)^{\frac{1}{\alpha}}\right\}-R_2^2}{R^2-R_2^2}\textbf{1}\left(\frac{P_T\xi_2R_2^{-2\alpha}}{\mathcal{N}}\geq\gamma \right).
\end{align}
\end{lemma}
\begin{proof}
According to the definition of $\bar{\mathcal{M}}_{1\near}$, it can be expressed as $\bar{\mathcal{M}}_{1\near}=\mathbb{E}_{r_1}\left[\Pr\left(\frac{P_T\xi_1r_1^{-2\alpha}}{\mathcal{N}}\geq\gamma\right)\right]$. After rearranging and evaluating this expression, we arrive at the result in~\eqref{eq:N1case}.
\end{proof}

\ifCLASSOPTIONpeerreview
\begin{remark}
Under the selection criterion of the reflection coefficient proposed in Proposition 1, it is clear that $\bar{\mathcal{M}}_2=2$ and $\bar{\mathcal{M}}_{1\near}=\bar{\mathcal{M}}_{1\far}=1$. Consequently, $\bar{\mathcal{C}}_{\suc}$ can be simplified into $\bar{\mathcal{C}}_{\suc}=\mathcal{L}\mathcal{R}(1+2p_{\near})+\frac{4\mathcal{L}\mathcal{R}p_{\near}^{\frac{M+2}{2}}}{M(1-p_{\near})^{\frac{4-M}{2}}}+\left((p_{\near}-1)\binom{M}{\frac{M+2}{2}}
\,_2F_1\left[1,\frac{2-M}{2},\frac{4+M}{2},\frac{p_{\near}}{p_{\near}-1}\right]-p_{\near}\binom{M}{\frac{M+4}{2}}
\,_2F_1\left[2,\frac{4-M}{2},\frac{6+M}{2},\frac{p_{\near}}{p_{\near}-1}\right] \right)$, which is the same as the total number of bits transmitted by BNs. Note that this quantity strongly relies on the radius $R_2$. The impact of $R_2$ will be presented in Section~\ref{sec:result::r2}.
\end{remark}
\else
\begin{remark}
Under the selection criterion of the reflection coefficient proposed in Proposition 1, it is clear that $\bar{\mathcal{M}}_2=2$ and $\bar{\mathcal{M}}_{1\near}=\bar{\mathcal{M}}_{1\far}=1$. Consequently, $\bar{\mathcal{C}}_{\suc}$ can be simplified into $\bar{\mathcal{C}}_{\suc}=\mathcal{L}\mathcal{R}(1+2p_{\near})+\frac{4\mathcal{L}\mathcal{R}p_{\near}^{\frac{M+2}{2}}}{M(1-p_{\near})^{\frac{4-M}{2}}}$ $+\left((p_{\near}-1)\binom{M}{\frac{M+2}{2}}
\,_2F_1\left[1,\frac{2-M}{2},\frac{4+M}{2},\frac{p_{\near}}{p_{\near}-1}\right]\right.$
 $\left.-p_{\near}\binom{M}{\frac{M+4}{2}}
\,_2F_1\left[2,\frac{4-M}{2},\frac{6+M}{2},\frac{p_{\near}}{p_{\near}-1}\right] \right)$, which is the same as the total number of bits transmitted by BNs. Note that this quantity strongly relies on the radius $R_2$. The impact of $R_2$ will be presented in Section~\ref{sec:result::r2}.
\end{remark}
\fi

Note that, when we set $\xi_2\geq\gamma\frac{\mathcal{N}R^{2\alpha}}{P_T}$, the average number of successful BNs is $\bar{\mathcal{M}}_{\suc}=2p_2$, which is directly proportional to $p_2$. The closed-form expression shown in~\eqref{eq:nofading:p2} involves the hypergeometric function coming from the noise term in the SINR, which makes it generally difficult to obtain any design intuition. By assuming that the noise is negligible, we obtain the following simplified asymptotic result for $\bar{\mathcal{M}}_{2}$ as
\ifCLASSOPTIONpeerreview
\begin{align}\label{eq:nofadingnonoise}
\lim_{\mathcal{N}\rightarrow 0}\bar{\mathcal{M}}_{2}=\left\{ \begin{array}{ll}
2, &{\gamma\kappa\leq 1 ;}\\
\frac{2R_2^2R^2+2R_1^2R_2^2-2R_1^2R^2-R_2^4\left((\gamma\kappa)^{-\frac{1}{\alpha}}+(\gamma\kappa)^{\frac{1}{\alpha}}\right)}{2(R_2^2-R_1^2)(R^2-R_2^2)}, &{1<\gamma\kappa\leq \frac{R_2^{-2\alpha}}{R^{-2\alpha}} ;}\\
\frac{2R_1^2R_2^2-2R_1^2R^2+(R^4-R_2^4)(\gamma\kappa)^{\frac{1}{\alpha}}}{2(R_2^2-R_1^2)(R^2-R_2^2)}, &{\frac{R_2^{-2\alpha}}{R^{-2\alpha}}<\gamma\kappa\leq \frac{R_1^{-2\alpha}}{R_2^{-2\alpha}} ;}\\
\frac{-2R_1^2R^2+R^4(\gamma\kappa)^{-\frac{1}{\alpha}}+R_1^4(\gamma\kappa)^{\frac{1}{\alpha}}}{2(R_2^2-R_1^2)(R^2-R_2^2)}, &{\frac{R_1^{-2\alpha}}{R_2^{-2\alpha}}<\gamma\kappa\leq \frac{R_1^{-2\alpha}}{R^{-2\alpha}} ;}\\
0, &{\gamma\kappa> \frac{R_1^{-2\alpha}}{R^{-2\alpha}}.}\\
\end{array} \right.
\end{align}
\else
\begin{align}\label{eq:nofadingnonoise}
\lim_{\mathcal{N}\rightarrow 0}\bar{\mathcal{M}}_{2}=\left\{ \begin{array}{ll}
2, \quad\quad\quad\quad\quad\quad\quad\quad\quad\quad\quad\quad\,\,\,\,\,{\gamma\kappa\leq 1 ;}\\
\frac{2R_2^2R^2+2R_1^2R_2^2-2R_1^2R^2-R_2^4\left((\gamma\kappa)^{-\frac{1}{\alpha}}+(\gamma\kappa)^{\frac{1}{\alpha}}\right)}{2(R_2^2-R_1^2)(R^2-R_2^2)},\\
 \quad\quad\quad\quad\quad\quad\quad\quad\quad\quad\,\,\,{1<\gamma\kappa\leq \frac{R_2^{-2\alpha}}{R^{-2\alpha}} ;}\\
\frac{2R_1^2R_2^2-2R_1^2R^2+(R^4-R_2^4)(\gamma\kappa)^{\frac{1}{\alpha}}}{2(R_2^2-R_1^2)(R^2-R_2^2)}, \\
\quad\quad\quad\quad\quad\quad\quad\quad\,\,\,\,{\frac{R_2^{-2\alpha}}{R^{-2\alpha}}<\gamma\kappa\leq \frac{R_1^{-2\alpha}}{R_2^{-2\alpha}} ;}\\
\frac{-2R_1^2R^2+R^4(\gamma\kappa)^{-\frac{1}{\alpha}}+R_1^4(\gamma\kappa)^{\frac{1}{\alpha}}}{2(R_2^2-R_1^2)(R^2-R_2^2)},\\
 \quad\quad\quad\quad\quad\quad\quad\quad\,\,\,\,\,{\frac{R_1^{-2\alpha}}{R_2^{-2\alpha}}<\gamma\kappa\leq \frac{R_1^{-2\alpha}}{R^{-2\alpha}} ;}\\
0, \quad\quad\quad\quad\quad\quad\quad\quad\quad\quad\quad\,\,\,\,{\gamma\kappa> \frac{R_1^{-2\alpha}}{R^{-2\alpha}}.}\\
\end{array} \right.
\end{align}
\fi
\begin{remark}
According to~\eqref{eq:nofadingnonoise}, for the given spatial and channel model, $\bar{\mathcal{M}}_{2}$ is totally determined by the ratio of reflection coefficients $\kappa$ and the threshold $\gamma$. It can be easily proved that the asymptotic result of $\bar{\mathcal{M}}_{2}$ is a monotonic decreasing function of the $\gamma$ and $\kappa$. Thus, when $\gamma\kappa\leq1$, $\bar{\mathcal{M}}_{2}$ is maximized. In other words, for the given channel threshold $\gamma$, $\kappa=\xi_2/\xi_1$ should be lower than $1/\gamma$ to optimize the network work performance, which is consistent with Proposition 1.
\end{remark}
\subsection{Multiple-Node Multiplexing Case ($N\geq$ 3)}\label{sec:generalcase}
Under this scenario, we analyze the average number of successful BNs given $N$ multiplexing BNs.

\begin{define}
Based on our NOMA-enhanced BackCom system in Section~\ref{sec:NOMAmodel}, the average number of successful BNs given $N$ multiplexing nodes, $\bar{\mathcal{M}}_N$, is given by
\begin{align}\label{eq:Ngeneral}
 \bar{\mathcal{M}}_N=\sum_{k=0}^{N}kp_k,
\end{align}
\noindent where $p_k$ is the probability that only the signals from the first $k$ BNs are successfully decoded.
\end{define}

The derivation of probability $p_k$ is very challenging. This is because the event that the signal from a BN in $i$-th subregion is successfully decoded is correlated with the event that the signal from the BN in the $i+1$-th subregion is unsuccessfully decoded. Thus, for analytical tractability, similar to most literatures~\cite{6954404,Psomas-2017}, we assume that each decoding step in the SIC is independent. We will show in Section~\ref{sec:result} that the independence assumption does not adversely affect the accuracy of the analysis. Based on this independence assumption, we can approximately express $p_k$ as
\begin{align}\label{eq:pkgeneral}
p_k\approx p_{\out}^{(k+1)}\prod_{i=1}^{k}\left(1-p_{\out}^{(i)}\right),
\end{align}
\noindent where $p_{\out}^{(i)}$ denotes the probability when the SINR of the $i$-th strongest signal (e.g., the signal from the BN in the $i$-th subregion) falls below $\gamma$ given that the $i-1$-th strongest signal is successfully decoded. Note that except $p_{\out}^{(1)}$, any $p_{\out}^{(i)}$ is the conditional outage probability.

Different from the previous two-node pairing scenario, there is no direct way to compute $p_{\out}^{(i)}$ when $N> 3$. Instead, we adopt the moment generating function (MGF)-based approach in~\cite{Guo-2013} to work out $p_{\out}^{(i)}$, which is presented in the following lemma.

\begin{lemma}
Based on our system model considered in Section~\ref{sec:system}, the probability that the signal from the $i$-th BN fails to be decoded, given that the $i-1$-th strongest signal is successfully decoded, is
\ifCLASSOPTIONpeerreview
\begin{align}\label{eq:general:mgf}
p_{\out}^{(i)}
&=1-\frac{2^{-\mathcal{B}}\exp(\frac{\mathcal{A}}{2})}{\gamma^{-1}}\sum_{b=0}^{\mathcal{B}}\binom{\mathcal{B}}{b} \sum_{c=0}^{\mathcal{C}+b}\frac{(-1)^c}{\mathcal{D}_c}\mathrm{Re} \left\{\frac{\mathcal{M}_{w_{i}}\left(s\right)}{s} \right\},
\end{align}
\noindent where $\mathcal{M}_{w_{i}}\left(s\right)\!=\!\!\!\mathlarger{\int}_{R_{i}}^{R_{i+1}}\exp\left(\frac{-sr_i^{2\alpha}\mathcal{N}}{P_T\xi_i}\right)\prod_{j=i+1}^{N}\frac{\left(sr_i^{2\alpha}\frac{\xi_{j}}{\xi_i}\right)^{\frac{1}{\alpha}}
\Gamma\left[-\frac{1}{\alpha},sr_i^{2\alpha}\frac{\xi_{j}}{\xi_iR_{j+1}^{2\alpha}},sr_i^{2\alpha}\frac{\xi_{j}}{\xi_iR_{j}^{2\alpha}}\right]}{\alpha(R_{j+1}^2-R_{j}^2)}\frac{2r_i}{R_{i+1}^2-R_{i}^2}\textup{d}r_i$. $\mathcal{D}_c= 2$ (if $c=0$) and $\mathcal{D}_c=1$ (if $c=1,2,\hdots$), $s=(\mathcal{A}+\mathbf{i}2\pi c)/(2\gamma^{-1})$. $\mathcal{A}$, $\mathcal{B}$ and $\mathcal{C}$ are three parameters employed to control the error estimation and following~\cite{Guo-2013}, we set $\mathcal{A} = 8 \ln 10$, $\mathcal{B} = 11$, $\mathcal{C} = 14$ in this work.
\else
\begin{align}\label{eq:general:mgf}
p_{\out}^{(i)}
&=1\!-\!\frac{2^{-\mathcal{B}}\exp(\frac{\mathcal{A}}{2})}{\gamma^{-1}}\sum_{b=0}^{\mathcal{B}}\binom{\mathcal{B}}{b} \sum_{c=0}^{\mathcal{C}+b}\frac{(-1)^c}{\mathcal{D}_c}\mathrm{Re} \left\{\!\frac{\mathcal{M}_{w_{i}}\left(s\right)}{s}\! \right\},
\end{align}
\noindent where $\mathcal{M}_{w_{i}}\left(s\right)\!=\mathlarger{\int}_{R_{i}}^{R_{i+1}}\exp\left(\frac{-sr_i^{2\alpha}\mathcal{N}}{P_T\xi_i}\right)$
$\times\prod_{j=i+1}^{N}\frac{\left(sr_i^{2\alpha}\frac{\xi_{j}}{\xi_i}\right)^{\frac{1}{\alpha}}
\Gamma\left[-\frac{1}{\alpha},sr_i^{2\alpha}\frac{\xi_{j}}{\xi_iR_{j+1}^{2\alpha}},sr_i^{2\alpha}\frac{\xi_{j}}{\xi_iR_{j}^{2\alpha}}\right]}{\alpha(R_{j+1}^2-R_{j}^2)}\frac{2r_i}{R_{i+1}^2-R_{i}^2}\textup{d}r_i$. $\mathcal{D}_c= 2$ (if $c=0$) and $\mathcal{D}_c=1$ (if $c=1,2,\hdots$), $s=(\mathcal{A}+\mathbf{i}2\pi c)/(2\gamma^{-1})$. $\mathcal{A}$, $\mathcal{B}$ and $\mathcal{C}$ are three parameters employed to control the error estimation and following~\cite{Guo-2013}, we set $\mathcal{A} = 8 \ln 10$, $\mathcal{B} = 11$, $\mathcal{C} = 14$ in this work.
\fi
\end{lemma}
\begin{proof}
Based on the MGF-approach, we can express $p_{\out}^{(i)}$ as
\ifCLASSOPTIONpeerreview
\begin{align}
p_{\out}^{(i)}&=\Pr\left(\textsf{SINR}_i<\gamma\right)=\Pr\left(\frac{P_T\xi_i r_i^{-2\alpha}}{\sum_{j=i+1}^{N}P_T\xi_j r_j^{-2\alpha}+\mathcal{N}}<\gamma\right)\nonumber\\
&=1-\frac{2^{-\mathcal{B}}\exp(\frac{\mathcal{A}}{2})}{\gamma^{-1}}\sum_{b=0}^{\mathcal{B}}\binom{\mathcal{B}}{b} \sum_{c=0}^{\mathcal{C}+b}\frac{(-1)^c}{\mathcal{D}_c}\mathrm{Re} \left\{\frac{\mathcal{M}_{w_{i}}\left(s\right)}{s} \right\},
\end{align}
\else
\begin{align}
p_{\out}^{(i)}&=\Pr\!\left(\textsf{SINR}_i<\gamma\right)=\Pr\!\left(\frac{P_T\xi_i r_i^{-2\alpha}}{\sum_{j=i+1}^{N}P_T\xi_j r_j^{-2\alpha}\!+\!\mathcal{N}}<\!\gamma\!\right)\nonumber\\
&=1\!-\!\frac{2^{-\mathcal{B}}\exp(\frac{\mathcal{A}}{2})}{\gamma^{-1}}\sum_{b=0}^{\mathcal{B}}\binom{\mathcal{B}}{b} \sum_{c=0}^{\mathcal{C}+b}\frac{(-1)^c}{\mathcal{D}_c}\mathrm{Re} \left\{\!\frac{\mathcal{M}_{w_{i}}\left(s\right)}{s} \!\right\},
\end{align}
\fi
\noindent where $w_{i}$ is the inverse $\textsf{SINR}_{i}$ and $\mathcal{M}_{w_{i}}\left(s\right)$ is its distribution's MGF.
\ifCLASSOPTIONpeerreview
Following the definition of MGF, we then can express the MGF of the distribution of $w_i$ as
\begin{align}\label{eq:mgfderi}
&\mathcal{M}_{w_i}\left(s\right)=\mathbb{E}_{r_i,r_{i+1},...,r_N}\left[\exp\left(-\frac{sr_i^{2\alpha}\mathcal{N}}{P_T\xi_i}-\sum_{j=i+1}^{N}sr_i^{2\alpha}r_{j}^{-\alpha}\frac{\xi_j}{\xi_i}\right)\right] \nonumber\\
&=\mathbb{E}_{r_i}\left[\exp\left(-\frac{sr_i^{2\alpha}\mathcal{N}}{P_T\xi_i}\right)\prod_{j=i+1}^{N}\int_{R_{j}}^{R_{j+1}}
\exp\left(-sr_i^{2\alpha}r_{j}^{-\alpha}\frac{\xi_j}{\xi_i}\right)\frac{2r_j}{R_{j+1}^2-R_j^2}\textup{d}r_j \right] \nonumber\\
&=\mathbb{E}_{r_i}\left[\exp\left(-\frac{sr_i^{2\alpha}\mathcal{N}}{P_T\xi_i}\right)\prod_{j=i+1}^{N}\frac{\left(sr_i^{2\alpha}\frac{\xi_{j}}{\xi_i}\right)^{\frac{1}{\alpha}}
\Gamma\left[-\frac{1}{\alpha},sr_i^{2\alpha}\frac{\xi_{j}}{\xi_iR_{j+1}^{2\alpha}},sr_i^{2\alpha}\frac{\xi_{j}}{\xi_iR_{j}^{2\alpha}}\right]}{\alpha(R_{j+1}^2-R_{j}^2)} \right].
\end{align}
\else
Following the definition of MGF, we then can express the MGF of the distribution of $w_i$ as in~\eqref{eq:mgfderi}, as shown at the top of next page.
   \begin{figure*}[!t]
\normalsize
\begin{align}\label{eq:mgfderi}
&\mathcal{M}_{w_i}\left(s\right)=\mathbb{E}_{r_i,r_{i+1},...,r_N}\!\left[\exp\!\!\left(-\frac{sr_i^{2\alpha}\mathcal{N}}{P_T\xi_i}\!-\!\sum_{j=i+1}^{N}sr_i^{2\alpha}\frac{\xi_j}{\xi_ir_{j}^{\alpha}}\!\right)\!\right] \!=\!\mathbb{E}_{r_i}\!\left[\exp\!\left(-\frac{sr_i^{2\alpha}\mathcal{N}}{P_T\xi_i}\right)\!\!\prod_{j=i+1}^{N}\!\int_{R_{j}}^{R_{j+1}}
\!\frac{\exp\!\left(-sr_i^{2\alpha}\frac{\xi_j}{\xi_ir_{j}^{\alpha}}\right)2r_j}{R_{j+1}^2-R_j^2}\textup{d}r_j \right] \nonumber\\
&=\mathbb{E}_{r_i}\left[\exp\left(-\frac{sr_i^{2\alpha}\mathcal{N}}{P_T\xi_i}\right)\prod_{j=i+1}^{N}\frac{\left(sr_i^{2\alpha}\frac{\xi_{j}}{\xi_i}\right)^{\frac{1}{\alpha}}
\Gamma\left[-\frac{1}{\alpha},sr_i^{2\alpha}\frac{\xi_{j}}{\xi_iR_{j+1}^{2\alpha}},sr_i^{2\alpha}\frac{\xi_{j}}{\xi_iR_{j}^{2\alpha}}\right]}{\alpha(R_{j+1}^2-R_{j}^2)} \right].
\end{align}
 \hrulefill
\vspace*{4pt}
\vspace{-0.05 in}
\end{figure*}
\fi

\end{proof}


\section{Two-Node Pairing Case with Fading}\label{sec:fading}
 In this section, we consider the analysis for two-node pairing by taking fading channel model into account. The block fading is considered, which indicates that the fading coefficient is unchanged within one time slot, but it may vary independently from one time slot to another time slot. The fading on the communication link is assumed to be the i.i.d. Nakagami-$m$ fading and let $g$ denote the fading power gain on the communication link that follows gamma distribution. Moreover, we assume that the downlink and uplink channels are reciprocal (i.e., the fading coefficient of the downlink channel is the transpose conjugate of the uplink channel and their fading amplitudes are the same)~\cite{Liu-2017}.

When fading is included, another type of randomness is added to the received signal. In this section, we consider two pairing approaches for power-domain NOMA, which are named as the region division and power division, respectively. Under the region division approach, the situation is similar to Section~\ref{sec:nofading}, where the reader pairs the BNs from the near subregion and far subregion. In the fading context, this approach requires the long term training to recognize the BNs either in the near subregion or in the far subregion. In terms of the power division approach, the reader pairs the BNs with the higher instantaneous backscattered power and the lower instantaneous backscattered power, which requires the instantaneous training to classify the BNs. Its explicit implementation will be explained in Section~\ref{sec:powerdivision}. Note these two approaches converge to the same one for the fading-free scenario.

 The average number of successfully decoded bits, $\bar{\mathcal{C}}_{\suc}$, is the metric investigated under the fading case. Its general expression is the same as~\eqref{eq:general:C} in Definition 1 for both approaches, while the key factors, such as $p_{\near}$, $p_1$, $p_2$, $\bar{\mathcal{M}}_{1\near}$ and $\bar{\mathcal{M}}_{1\far}$, are changed. The analysis for these factors are presented as follows and the summary is presented in Table~\ref{tb:1} (cf. Section~\ref{sec:fading:summary}).
\subsection{Region Division Approach}\label{sec:regiondevision}
For the region division approach, $p_{\near}$ is the probability that the BN is located in the near subregion, which is the same as Section~\ref{sec:nofading} and is equal to $p_{\near}=\frac{R_2^2-R_1^2}{R^2-R_1^2}$. The analysis of $p_k$ (i.e., the probability that the signals from $k$ BNs are successfully decoded) becomes complicated due to the consideration of fading.

Note that our considered SIC scheme is based on the instantaneous received power at the reader, i.e., $P_T\xi g^{2}r^{-2\alpha}$. Under the region division approach, the stronger signal may not come from the BN in the near subregion. Before deriving $p_k$, we first present the following lemma which shows the composite distribution of the random distance and fading.
\begin{lemma}\label{lemma:composite}
Let $r$ denote a random variable following the distribution of $f_r(r)=\frac{2r}{R_{u}^2-R_{l}^2}$, where $r\in[R_{l},R_{u}]$, and $g$ is a random variable following the gamma distribution, i.e., $f_g(g)=\frac{m^m g^{m-1}\exp(-m g)}{\Gamma[m]}$. The cumulative distribution function (CDF) and PDF for the composite random variable $x\triangleq g^{2}r^{-2\alpha}$ are given by
\ifCLASSOPTIONpeerreview
\begin{align}
 \Phi(x,R_l,R_u)&=1-\frac{R_u^2\Gamma\!\left[m,mR_u^{\alpha}\sqrt{x} \right]-R_l^2\Gamma\!\left[m,mR_l^{\alpha}\sqrt{x}\right]+\left(m\sqrt{x}\right)^{-\frac{2}{\alpha}}\Gamma\!\left[m+\frac{2}{\alpha},mR_l^{\alpha}\sqrt{x},mR_u^{\alpha}\sqrt{x}\right]}{\left(R_{u}^2-R_{l}^2\right)\Gamma[m]},\label{eq:cdf1} \\
  \phi(x,R_l,R_u)&=\frac{\Gamma\!\left[m+\frac{2}{\alpha},mR_l^{\alpha}\sqrt{x},mR_u^{\alpha}\sqrt{x}\right]}{m^{\frac{2}{\alpha}}x^{\frac{1}{\alpha}+1}\alpha\left(R_{u}^2-R_{l}^2\right)\Gamma[m]},
 \end{align}
 \else
 \begin{align}
 &\Phi(x,R_l,R_u)=1-\frac{R_u^2\Gamma\!\left[m,mR_u^{\alpha}\sqrt{x} \right]\!-\!R_l^2\Gamma\!\left[m,mR_l^{\alpha}\sqrt{x}\right]}{\left(R_{u}^2-R_{l}^2\right)\Gamma[m]}\nonumber\\
 &\quad\quad\quad-\frac{\left(m\sqrt{x}\right)^{-\frac{2}{\alpha}}\Gamma\!\left[m+\frac{2}{\alpha},mR_l^{\alpha}\sqrt{x},mR_u^{\alpha}\sqrt{x}\right]}{\left(R_{u}^2-R_{l}^2\right)\Gamma[m]}, \label{eq:cdf1} \\
  &\phi(x,R_l,R_u)=\frac{\Gamma\!\left[m+\frac{2}{\alpha},mR_l^{\alpha}\sqrt{x},mR_u^{\alpha}\sqrt{x}\right]}{m^{\frac{2}{\alpha}}x^{\frac{1}{\alpha}+1}\alpha\left(R_{u}^2-R_{l}^2\right)\Gamma[m]},
 \end{align}
 \fi
\noindent respectively.
\end{lemma}
\begin{proof}
The CDF of $x$ can be written as
\ifCLASSOPTIONpeerreview
\begin{align}
\Phi(x,R_l,R_u)&=\Pr\left(g^{2}r^{-2\alpha}< x\right) =\mathbb{E}_{r}\left\{\Pr\left(g< r^{\alpha}\sqrt{x}\right)\right\}\nonumber\\
&=\int_{R_l}^{R_u}\frac{\Gamma\left[m,0,m\sqrt{x}r^{\alpha}\right]}{\Gamma[m]}\frac{2r}{R_{u}^2-R_{l}^2}\, \textup{d}r\nonumber\\
&=1-\frac{R_u^2\Gamma\!\left[m,mR_u^{\alpha}\sqrt{x} \right]-R_l^2\Gamma\!\left[m,mR_l^{\alpha}\sqrt{x}\right]+\left(m\sqrt{x}\right)^{-\frac{2}{\alpha}}\Gamma\!\left[m+\frac{2}{\alpha},mR_l^{\alpha}\sqrt{x},mR_u^{\alpha}\sqrt{x}\right]}{\left(R_{u}^2-R_{l}^2\right)\Gamma[m]}.
\end{align}
\else
\begin{align}
&\Phi(x,R_l,R_u)=\Pr\left(g^{2}r^{-2\alpha}< x\right) =\mathbb{E}_{r}\left\{\Pr\left(g< r^{\alpha}\sqrt{x}\right)\right\}\nonumber\\
&\quad\quad=\int_{R_l}^{R_u}\frac{\Gamma\left[m,0,m\sqrt{x}r^{\alpha}\right]}{\Gamma[m]}\frac{2r}{R_{u}^2-R_{l}^2}\, \textup{d}r\nonumber\\
&\quad\quad=1-\frac{R_u^2\Gamma\!\left[m,mR_u^{\alpha}\sqrt{x} \right]-R_l^2\Gamma\!\left[m,mR_l^{\alpha}\sqrt{x}\right]}{\left(R_{u}^2-R_{l}^2\right)\Gamma[m]}\nonumber\\
 &\quad\quad\quad-\frac{\left(m\sqrt{x}\right)^{-\frac{2}{\alpha}}\Gamma\!\left[m+\frac{2}{\alpha},mR_l^{\alpha}\sqrt{x},mR_u^{\alpha}\sqrt{x}\right]}{\left(R_{u}^2-R_{l}^2\right)\Gamma[m]}.
\end{align}
\fi
Taking the derivative of $\Phi(x,R_l,R_u)$ with respect to $x$, we obtain its PDF.
\end{proof}

According to Lemma~\ref{lemma:composite} and probability theory, the key elements for the region division approach with fading are shown in the following lemmas.
\ifCLASSOPTIONpeerreview
\begin{lemma}
Based on our system model considered in Sections~\ref{sec:system} and~\ref{sec:fading}, under the fading scenario with region division approach, the probability that the signals from two BNs are successfully decoded and the probability that the signal from only one BN is successfully decoded are given by
\begin{align}\label{eq:general:p2}
p_2=\left\{ \begin{array}{ll}
       \mathop{\mathlarger{\int}}_{\!\!\frac{\mathcal{N}\gamma}{P_T\xi_2(1-\gamma)}}^{\infty}\!\left(1-\Phi^A\!\left(\kappa x_2\right)\right)\phi^B(x_2)\textup{d}x_2+ \mathop{\mathlarger{\int}}_{\!\!\frac{\mathcal{N}\gamma}{P_T\xi_2}}^{\frac{\mathcal{N}\gamma}{P_T\xi_2(1-\gamma)}}\!\left(1-\Phi^A\!\left(\gamma\kappa x_2\!+\!\frac{\mathcal{N}\gamma}{P_T\xi_1}\right)\!\right)\phi^B(x_2)\textup{d}x_2 \\
       +\mathop{\mathlarger{\int}}_{\!\!\frac{\mathcal{N}\gamma}{P_T\xi_1(1-\gamma)}}^{\infty}\!\left(1-\Phi^B\!\left(\frac{x_1}{\kappa}\right)\right)\phi^A(x_1)\textup{d}x_1+ \mathop{\mathlarger{\int}}_{\!\!\frac{\mathcal{N}\gamma}{P_T\xi_1}}^{\frac{\mathcal{N}\gamma}{P_T\xi_1(1-\gamma)}}\!\left(1-\Phi^B\!\left(\frac{\gamma x_1}{\kappa}\!+\!\frac{\mathcal{N}\gamma}{P_T\xi_2}\right)\!\right)\phi^A(x_1)\textup{d}x_1, &{\gamma<1 ;}\\
        \mathop{\mathlarger{\int}}_{\!\!\frac{\mathcal{N}\gamma}{P_T\xi_2}}^{\infty}\!\left(1-\Phi^A\!\left(\gamma\kappa x_2\!+\!\frac{\mathcal{N}\gamma}{P_T\xi_1}\right)\!\right)\phi^B(x_2)\textup{d}x_2 + \mathop{\mathlarger{\int}}_{\!\!\frac{\mathcal{N}\gamma}{P_T\xi_1}}^{\infty}\!\left(1-\Phi^B\!\left(\frac{\gamma x_1}{\kappa}\!+\!\frac{\mathcal{N}\gamma}{P_T\xi_2}\right)\!\right)\phi^A(x_1)\textup{d}x_1, &{\gamma \geq 1 ;}\\
                    \end{array} \right.
\end{align}
\begin{align}\label{eq:general:p1}
p_1&= \mathlarger{\int}_0^{\!\!\frac{\mathcal{N}\gamma}{P_T\xi_2}}\!\left(1-\Phi^A\!\left(\gamma\kappa x_2\!+\!\frac{\mathcal{N}\gamma}{P_T\xi_1}\right)\!\right)\phi^B(x_2)\textup{d}x_2 + \mathlarger{\int}_0^{\!\!\frac{\mathcal{N}\gamma}{P_T\xi_1}}\!\left(1-\Phi^B\!\left(\frac{\gamma x_1}{\kappa}\!+\!\frac{\mathcal{N}\gamma}{P_T\xi_2}\right)\!\right)\phi^A(x_1)\textup{d}x_1,
\end{align}
\noindent respectively, where $\kappa=\xi_2/\xi_1$, $\Phi^A(x_1)\triangleq\Phi(x_1,R_1,R_2)$, $\phi^A(x_1)\triangleq\phi(x_1,R_1,R_2)$, $\Phi^B(x_2)\triangleq\Phi(x_2,R_2,R)$ and $\phi^B(x_2)\triangleq\phi(x_2,R_2,R)$. $\Phi(\cdot,\cdot,\cdot)$ and $\phi(\cdot,\cdot,\cdot)$ are defined in Lemma~\ref{lemma:composite}.
\end{lemma}
\else
\begin{lemma}
   \begin{figure*}[!t]
\normalsize
\begin{align}\label{eq:general:p2}
p_2=\left\{ \begin{array}{ll}
       \mathop{\mathlarger{\int}}_{\!\!\frac{\mathcal{N}\gamma}{P_T\xi_2(1-\gamma)}}^{\infty}\!\left(1-\Phi^A\!\left(\kappa x_2\right)\right)\phi^B(x_2)\textup{d}x_2+ \mathop{\mathlarger{\int}}_{\!\!\frac{\mathcal{N}\gamma}{P_T\xi_2}}^{\frac{\mathcal{N}\gamma}{P_T\xi_2(1-\gamma)}}\!\left(1-\Phi^A\!\left(\gamma\kappa x_2\!+\!\frac{\mathcal{N}\gamma}{P_T\xi_1}\right)\!\right)\phi^B(x_2)\textup{d}x_2 \\
       +\mathop{\mathlarger{\int}}_{\!\!\frac{\mathcal{N}\gamma}{P_T\xi_1(1-\gamma)}}^{\infty}\!\left(1-\Phi^B\!\left(\frac{x_1}{\kappa}\right)\right)\phi^A(x_1)\textup{d}x_1+ \mathop{\mathlarger{\int}}_{\!\!\frac{\mathcal{N}\gamma}{P_T\xi_1}}^{\frac{\mathcal{N}\gamma}{P_T\xi_1(1-\gamma)}}\!\left(1-\Phi^B\!\left(\frac{\gamma x_1}{\kappa}\!+\!\frac{\mathcal{N}\gamma}{P_T\xi_2}\right)\!\right)\phi^A(x_1)\textup{d}x_1, &{\gamma<1 ;}\\
        \mathop{\mathlarger{\int}}_{\!\!\frac{\mathcal{N}\gamma}{P_T\xi_2}}^{\infty}\!\left(1-\Phi^A\!\left(\gamma\kappa x_2\!+\!\frac{\mathcal{N}\gamma}{P_T\xi_1}\right)\!\right)\phi^B(x_2)\textup{d}x_2 + \mathop{\mathlarger{\int}}_{\!\!\frac{\mathcal{N}\gamma}{P_T\xi_1}}^{\infty}\!\left(1-\Phi^B\!\left(\frac{\gamma x_1}{\kappa}\!+\!\frac{\mathcal{N}\gamma}{P_T\xi_2}\right)\!\right)\phi^A(x_1)\textup{d}x_1, &{\gamma \geq 1 ;}\\
                    \end{array} \right.
\end{align}
\begin{align}\label{eq:general:p1}
p_1&= \mathlarger{\int}_0^{\!\!\frac{\mathcal{N}\gamma}{P_T\xi_2}}\!\left(1-\Phi^A\!\left(\gamma\kappa x_2\!+\!\frac{\mathcal{N}\gamma}{P_T\xi_1}\right)\!\right)\phi^B(x_2)\textup{d}x_2 + \mathlarger{\int}_0^{\!\!\frac{\mathcal{N}\gamma}{P_T\xi_1}}\!\left(1-\Phi^B\!\left(\frac{\gamma x_1}{\kappa}\!+\!\frac{\mathcal{N}\gamma}{P_T\xi_2}\right)\!\right)\phi^A(x_1)\textup{d}x_1.
\end{align}
 \hrulefill
\vspace*{4pt}
\vspace{-0.05 in}
\end{figure*}
Based on our system model considered in Sections~\ref{sec:system} and~\ref{sec:fading}, under the fading scenario with region division approach, the probability that the signals from two BNs are successfully decoded and the probability that the signal from only one BN is successfully decoded are given by~\eqref{eq:general:p2} and~\eqref{eq:general:p1}, respectively, at the top of this page, where $\kappa=\xi_2/\xi_1$, $\Phi^A(x_1)\triangleq\Phi(x_1,R_1,R_2)$, $\phi^A(x_1)\triangleq\phi(x_1,R_1,R_2)$, $\Phi^B(x_2)\triangleq\Phi(x_2,R_2,R)$ and $\phi^B(x_2)\triangleq\phi(x_2,R_2,R)$. $\Phi(\cdot,\cdot,\cdot)$ and $\phi(\cdot,\cdot,\cdot)$ are defined in Lemma~\ref{lemma:composite}.
\end{lemma}
\fi
\begin{proof}
 \ifCLASSOPTIONpeerreview
  In order to ensure that the signals from both paired BNs are successfully decoded, it requires both of the SINR from the stronger signal and the SNR from the weaker signal to be greater than the channel threshold. Based on the decoding order, $p_{2}$ can be decomposed into
\begin{align}\label{eq:deliever1}
p_{2}=&\underbrace{\Pr\!\left(\frac{P_T\xi_1g_1^{2}r_1^{-2\alpha}}{P_T\xi_2g_2^{2}r_2^{-2\alpha}+\mathcal{N}}\geq\gamma \,\,\&\& \,\,\frac{P_T\xi_2g_2^{2}r_2^{-2\alpha}}{\mathcal{N}}\geq\gamma\,\, \&\&\,\, \frac{\xi_1g_1^{2}}{r_1^{2\alpha}}\geq \frac{\xi_2g_2^{2}}{r_2^{2\alpha}}\right)}_{p_2^A} \nonumber \\
&+\underbrace{\Pr\!\left(\frac{P_T\xi_2g_2^{2}r_2^{-2\alpha}}{P_T\xi_1g_1^{2}r_1^{-2\alpha}+\mathcal{N}}\geq\gamma \,\,\&\& \,\,\frac{P_T\xi_1g_1^{2}r_1^{-2\alpha}}{\mathcal{N}}\geq\gamma\,\, \&\& \,\, \frac{\xi_1g_1^{2}}{r_1^{2\alpha}}<\frac{\xi_2g_2^{2}}{r_2^{2\alpha}}\right)}_{p_2^B},
\end{align}
\noindent where $x_1\triangleq r_1^{2\alpha}/g_1^{2}$, $x_2\triangleq r_2^{2\alpha}/g_2^{2}$, $g_1$ and $g_2$ represent the fading power gain for the BN from the near subregion and far subregion, respectively.
\else
   \begin{figure*}[!t]
\normalsize
\begin{align}\label{eq:deliever1}
p_{2}=&\underbrace{\Pr\!\left(\frac{P_T\xi_1g_1^{2}r_1^{-2\alpha}}{P_T\xi_2g_2^{2}r_2^{-2\alpha}+\mathcal{N}}\geq\gamma \,\,\&\& \,\,\frac{P_T\xi_2g_2^{2}r_2^{-2\alpha}}{\mathcal{N}}\geq\gamma\,\, \&\&\,\, \frac{\xi_1g_1^{2}}{r_1^{2\alpha}}\geq \frac{\xi_2g_2^{2}}{r_2^{2\alpha}}\right)}_{p_2^A} \nonumber \\
&+\underbrace{\Pr\!\left(\frac{P_T\xi_2g_2^{2}r_2^{-2\alpha}}{P_T\xi_1g_1^{2}r_1^{-2\alpha}+\mathcal{N}}\geq\gamma \,\,\&\& \,\,\frac{P_T\xi_1g_1^{2}r_1^{-2\alpha}}{\mathcal{N}}\geq\gamma\,\, \&\& \,\, \frac{\xi_1g_1^{2}}{r_1^{2\alpha}}<\frac{\xi_2g_2^{2}}{r_2^{2\alpha}}\right)}_{p_2^B}.
\end{align}
 \hrulefill
\vspace*{4pt}
\vspace{-0.05 in}
\end{figure*}
 In order to ensure that the signals from both paired BNs are successfully decoded, it requires both of the SINR from the stronger signal and the SNR from the weaker signal to be greater than the channel threshold. Based on the decoding order, $p_{2}$ can be decomposed into~\eqref{eq:deliever1}, as shown at the top of next page, where $x_1\triangleq r_1^{2\alpha}/g_1^{2}$, $x_2\triangleq r_2^{2\alpha}/g_2^{2}$, $g_1$ and $g_2$ represent the fading power gain for the BN from the near subregion and far subregion, respectively.
\fi

Let us consider $p_2^A$ firstly, which is the probability that both BNs are successfully decoded when the signal from the near BN is decoded at first. The condition of the signal from the near BN being decoded at first is $\xi_1 x_1\geq\xi_2 x_2$ (equivalently, $x_1\geq \kappa x_2$). Additionally, the condition that the signal from the far BN is successfully decoded is that $x_2$ must be greater than $\frac{\mathcal{N}\gamma}{P_T\xi_2}$. Then, we can express $p_2^A$ as
\ifCLASSOPTIONpeerreview
\begin{align}\label{eq:p2_derive}
p_2^A&=\Pr\left(\frac{x_1}{\kappa x_2+\frac{\mathcal{N}}{P_T\xi_1}}\geq\gamma\right)=\mathbb{E}_{x_2}\left\{\Pr\left(x_1\geq \gamma\kappa x_2+\frac{\mathcal{N}\gamma}{P_T\xi_1} \right)\right\},
\end{align}
\else
\begin{align}\label{eq:p2_derive}
p_2^A&=\Pr\left(\frac{x_1}{\kappa x_2+\frac{\mathcal{N}}{P_T\xi_1}}\geq\gamma\right)\nonumber\\
&=\mathbb{E}_{x_2}\left\{\Pr\left(x_1\geq \gamma\kappa x_2+\frac{\mathcal{N}\gamma}{P_T\xi_1} \right)\right\},
\end{align}
\fi
\noindent where $x_1\in\left(\kappa x_2,\infty\right)$ and $x_2\in\left(\frac{\mathcal{N}\gamma}{P_T\xi_2},\infty\right)$.

Then following the similar procedure as presented in Appendix A, we obtain the expression of $p_2^A$. $p_2^B$ can be derived using the same procedure. After combining these two results, we arrive at the final result in~\eqref{eq:general:p2}. The derivation of $p_1$ is similar.
\end{proof}

Due to the complexity of functions $\Phi$ and $\phi$, it is not possible to obtain the closed-form results. But the single-fold integration can be easily numerically evaluated using standard mathematical packages such as Mathematica or Matlab.

\begin{lemma}
Based on our system model considered in Sections~\ref{sec:system} and~\ref{sec:fading}, under the fading scenario with region division approach, the average number of successful BNs given that only one BN from the near subregion accesses the reader is
\begin{align}\label{eq:m1nr}
\bar{\mathcal{M}}_{1\near}&=1-\Phi^A\left(\frac{\mathcal{N}\gamma}{P_T\xi_1}\right),
\end{align}
\noindent and the average number of successful BNs given that only one BN from the far subregion accesses the reader is
\begin{align}\label{eq:m1fr}
\bar{\mathcal{M}}_{1\far}&=1-\Phi^B\left(\frac{\mathcal{N}\gamma}{P_T\xi_2}\right).
\end{align}
\end{lemma}
Since the derivation is similar to the proof of Lemma 2, we skip it here for the sake of brevity.

\subsection{Power Division Approach}\label{sec:powerdivision}
\ifCLASSOPTIONpeerreview
\else
   \begin{figure*}[!t]
\normalsize
\begin{align}\label{eq:p2newnew}
p_2\!=\left\{ \begin{array}{ll}
       0,&{\tilde{\beta}\leq\frac{\mathcal{N}\gamma}{P_T\xi_2} ;}\\
       \mathlarger{\int}_{\textrm{min}\!\left\{\tilde{\beta},\textrm{max}\!\left\{\frac{\mathcal{N}\gamma}{P_T\xi_2},\frac{\tilde{\beta}}{\gamma\kappa}-\frac{\mathcal{N}}{P_T\xi_2}\right\}\right\}}^{\tilde{\beta}}\left(1-\frac{\Phi\left(\gamma\kappa x_2'+\frac{\mathcal{N}\gamma}{P_T\xi_1},R_1,R\right)-\Phi\left(\tilde{\beta},R_1,R\right)}{1-\Phi\left(\tilde{\beta},R_1,R\right)}\right)\frac{\phi\left(x_2',R_1,R\right)}{\Phi\left(\tilde{\beta},R_1,R\right)}\textup{d}x_2'\\
       +\frac{\Phi\left(\textrm{min}\!\left\{\tilde{\beta},\textrm{max}\!\left\{\frac{\mathcal{N}\gamma}{P_T\xi_2},\frac{\tilde{\beta}}{\gamma\kappa}-\frac{\mathcal{N}}{P_T\xi_2}\right\}\right\},R_1,R\right)
-\Phi\left(\frac{\mathcal{N}\gamma}{P_T\xi_2},R_1,R\right)}{\Phi\left(\tilde{\beta},R_1,R\right)},&{\tilde{\beta}>\frac{\mathcal{N}\gamma}{P_T\xi_2} ;}\\
                    \end{array} \right.
\end{align}
\begin{align}\label{eq:p1newnew}
p_1=&\mathlarger{\int}_{\textrm{min}\!\left\{\tilde{\beta},\frac{\mathcal{N}\gamma}{P_T\xi_2},\textrm{max}\!\left\{0,\frac{\tilde{\beta}}{\gamma\kappa}-\frac{\mathcal{N}}{P_T\xi_2}\right\}\right\}}^{\textrm{min}\left\{\tilde{\beta}, \frac{\mathcal{N}\gamma}{P_T\xi_2}\right\}}\left(1-\frac{\Phi\left(\gamma\kappa x_2'+\frac{\mathcal{N}\gamma}{P_T\xi_1},R_1,R\right)-\Phi\left(\tilde{\beta},R_1,R\right)}{1-\Phi\left(\tilde{\beta},R_1,R\right)}\right)\frac{\phi\left(x_2',R_1,R\right)}{\Phi\left(\tilde{\beta},R_1,R\right)}\textup{d}x_2'\nonumber\\
&+\frac{\Phi\left(\textrm{min}\!\left\{\tilde{\beta},\frac{\mathcal{N}\gamma}{P_T\xi_2},\textrm{max}\!\left\{0,\frac{\tilde{\beta}}{\gamma\kappa}-\frac{\mathcal{N}}{P_T\xi_2}\right\}\right\},R_1,R\right)}{\Phi\left(\tilde{\beta},R_1,R\right)}.
\end{align}
 \hrulefill
\vspace*{4pt}
\vspace{-0.05 in}
\end{figure*}
\fi
Under the power division approach, rather than pairing the BNs from different subregions, the reader pairs the BNs with different power levels. Specifically, for the reader, there is a pre-defined threshold $\beta$ and training period at the start of each time slot. By comparing the threshold $\beta$ with the instantaneous backscattered signal power from each node, the reader categorizes the BNs into high power level group and low power level group. Correspondingly, each BN can pick its reflection coefficient by comparing its received power with the threshold $(1-\xi_1)\sqrt{P_T\beta/\xi_1}$ \footnote{In the training period, all the BNs' reflection coefficients are assumed to be $\xi_1$.}. If the received power is greater than the threshold, this BN belongs to the high power level group and its reflection coefficient will be set to $\xi_1$. Otherwise, it belongs to the low level power level group and the reflection coefficient is set to $\xi_2$.

According to the principle of power division approach, $p_{\near}$ can be interpreted as the probability that the backscattered signal power for the node is greater than the threshold $\beta$. Thus, $p_{\near}$ can be written as $p_{\near}=1-\Phi\left(\tilde{\beta},R_1,R\right)$, where $\tilde{\beta}\triangleq \beta/(P_T \xi_1)$ is the normalized threshold. The key results for $p_2$, $p_1$, $\bar{\mathcal{M}}_{1\near}$ and $\bar{\mathcal{M}}_{1\far}$ are given in the following lemmas.

\ifCLASSOPTIONpeerreview
\begin{lemma}
Based on our system model considered in Sections~\ref{sec:system} and~\ref{sec:fading}, under the fading scenario with power division approach, the probability that the signals from two BNs are successfully decoded is
\begin{align}\label{eq:p2newnew}
p_2\!=\left\{ \begin{array}{ll}
       0,&{\tilde{\beta}\leq\frac{\mathcal{N}\gamma}{P_T\xi_2} ;}\\
       \mathlarger{\int}_{\textrm{min}\!\left\{\tilde{\beta},\textrm{max}\!\left\{\frac{\mathcal{N}\gamma}{P_T\xi_2},\frac{\tilde{\beta}}{\gamma\kappa}-\frac{\mathcal{N}}{P_T\xi_2}\right\}\right\}}^{\tilde{\beta}}\left(1-\frac{\Phi\left(\gamma\kappa x_2'+\frac{\mathcal{N}\gamma}{P_T\xi_1},R_1,R\right)-\Phi\left(\tilde{\beta},R_1,R\right)}{1-\Phi\left(\tilde{\beta},R_1,R\right)}\right)\frac{\phi\left(x_2',R_1,R\right)}{\Phi\left(\tilde{\beta},R_1,R\right)}\textup{d}x_2'\\
       +\frac{\Phi\left(\textrm{min}\!\left\{\tilde{\beta},\textrm{max}\!\left\{\frac{\mathcal{N}\gamma}{P_T\xi_2},\frac{\tilde{\beta}}{\gamma\kappa}-\frac{\mathcal{N}}{P_T\xi_2}\right\}\right\},R_1,R\right)
-\Phi\left(\frac{\mathcal{N}\gamma}{P_T\xi_2},R_1,R\right)}{\Phi\left(\tilde{\beta},R_1,R\right)},&{\tilde{\beta}>\frac{\mathcal{N}\gamma}{P_T\xi_2} ;}\\
                    \end{array} \right.
\end{align}
\noindent and the probability that the signal from only one BN is successfully decoded is given by
\begin{align}\label{eq:p1newnew}
p_1=&\mathlarger{\int}_{\textrm{min}\!\left\{\tilde{\beta},\frac{\mathcal{N}\gamma}{P_T\xi_2},\textrm{max}\!\left\{0,\frac{\tilde{\beta}}{\gamma\kappa}-\frac{\mathcal{N}}{P_T\xi_2}\right\}\right\}}^{\textrm{min}\left\{\tilde{\beta}, \frac{\mathcal{N}\gamma}{P_T\xi_2}\right\}}\left(1-\frac{\Phi\left(\gamma\kappa x_2'+\frac{\mathcal{N}\gamma}{P_T\xi_1},R_1,R\right)-\Phi\left(\tilde{\beta},R_1,R\right)}{1-\Phi\left(\tilde{\beta},R_1,R\right)}\right)\frac{\phi\left(x_2',R_1,R\right)}{\Phi\left(\tilde{\beta},R_1,R\right)}\textup{d}x_2'\nonumber\\
&+\frac{\Phi\left(\textrm{min}\!\left\{\tilde{\beta},\frac{\mathcal{N}\gamma}{P_T\xi_2},\textrm{max}\!\left\{0,\frac{\tilde{\beta}}{\gamma\kappa}-\frac{\mathcal{N}}{P_T\xi_2}\right\}\right\},R_1,R\right)}{\Phi\left(\tilde{\beta},R_1,R\right)}.
\end{align}
\end{lemma}
\else
\begin{lemma}
Based on our system model considered in Sections~\ref{sec:system} and~\ref{sec:fading}, under the fading scenario with power division approach, the probability that the signals from two BNs are successfully decoded is~\eqref{eq:p2newnew} and the probability that the signal from only one BN is successfully decoded is given by~\eqref{eq:p1newnew}, as shown at the top of this page.
\end{lemma}
\fi
\begin{proof}
Let $x_1'$ represent the normalized instantaneous received power from a BN belonging to the high power level group, which is normalized over $P_T$ and $\xi_1$, and its CDF can be expressed as $F_{x_1'}(x_1')=\frac{\Phi\left(x_1',R_1,R\right)-\Phi\left(\tilde{\beta},R_1,R\right)}{1-\Phi\left(\tilde{\beta},R_1,R\right)}$, where $x_1'\in[\tilde{\beta},\infty)$. Similarly, let
$x_2'$ denote the normalized instantaneous received power from a BN belonging to the lower power level group and its CDF is $F_{x_2'}(x_2')=\frac{\Phi\left(x_2',R_1,R\right)}{\Phi\left(\tilde{\beta},R_1,R\right)}$, where $x_2'\in(0,\tilde{\beta})$.

Clearly, under the power division approach, the decoding order is always from the high power level group to the low power level group. $p_2$ and $p_1$ are then written as $p_2=\mathbb{E}_{x_1',x_2'}\left[\Pr\left(\frac{P_T\xi_1x_1'}{P_T\xi_2x_2'+\mathcal{N}}\geq\gamma \&\&\frac{P_T\xi_2x_2'}{\mathcal{N}}\geq\gamma\right)\right]$ and $p_1=\mathbb{E}_{x_1',x_2'}\left[\Pr\left(\frac{P_T\xi_1x_1'}{P_T\xi_2x_2'+\mathcal{N}}\geq\gamma \&\&\frac{P_T\xi_2x_2'}{\mathcal{N}}<\gamma\right)\right]$, respectively. Following the similar derivation approach in Appendix A, we arrive at the results in~\eqref{eq:p2newnew} and~\eqref{eq:p1newnew}.
\end{proof}

\begin{lemma}
Based on our system model considered in Sections~\ref{sec:system} and~\ref{sec:fading}, under the fading scenario with power division approach, the average number of successful BNs given that only one BN from the near subregion accesses the reader is
\begin{align}
 \bar{\mathcal{M}}_{1\near} &= 1-\Phi\left(\frac{\mathcal{N}\gamma}{\xi_1P_T},R_1,R\right)\textbf{1}\left(\frac{\mathcal{N}\gamma}{\xi_1P_T}>\tilde{\beta} \right),\label{eq:m1nearnewnew}
\end{align}
\noindent and the average number of successful BNs given that only one BN from the far subregion accesses the reader is
\begin{align}\label{eq:m1farnewnew}
\bar{\mathcal{M}}_{1\far} &= \left(1-\Phi\left(\frac{\mathcal{N}\gamma}{\xi_2P_T},R_1,R\right)\right)\textbf{1}\left(\frac{\mathcal{N}\gamma}{\xi_2P_T}<\tilde{\beta} \right).
\end{align}
\end{lemma}
The derivation is similar to the proof of Lemma 2; hence, we skip it here for the sake of brevity.
\begin{remark}
The reflection coefficient selection criterion for the region division approach depends on the subregion radius and channel threshold, while the selection criterion for the power division approach strongly relies on the threshold $\beta$ and channel threshold. Following the same derivation procedure for Proposition 1, to achieve the better system performance, we can compute the relationship between $\xi_1$ and $\xi_2$ as
\begin{align}\label{eq:design2}
\xi_1\geq \max\left\{\xi_{2},\gamma\left (\xi_2+\frac{\mathcal{N}}{P_T\tilde{\beta}} \right)\right\}.
\end{align}
\end{remark}

\subsection{Summary}\label{sec:fading:summary}
Table~\ref{tb:1} summarizes the key factors used to calculate the average number of successfully decoded bits, $\bar{\mathcal{C}}_{\suc}$, under the fading-free and fading scenarios with different pairing approaches. The general expression of $\bar{\mathcal{C}}_{\suc}$ for two-node pairing is given in~\eqref{eq:general:C}, where $\bar{\mathcal{M}}_2=p_1+2p_2$.
 \begin{table*}[t]
\centering
\caption{Key factors determining $\bar{\mathcal{C}}_{\suc}$ for two-node pairing case.}\label{tb:1}
\begin{tabular}{|c||c|c|c|c|c|} \hline
 Scenario & $p_{\near}$ & $p_2$ & $p_1$ & $\bar{\mathcal{M}}_{1\near}$ & $\bar{\mathcal{M}}_{1\far}$\\ \hline
\multirow{2}{*}{\makecell[c]{Fading-free\\Fading: region division}} & \multirow{2}{*}{$\frac{R_2^2-R_1^2}{R^2-R_1^2}$} & \multirow{2}{*}{\makecell[c]{\eqref{eq:nofading:p2}\\\eqref{eq:general:p2}} } &\multirow{2}{*}{\makecell[c]{\eqref{eq:nofading:p1}\\\eqref{eq:general:p1}}} &\multirow{2}{*}{\makecell[c]{\eqref{eq:N1case}\\\eqref{eq:m1nr}}}&\multirow{2}{*}{\makecell[c]{\eqref{eq:N1case1}\\ \eqref{eq:m1fr}}}
 \\ \cline{1-1}\cline{3-6}
&&&& &\\ \cline{1-6}
  Fading: power division&$1-\Phi\left(\tilde{\beta},R_1,R\right)$&\eqref{eq:p2newnew} & \eqref{eq:p1newnew}& \eqref{eq:m1nearnewnew}& \eqref{eq:m1farnewnew} \\ \hline
\end{tabular}
\end{table*}
\ifCLASSOPTIONpeerreview
\else
   \begin{figure*}[!t]
\centering
\subfigure[Two-node pairing.]{\label{fig1b}\includegraphics[width=0.32\textwidth]{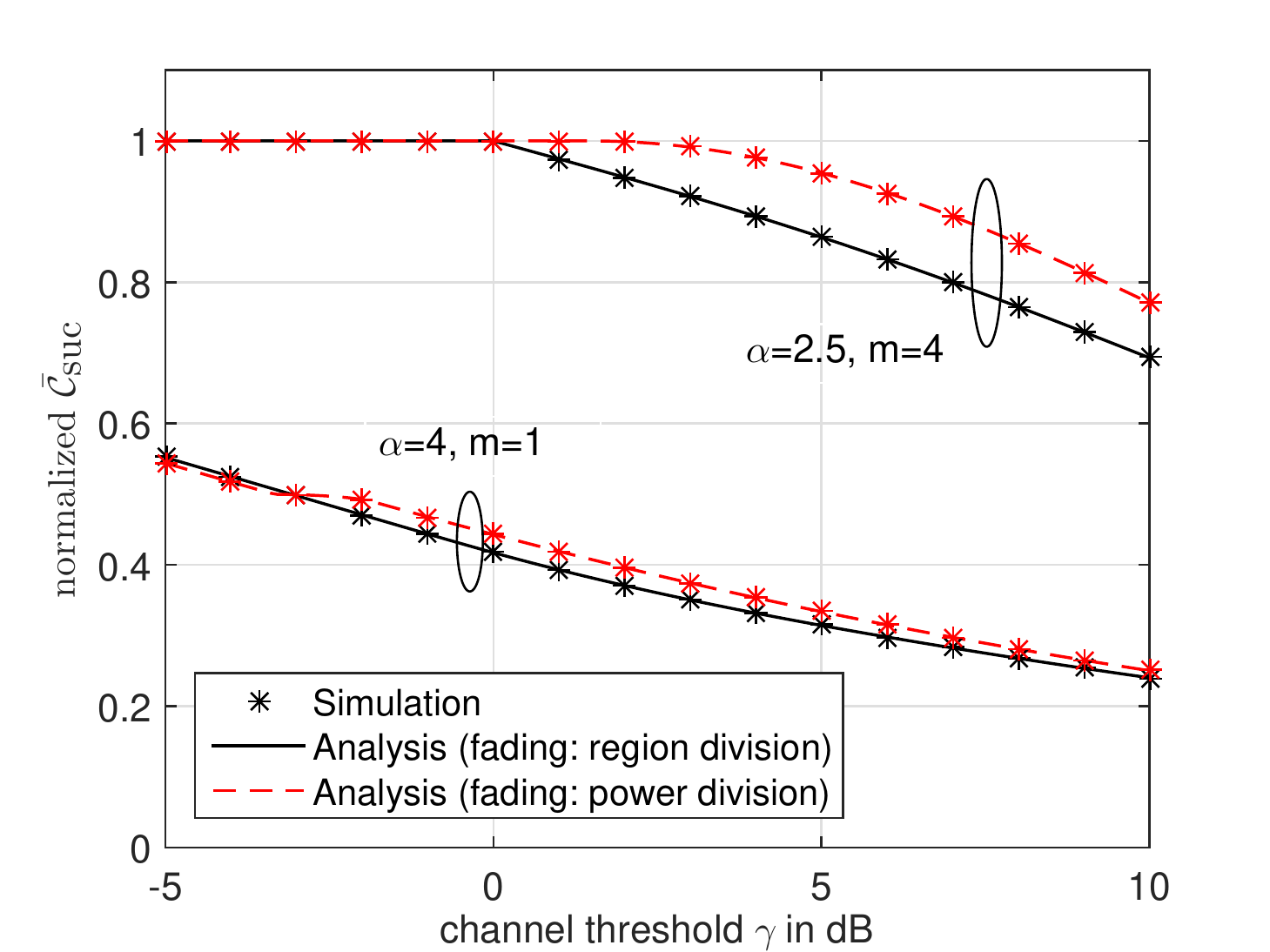}}
\subfigure[Multiple-node multiplexing.]{ \label{fig1a}\includegraphics[width=0.32\textwidth]{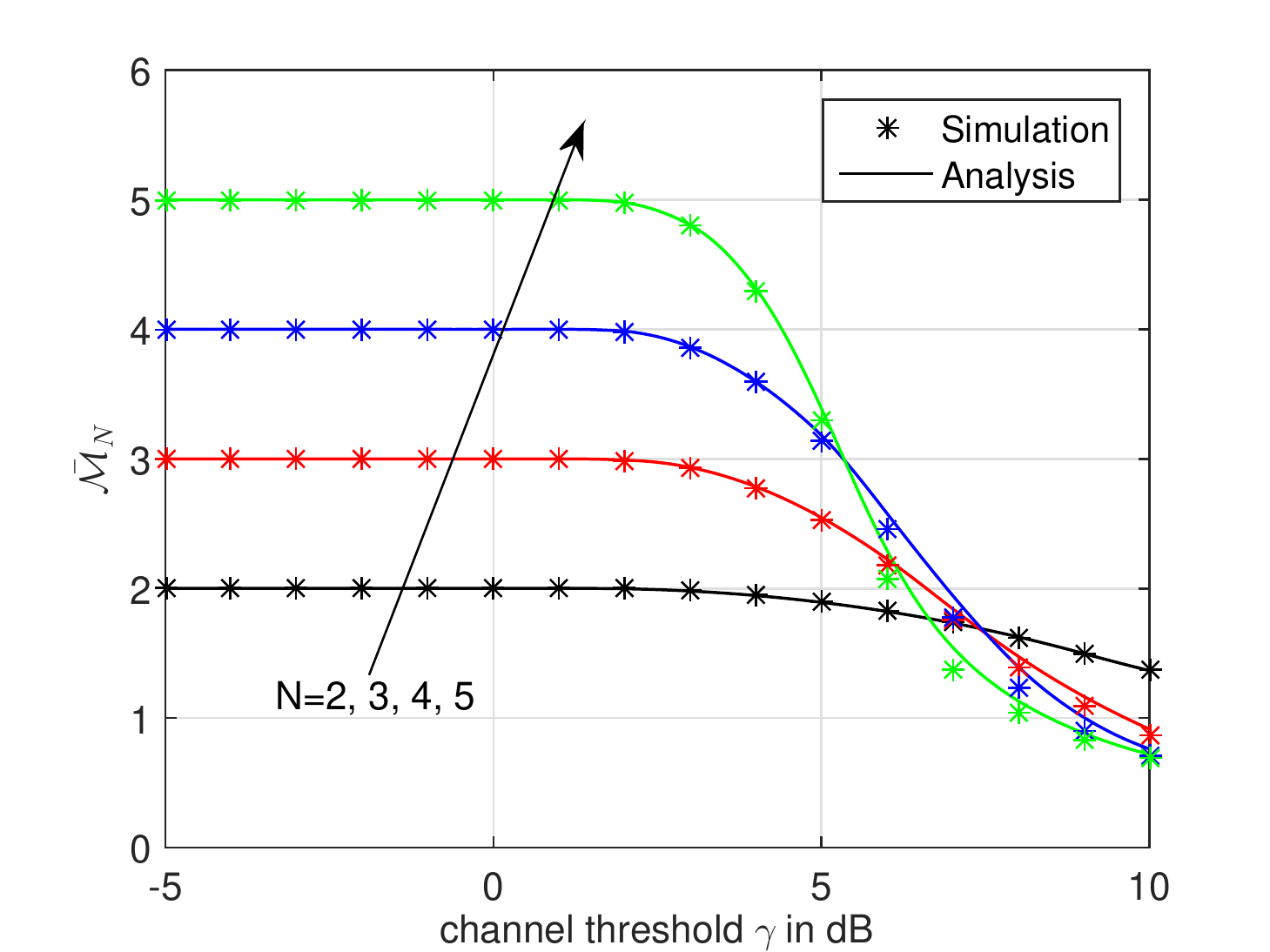}}
\subfigure[Multiple-node multiplexing (simulation only).]{ \label{fig1c}\includegraphics[width=0.32\textwidth]{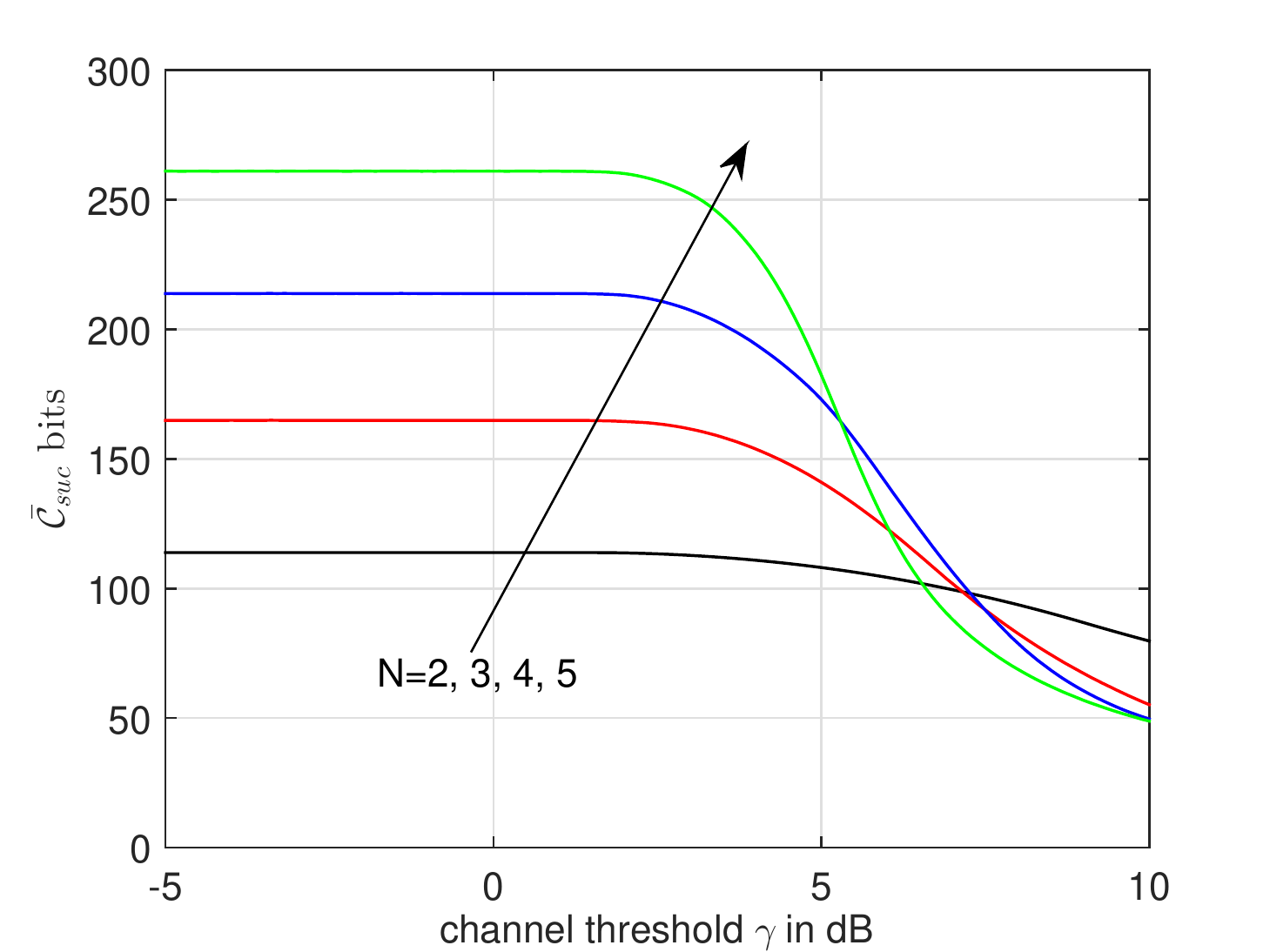}}

\caption{Channel threshold $\gamma$ versus (a) the normalized average number of successfully decoded bits under two-node pairing; (b) the average number of successful BNs $\bar{\mathcal{M}}_{N}$ given $N$ multiplexing nodes and (c) the average number of successfully decoded bits $\bar{\mathcal{C}}_{\suc}$ for general multiplexing case.}\label{fig_valid}
\vspace*{4pt}
\vspace{-0.05 in}
\end{figure*}
\fi

\section{Numerical Results}\label{sec:result}
In this section, we present the numerical results to investigate the performance of the NOMA-enhanced BackCom system. In order to validate the numerical results, we also present simulation results which are generated using Matlab and are averaged over $10^6$ simulation runs. Unless specified otherwise, the following values of the main system parameters are adopted: the outer radius of the coverage zone $R=65$ m, the inner radius of the coverage zone $R_1=1$ m, the number of BNs $M=60$, the path-loss exponent $\alpha=2.5$ for Nakagami-$m$ ($m=4$) fading scenario and fading-free scenario while $\alpha=4$ for Rayleigh fading case, the reader's transmit power $P_T=35$ dBm, the noise power $\mathcal{N}=-100$ dBm, and the product of the time slot and the data rate $\mathcal{L}\mathcal{R}=60$ bits. In addition, for the region division approach, we set $R_{i}=\sqrt{\frac{(i-1)R^2+(N+1-i)R_1^2}{N}}$ for $i={2,...,N}$. As for the power division approach, we find a $\tilde{\beta}$ value that makes $p_{\near}=0.5$. Note that such a value ensures that the average number of BNs in each group is the same. The impact of $R_2$ and $\tilde{\beta}$ will be analyzed in the following Section~\ref{sec:result::r2}.

\subsection{Analysis Validation}
\ifCLASSOPTIONpeerreview
\begin{figure}[t]
\centering
\subfigure[Two-node pairing.]{\label{fig1b}\includegraphics[width=0.45\textwidth]{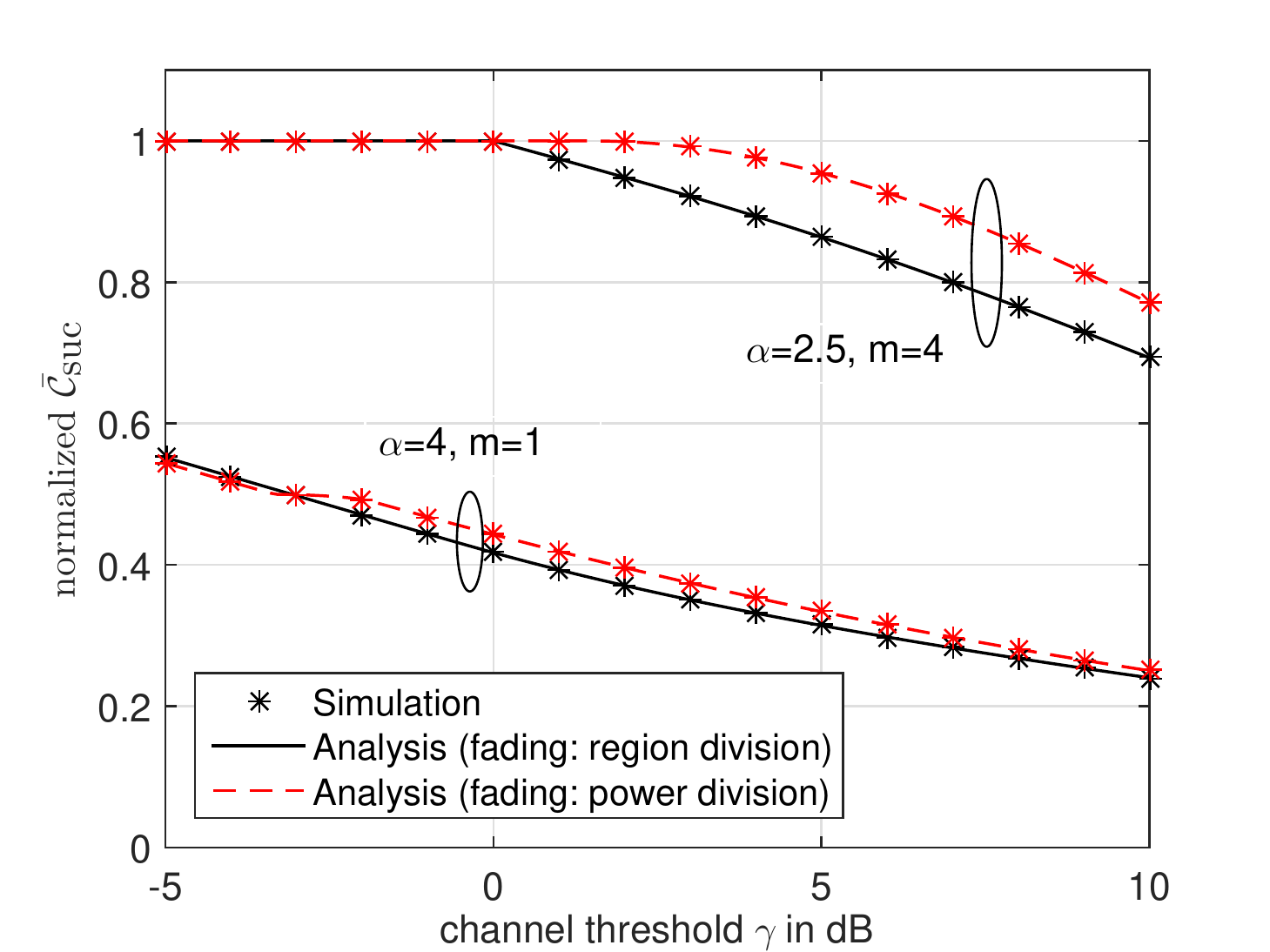}}\\
\subfigure[Multiple-node multiplexing.]{ \label{fig1a}\includegraphics[width=0.45\textwidth]{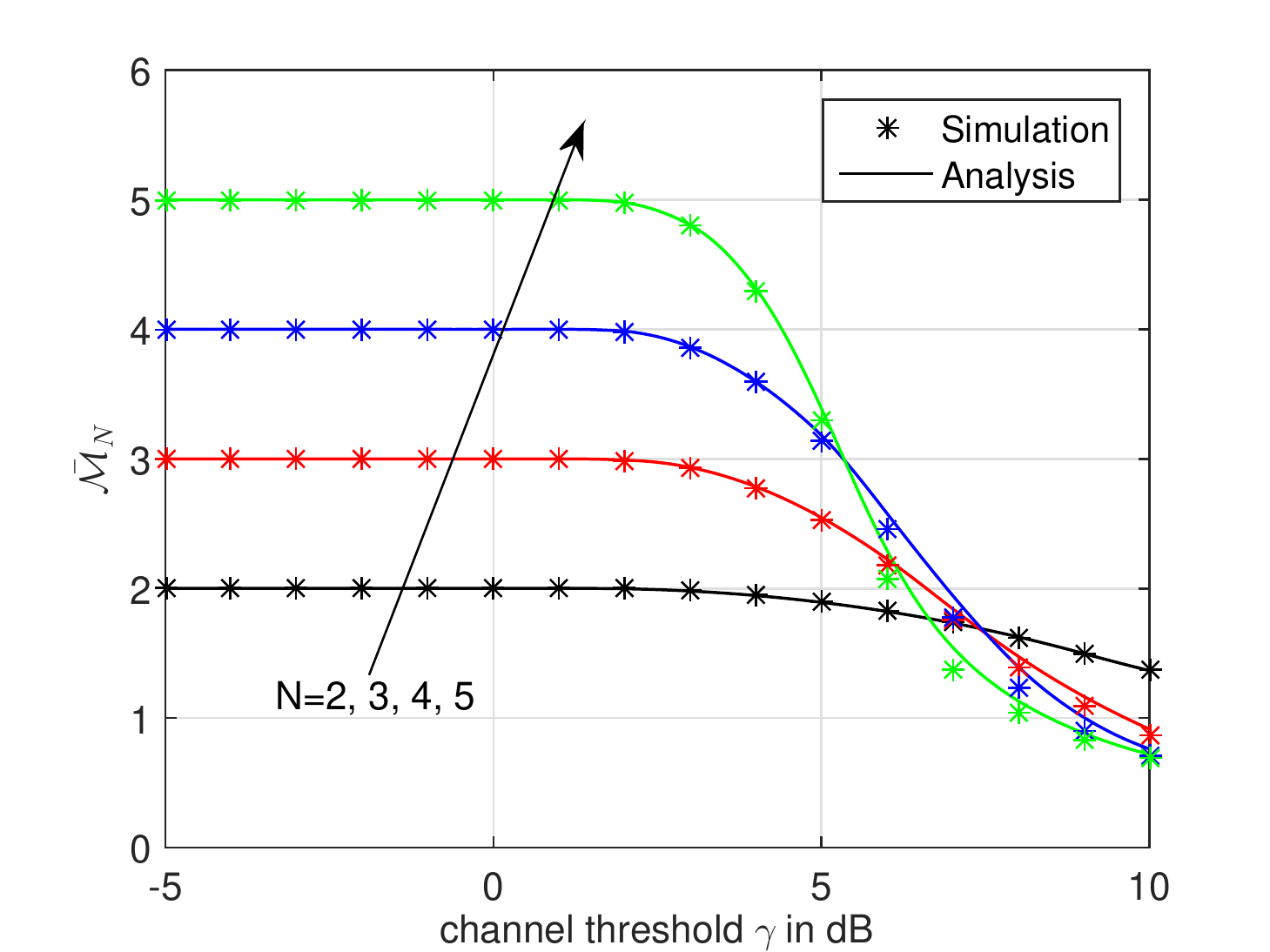}}
\subfigure[Multiple-node multiplexing (simulation only).]{ \label{fig1c}\includegraphics[width=0.45\textwidth]{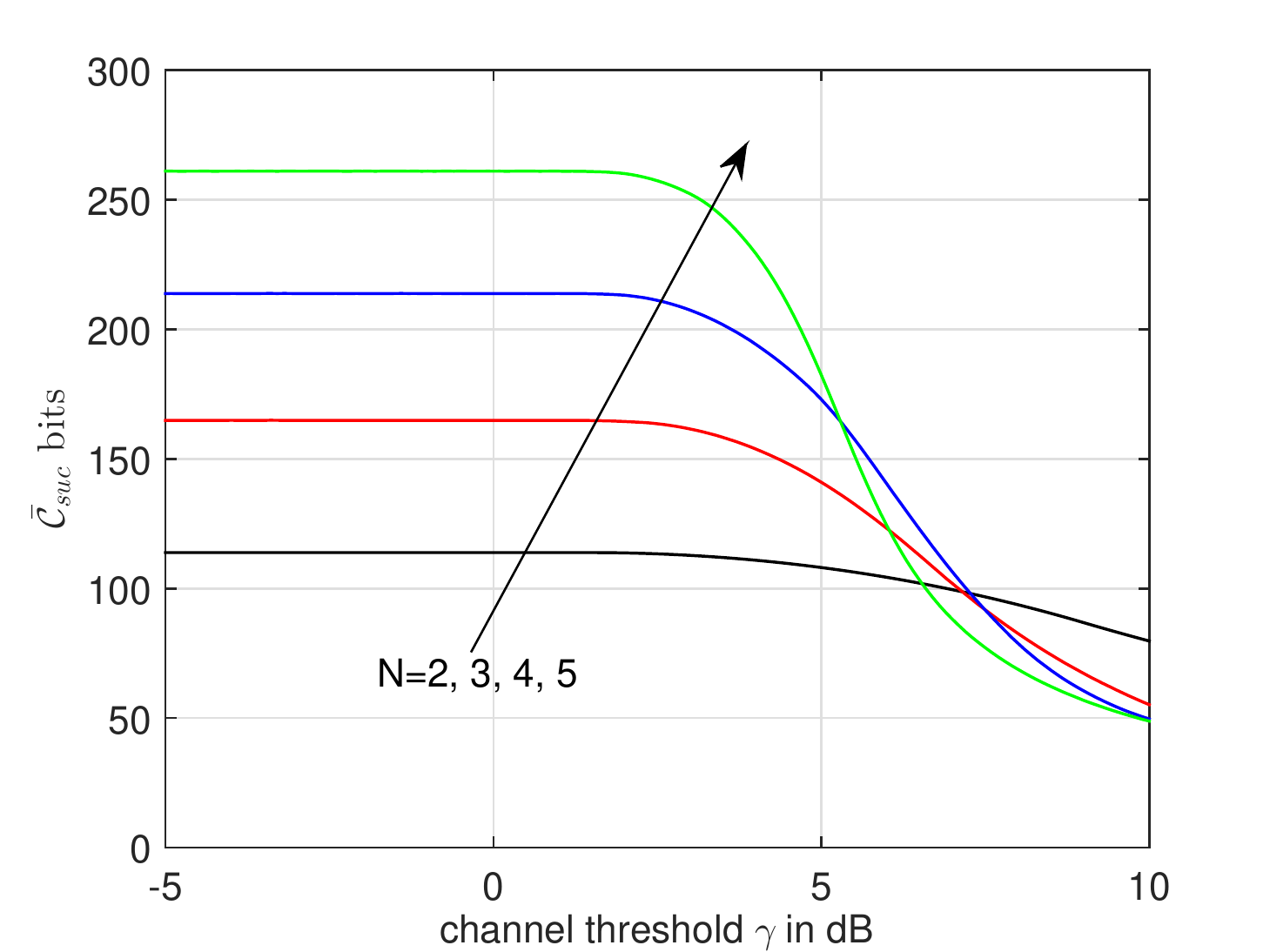}}

\caption{Channel threshold $\gamma$ versus (a) the normalized average number of successfully decoded bits under two-node pairing; (b) the average number of successful BNs $\bar{\mathcal{M}}_{N}$ given $N$ multiplexing nodes and (c) the average number of successfully decoded bits $\bar{\mathcal{C}}_{\suc}$ for general multiplexing case.}\label{fig_valid}
\end{figure}
\fi
Fig.~\ref{fig_valid} plots the channel threshold $\gamma$ versus the (normalized) average number of successfully decoded bits $\bar{\mathcal{C}}_{\suc}$ and the average number of successful BNs $\bar{\mathcal{M}}_{N}$ given $N$ multiplexing nodes for different fading and multiplexing scenarios. Note that the \textit{normalized average number of successfully decoded bits} is defined as the average number of successfully decoded bits over the total number of bits transmitted by BNs, where the latter term is a constant for the given system setup and is given by the formulation in Remark 2. We set $\xi_1=0.7$, $\xi_2=0.5$, $\xi_3=0.3$, $\xi_4=0.1$ and $\xi_5=0.05$. The curves in Fig.~\ref{fig1c} are generated using simulations only. From Fig.~\ref{fig1b} and $N=2$ curve in Fig.~\ref{fig1a}, we can see that the simulation results match perfectly with the analytical results as expected. As for the $N\geq3$ curves in Fig.~\ref{fig1a}, we find that, when the channel threshold is large, the analytical results slightly deviate from the simulation results. This is due to the independence assumption we made when calculating $p_k$ for more users multiplexing scenario. The (close) match of the simulation and analytical results demonstrates the accuracy of our derivations. Furthermore, comparing Fig.~\ref{fig1a} with Fig.~\ref{fig1c}, we find the trends for $\bar{\mathcal{M}}_{N}$ and $\bar{\mathcal{C}}_{\suc}$ are the same. This indicates that $\bar{\mathcal{M}}_{N}$ is a reasonable metric to investigate the performance.

As shown in Fig.~\ref{fig_valid}, when the considered reflection coefficient sets satisfy the criteria proposed in Proposition 1 and Remark 4 for certain value of $\gamma$, the curves for fading-free and Nakagami fading (i.e., $\alpha=2.4$, $m=4$) scenarios are constant. These are the best achievable system performance, where all the transmitted bits are successfully delivered and the average number of successfully decoded bits is the same as the the total number of bits transmitted by BNs. Note that the normalized $\bar{\mathcal{C}}_{\suc}$ under Rayleigh fading can only achieve about half of the best performance; hence, in the following subsections, we focus on the fading-free and Nakagami fading (i.e., $\alpha=2.4$, $m=4$) scenarios.
\subsection{Effect of the Reflection Coefficient for Two-Node Pairing}\label{sec:result:4}
\ifCLASSOPTIONpeerreview
  \begin{figure}[t]
        \centering
        \includegraphics[width=0.5  \textwidth]{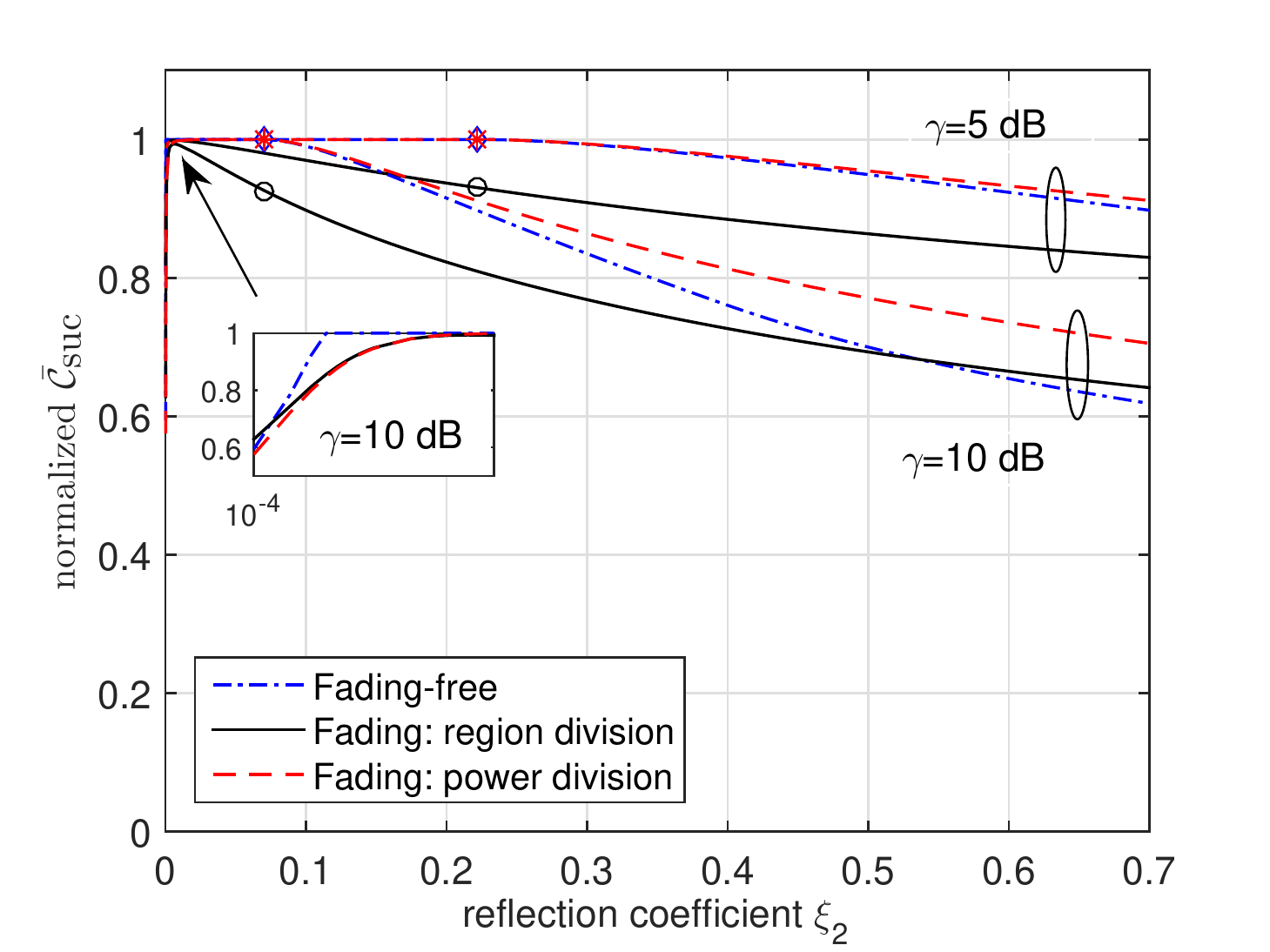}
        \caption{The reflection coefficient of the far backscatter group $\xi_2$ versus the normalized average number of successfully decoded bits for two-node pairing.}
        \label{fig_sec2}
\end{figure}
\fi
In this subsection, we investigate the effect of reflection coefficient for the two-node pairing case and examine the proposed reflection coefficient selection criteria. Fig.~\ref{fig_sec2} plots the reflection coefficient of the far backscatter group $\xi_2$ versus the normalized average number of successfully decoded bits. We set $\xi_1=0.7$.

From Fig.~\ref{fig_sec2}, we can see that the general trend for the normalized $\bar{\mathcal{C}}_{\suc}$ is decreasing as $\xi_2$ increases. This shows that a smaller value of $\xi_2$ can benefit the system. This is because, by reducing $\xi_2$, the interference from the weaker signal is reduced; the stronger signal, thus, has the higher chance to be decoded successfully. However, $\xi_2$ cannot be set too small, as the curves begin to decrease when $\xi_2$ approaches to an extremely small value (e.g., $10^{-4}$). When $\xi_2$ is extremely small, the weaker signal is less likely to be decoded successfully (i.e., SNR is very small for most of the time), which leads to the reduction of $p_2$, $\bar{\mathcal{M}}_{2}$ and $\bar{\mathcal{C}}_{\suc}$ correspondingly.

We also mark the maximum $\xi_2$ and the corresponding normalized $\bar{\mathcal{C}}_{\suc}$ which satisfies~\eqref{eq:designguide} or~\eqref{eq:design2} in Fig.~\ref{fig_sec2}. We find that, under the fading-free scenario or fading case with the power division approach, the marked normalized $\bar{\mathcal{C}}_{\suc}$ is equal to 1, which implies that the signals from all the paired BNs are successfully decoded and the system performance is consequently optimized. It also validates our proposed selection criteria. For the fading case with the region division approach, the proposed selection criterion still provides a good performance (i.e., the marked normalized $\bar{\mathcal{C}}_{\suc}$ is 0.9265  for $\gamma=10$ dB and 0.9306 for $\gamma=5$ dB).
\ifCLASSOPTIONpeerreview
\else
  \begin{figure}[t]
        \centering
        \includegraphics[width=0.45  \textwidth]{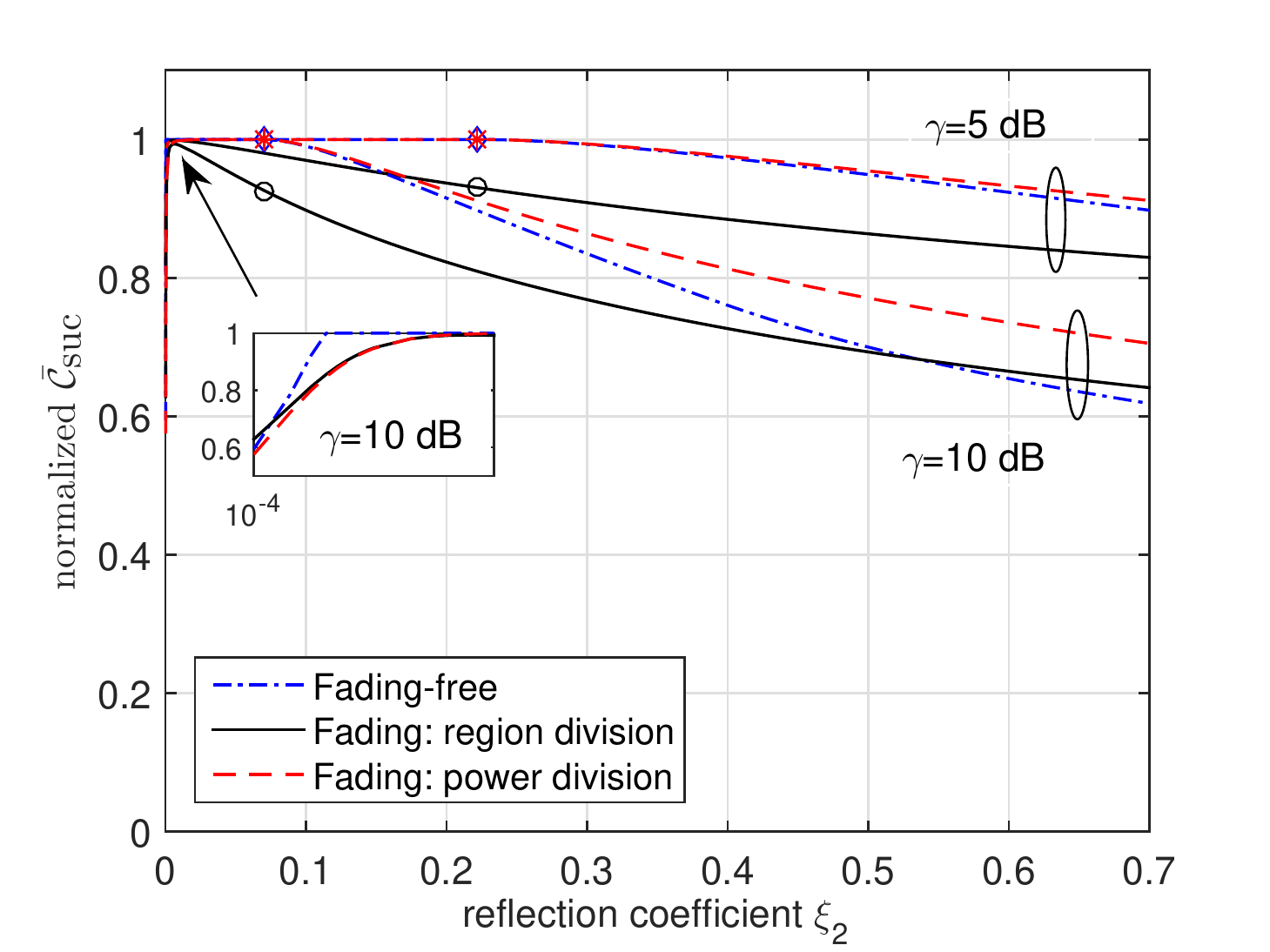}
        \caption{The reflection coefficient of the far backscatter group $\xi_2$ versus the normalized average number of successfully decoded bits for two-node pairing.}
        \label{fig_sec2}
\end{figure}
\fi
\subsection{Effect of the Reflection Coefficient for Multiple-Node Multiplexing}
\ifCLASSOPTIONpeerreview
\begin{figure}[t]
\centering
\subfigure[$N=3$.]{\label{fig3a}\includegraphics[width=0.45\textwidth]{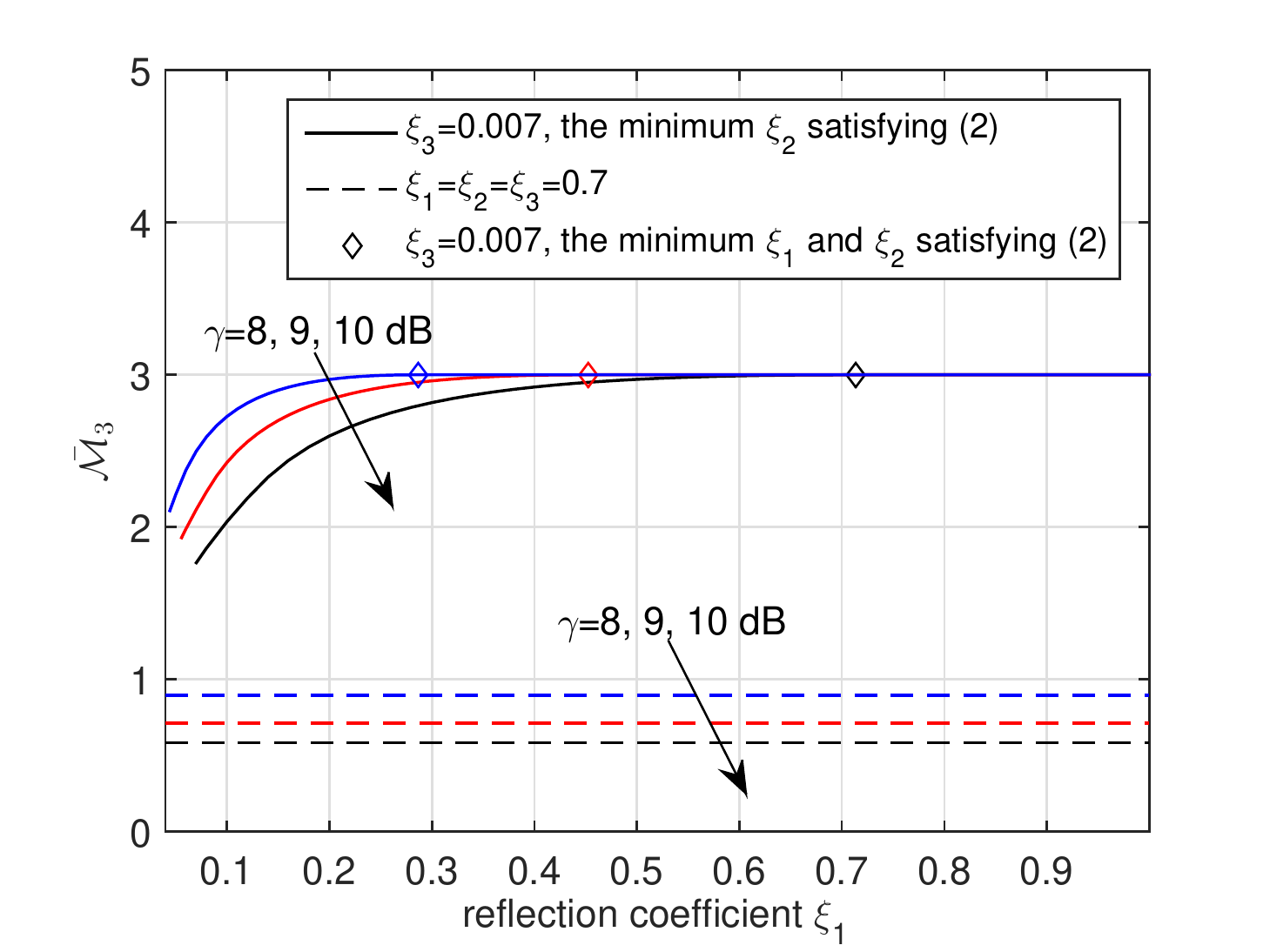}}
\mbox{\hspace{0.5cm}}
\subfigure[$N=5$.]{ \label{fig3b}\includegraphics[width=0.45\textwidth]{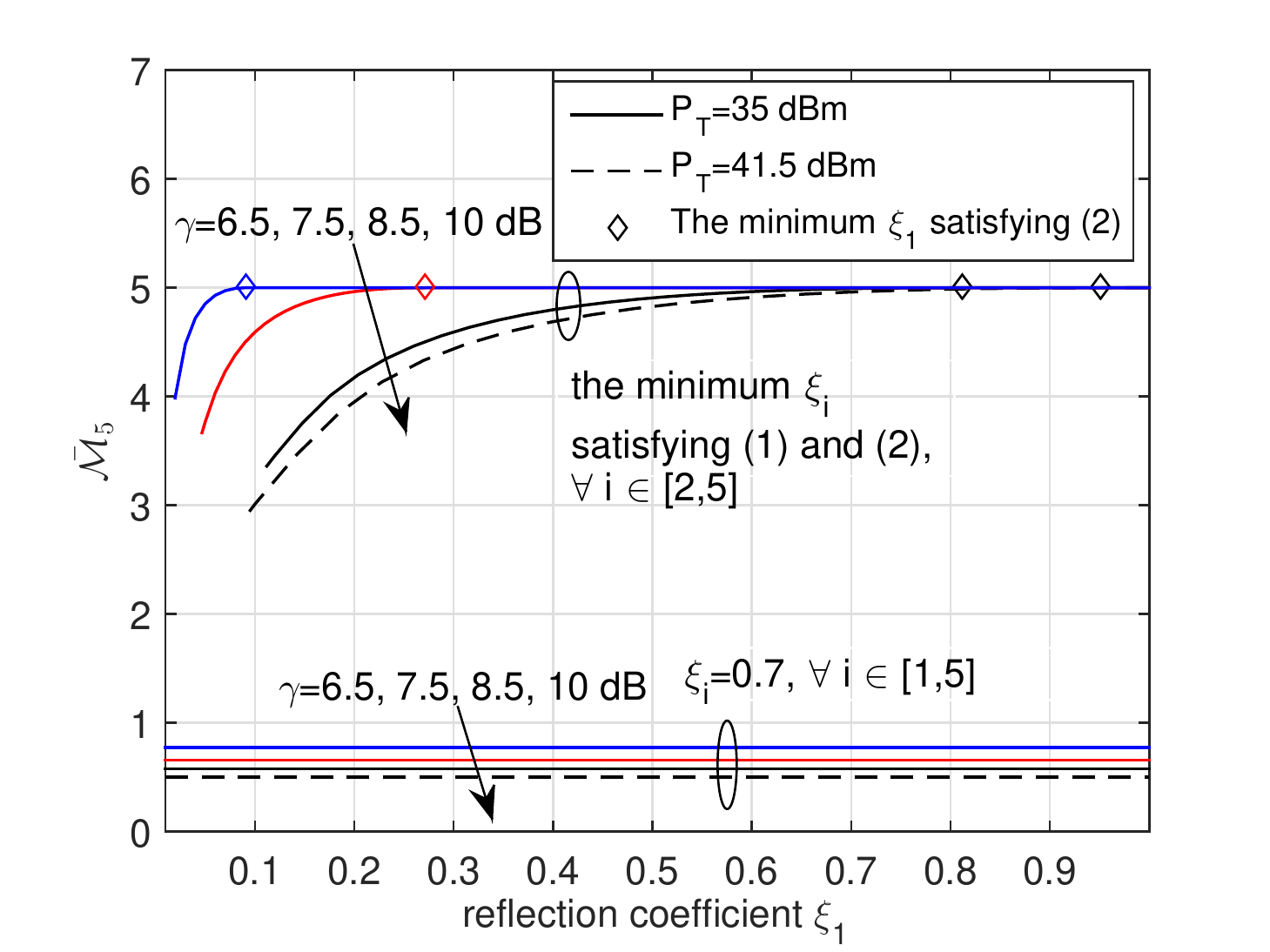}}

\caption{The reflection coefficient of the first backscatter group $\xi_1$ versus the average number of successful BNs $\bar{\mathcal{M}}_{N}$ given $N$ multiplexing BNs.}\label{fig_sec3}
\end{figure}
\else
\begin{figure}[t]
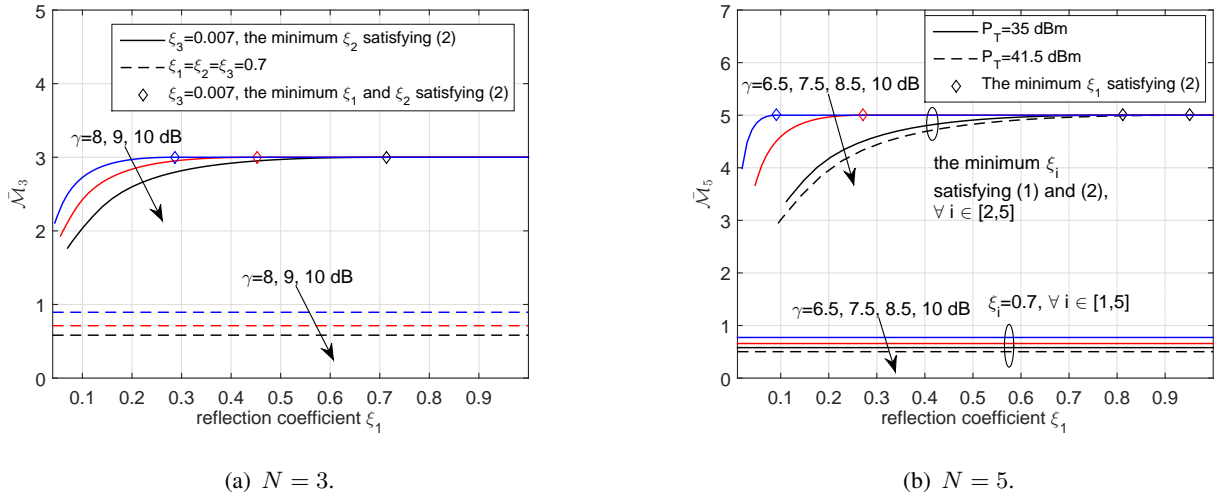

\centering
\subfigure[$N=3$.]{\label{fig3a}\includegraphics[width=0.45\textwidth]{fig3a}}\\
\subfigure[$N=5$.]{ \label{fig3b}\includegraphics[width=0.45\textwidth]{fig3b}}

\caption{The reflection coefficient of the first backscatter group $\xi_1$ versus the average number of successful BNs $\bar{\mathcal{M}}_{N}$ given $N$ multiplexing BNs.}\label{fig_sec3}
\end{figure}
\fi
Fig.~\ref{fig_sec3} plots the reflection coefficient of the first backscatter group $\xi_1$ versus the average number of successful BNs $\bar{\mathcal{M}}_{N}$ given $N$ multiplexing BNs, under the fading-free scenario, for $N=3$ and $N=5$, respectively. In Fig.~\ref{fig3a}, we set $\xi_3=0.007$ and $\xi_2$ to be the minimum value satisfying~\eqref{eq:designguide}. In Fig.~\ref{fig3b}, we set $\xi_5$ to be the minimum value which satisfies~\eqref{eq:designguide1} and $\xi_i$ (where $i \in [2,4]$) to be the minimum value satisfying~\eqref{eq:designguide}. We also mark the minimum $\xi_1$ satisfying~\eqref{eq:designguide}. As expected, these curves increase as $\xi_1$ increases and then become a constant (e.g., $\bar{\mathcal{M}}_{N}=N$) after the marked points. For the purpose of comparison, we also plot the curves when the reflection coefficients for all the subregions are the same. It is clear that, by properly selecting the reflection coefficients, the performance for BackCom system with NOMA can be greatly enhanced.

In addition, as shown in Fig.~\ref{fig3b}, when $N=5$ and $P_T=35$ dBm, the system can achieve the optimum performance when the channel threshold $\gamma$ is less than $8.5$ dB. If $\gamma$ further increases, the minimum $\xi_1$ satisfying~\eqref{eq:designguide} will be greater than one, which is impossible. Thus, we have to increase the transmit power of the reader in order to set a smaller $\xi_5$. In Fig.~\ref{fig3b}, we plot the curve for $P_T=41.5$ dBm, which allows to achieve the best performance when $\gamma$ is less than 10 dB.

\subsection{Effect of $R_2$ and $\tilde{\beta}$}\label{sec:result::r2}
\ifCLASSOPTIONpeerreview
\begin{figure}[t]
\centering
\subfigure[Fading-free scenario.]{\label{fig4a}\includegraphics[width=0.45\textwidth]{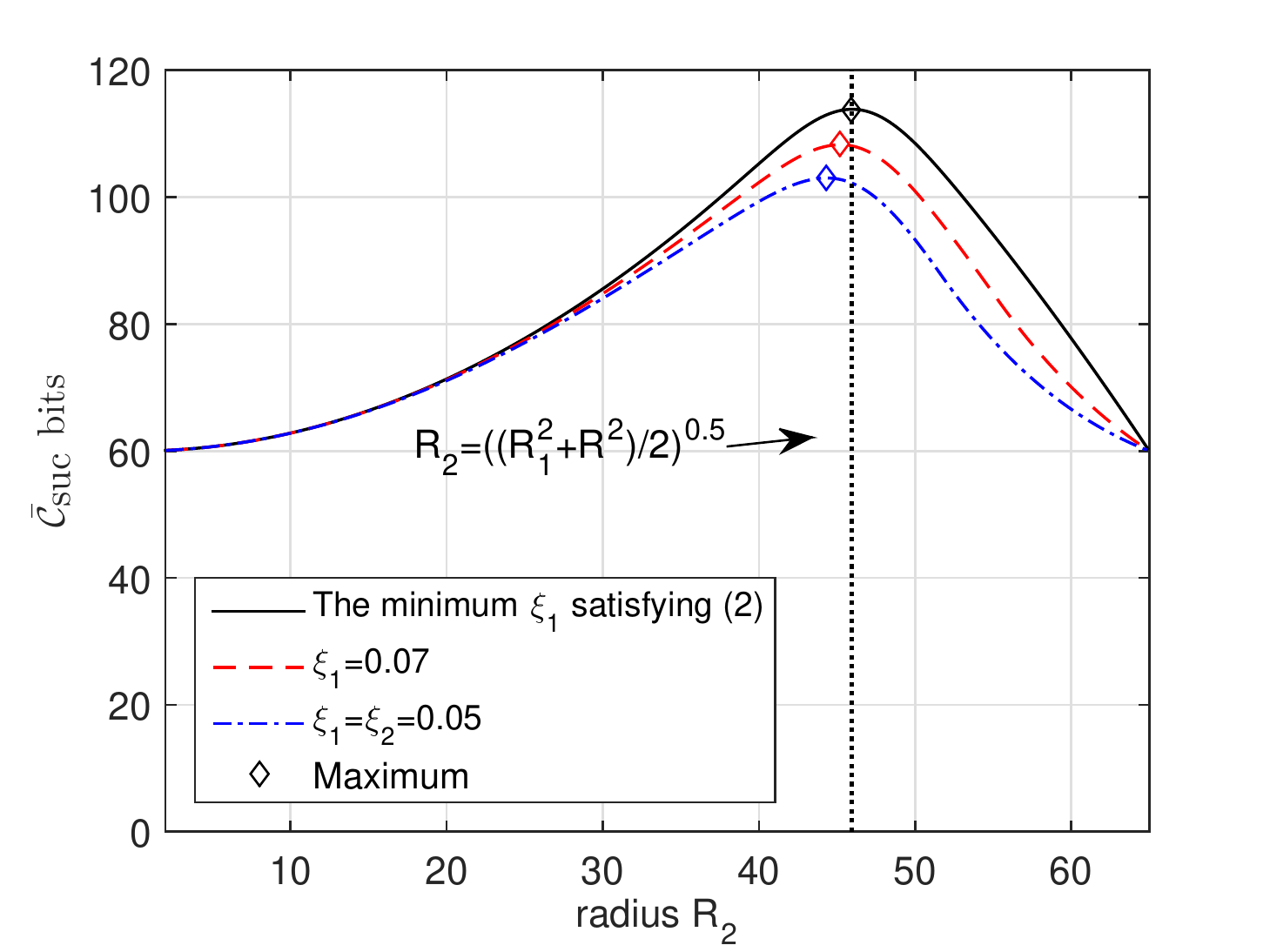}}\\
\subfigure[Fading: region division.]{ \label{fig4b}\includegraphics[width=0.45\textwidth]{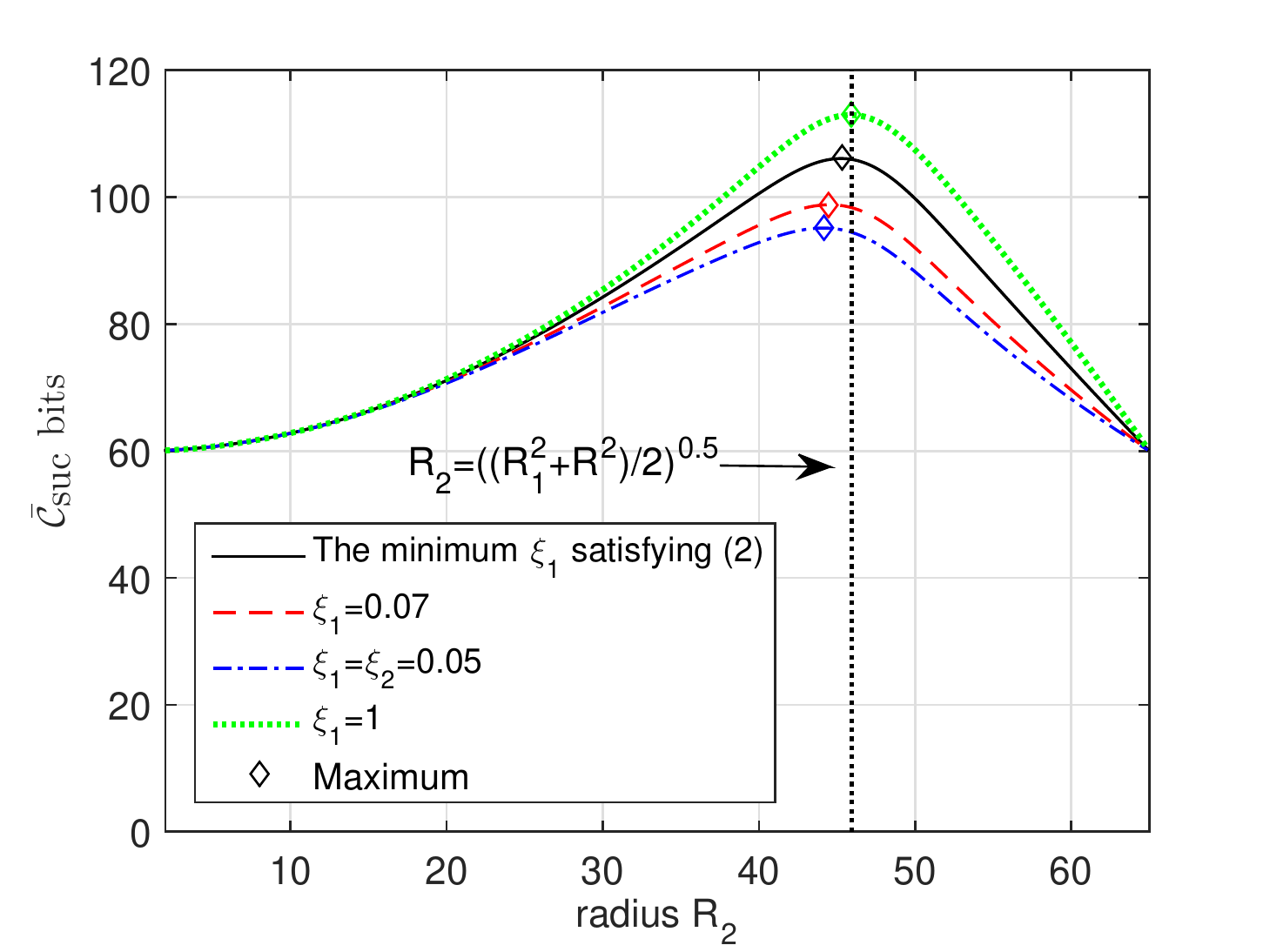}}
\subfigure[Fading: power division.]{ \label{fig4c}\includegraphics[width=0.45\textwidth]{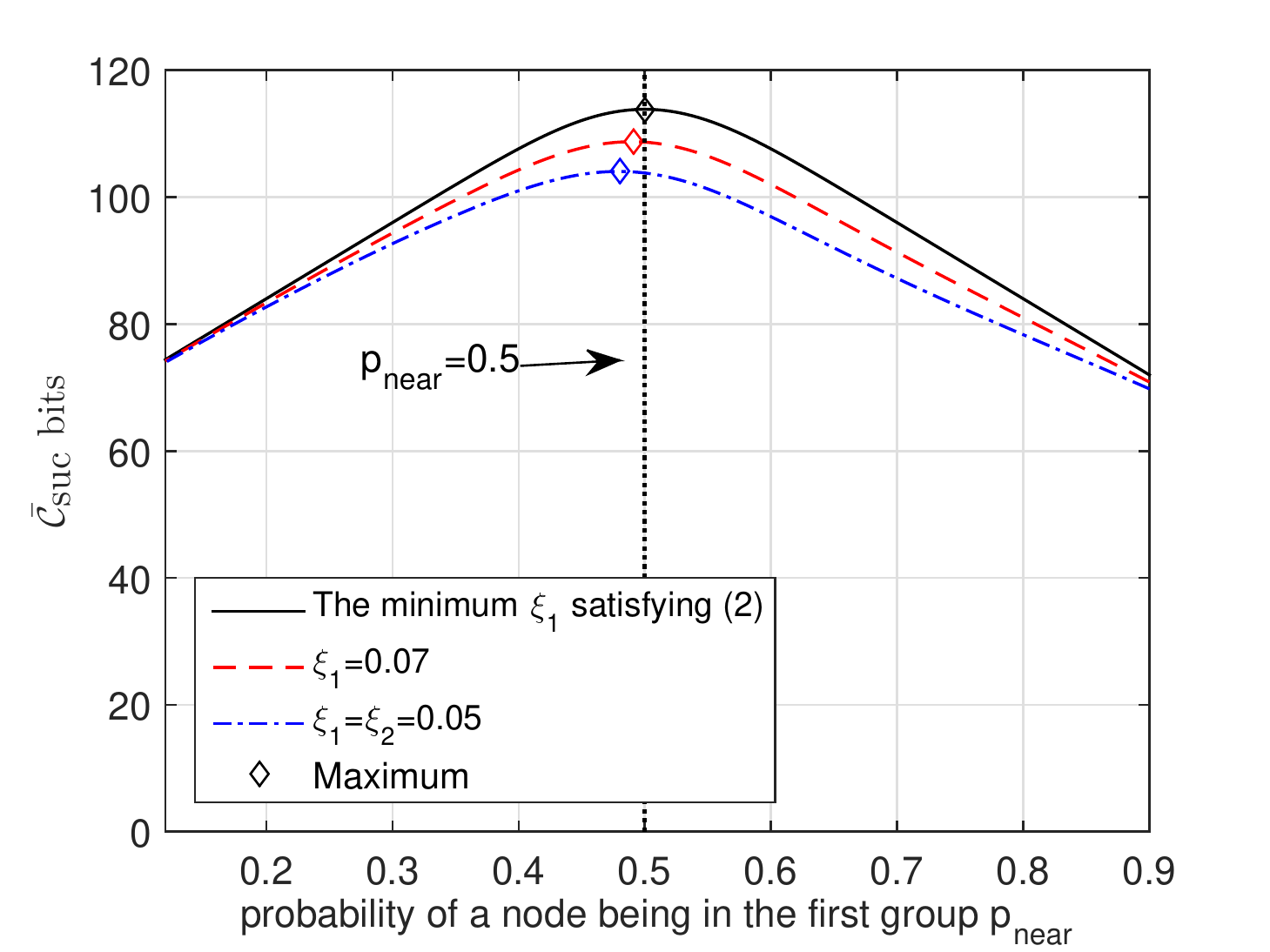}}

\caption{The average number of successfully decoded bits $\bar{\mathcal{C}}_{\suc}$ versus (a) the radius $R_2$ under fading-free scenario; (b) the radius $R_2$ with fading and (c) the probability $p_{\near}$ with power division approach.}\label{fig_sec4}
\end{figure}
\else
  \begin{figure*}[!t]
\centering
\subfigure[Fading-free scenario.]{\label{fig4a}\includegraphics[width=0.32\textwidth]{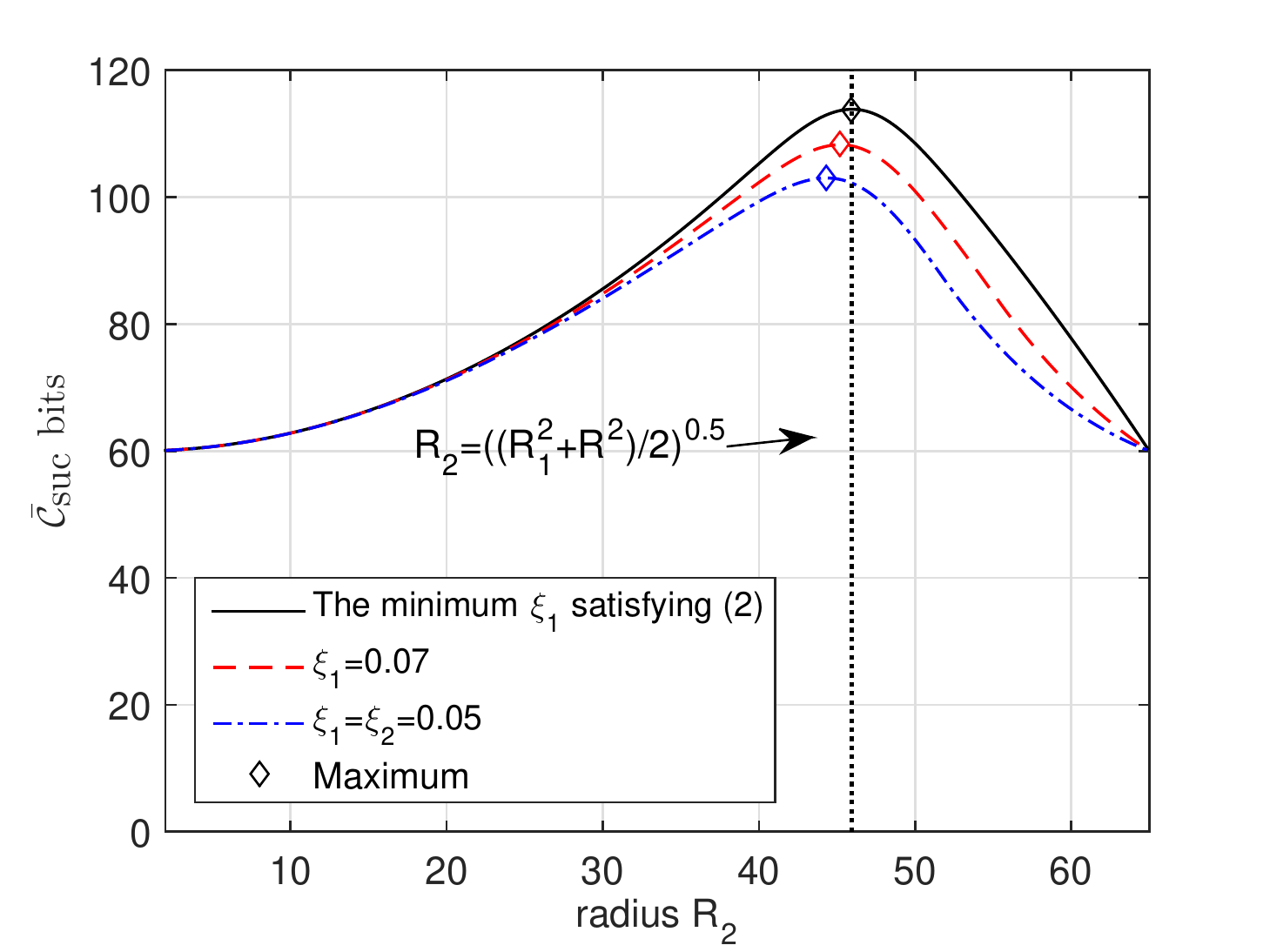}}
\subfigure[Fading: region division.]{ \label{fig4b}\includegraphics[width=0.32\textwidth]{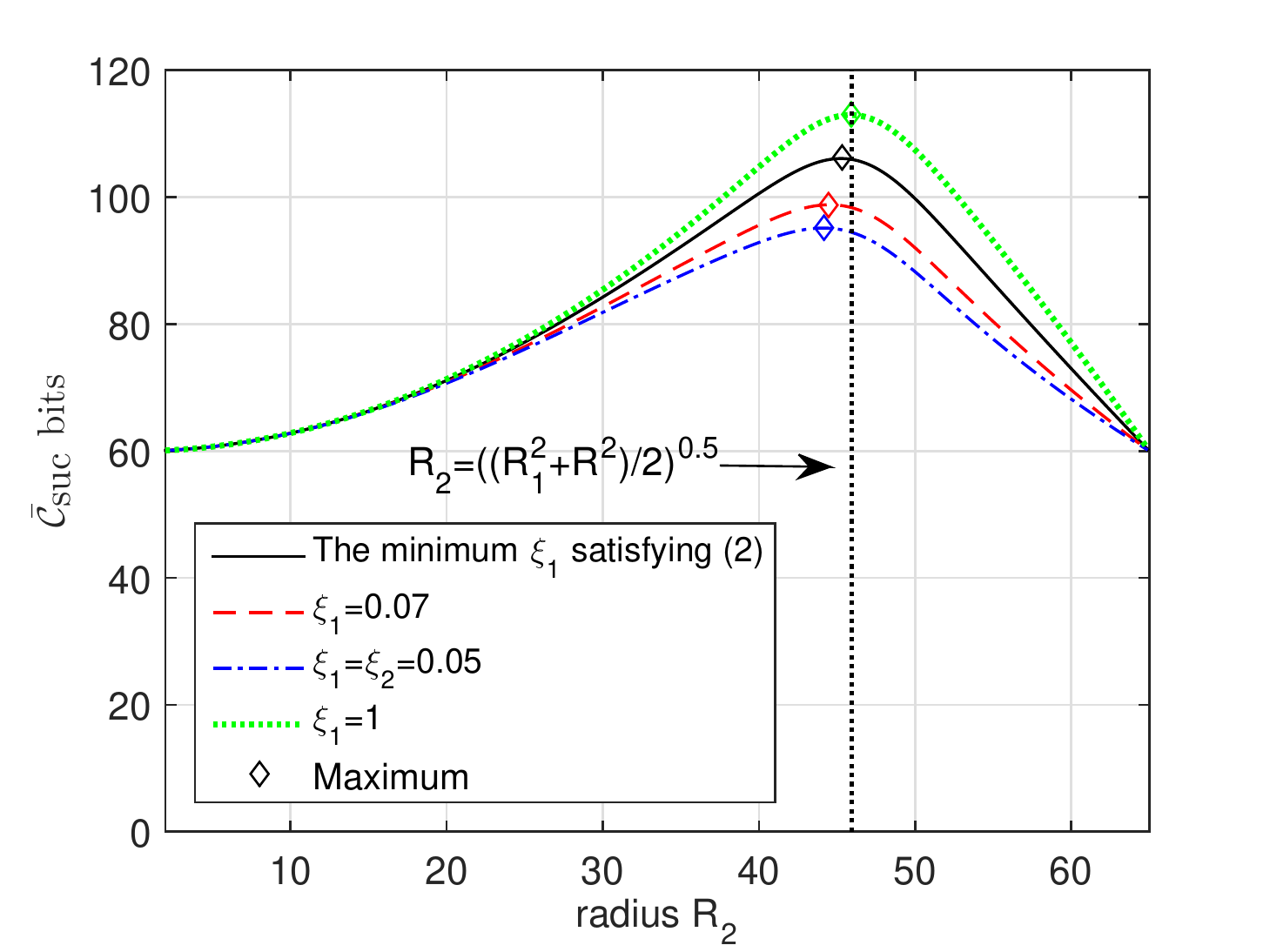}}
\subfigure[Fading: power division.]{ \label{fig4c}\includegraphics[width=0.32\textwidth]{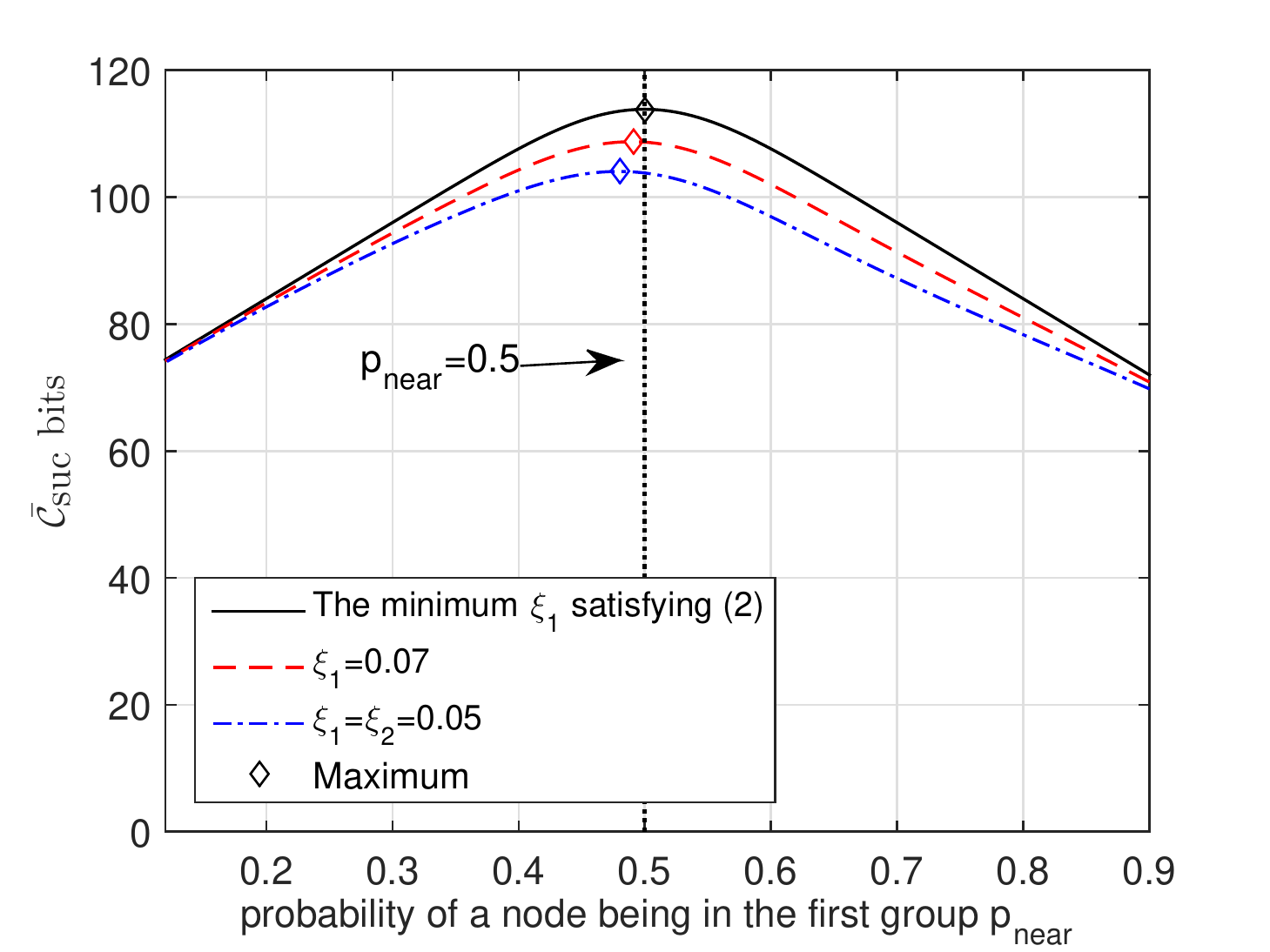}}

\caption{The average number of successfully decoded bits $\bar{\mathcal{C}}_{\suc}$ versus (a) the radius $R_2$ under fading-free scenario; (b) the radius $R_2$ with fading and (c) the probability $p_{\near}$ with power division approach.}\label{fig_sec4}
\vspace*{4pt}
\vspace{-0.05 in}
\end{figure*}
\fi
In this subsection, we investigate the impact of the radius $R_2$ for the region division approach and the threshold $\beta$ for the power division approach. To analyze the impact of system parameters, we focus on the metric $\bar{\mathcal{C}}_{\suc}$ rather than the normalized $\bar{\mathcal{C}}_{\suc}$, since the total number of bits transmitted by BNs varies for different system setups. Figs.~\ref{fig4a} and~\ref{fig4b} plot the radius $R_2$ versus the average number of successfully decoded bits $\bar{\mathcal{C}}_{\suc}$ for fading-free and fading scenarios, respectively. Fig.~\ref{fig4c} plots the probability $p_{\near}$ versus $\bar{\mathcal{C}}_{\suc}$ for the fading case with the power division approach. We set channel threshold $\gamma=5$ dB and $\xi_2=0.05$. We also mark the maximum $\bar{\mathcal{C}}_{\suc}$ reached by each case. From these figures, we can see that, when $p_{\near}$ is varying from 0 to 1 (equivalently, $R_2$ varies from $R_1$ to $R$ for the region division approach), $\bar{\mathcal{C}}_{\suc}$ first increases and then decreases. When $\xi_1$ follows the selection criterion in Proposition 1 and Remark 4, the maximum $\bar{\mathcal{C}}_{\suc}$ is achieved for $p_{\near}=0.5$ (equivalently, $R_2=\sqrt{\frac{R_1^2+R^2}{2}}$ for the region division approach). This is due to the fact that when the reflection coefficients follow the selection criterion in Proposition 1 or Remark 4, $\bar{M}_{2}$ is always equal to 2 and it is the best performance gain achieved by pairing BNs. Hence, the overall system can benefit more when more BNs are paired. $p_{\near}=0.5$ can result in the highest probability that all BNs are paired, thereby maximizing the average number of successfully decoded bits $\bar{\mathcal{C}}_{\suc}$.

For other scenarios, $p_{\near}=0.5$ may not lead to the maximum $\bar{\mathcal{C}}_{\suc}$. This is because, when $\bar{M}_{2}$ is no longer equal to 2, both $\bar{M}_{2}$ and the probability of different pairing cases (equivalently, $p_{\near}$) are determined by $R_2$ or $\tilde{\beta}$. The varying of $R_2$ or $\tilde{\beta}$ results in the different value of $\bar{M}_{2}$ and $p_{\near}$, and the interplay of these two factors results in the different maximum $\bar{\mathcal{C}}_{\suc}$ that can be achieved by the system. In addition, we find that the maximum $\bar{\mathcal{C}}_{\suc}$ achieved by the system where $\bar{M}_{2}<2$ is always less than the maximum $\bar{\mathcal{C}}_{\suc}$ achieved by the system where $\bar{M}_{2}=2$. This shows the importance of carefully selecting system parameters in order to achieve the best performance.

\subsection{Performance Gain Achieved by Applying NOMA to the BackCom System}
\ifCLASSOPTIONpeerreview
\begin{figure}[t]
\centering
\subfigure[The average number of successfully decoded bits.]{\label{fig5a}\includegraphics[width=0.45\textwidth]{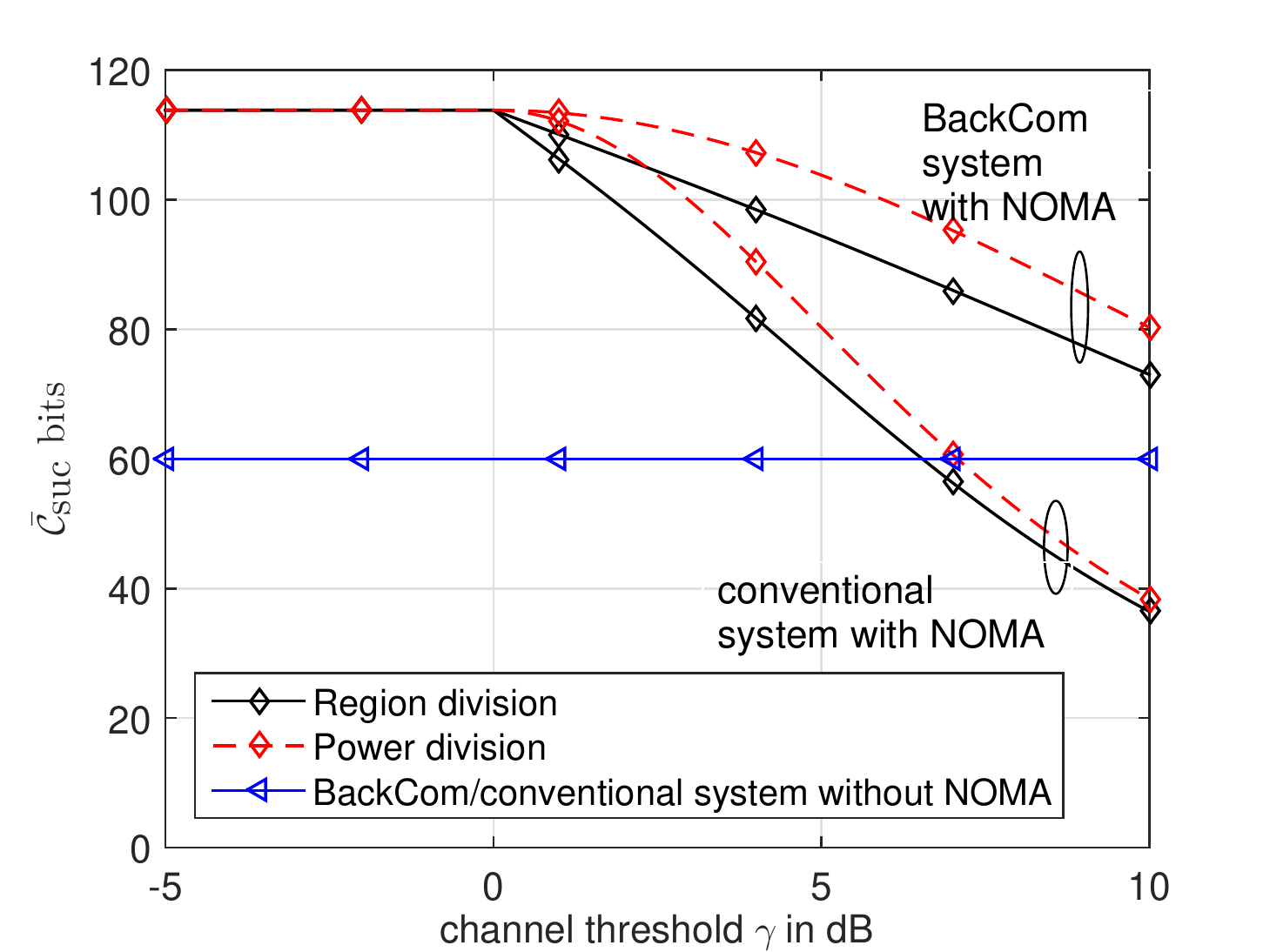}}
\mbox{\hspace{0.5cm}}
\subfigure[The ratio of $\bar{\mathcal{C}}_{\suc}$ for BackCom system with/without NOMA.]{ \label{fig5c}\includegraphics[width=0.45\textwidth]{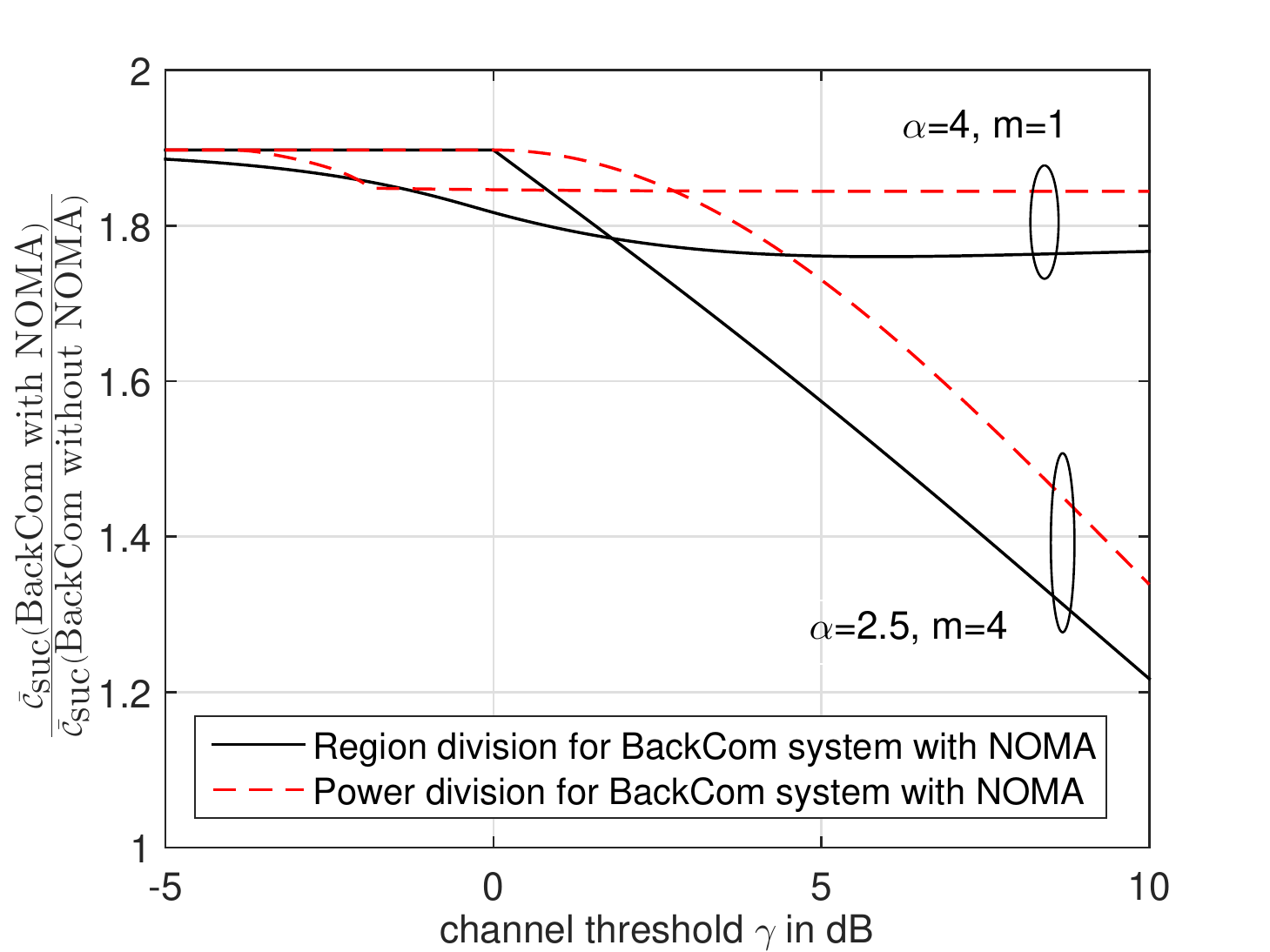}}

\caption{Channel threshold $\gamma$ versus (a) the average number of successfully decoded bits $\bar{\mathcal{C}}_{\suc}$ under $\alpha=2.5$, $m=4$ and (b) the ratio of $\bar{\mathcal{C}}_{\suc}$ for BackCom system with/without NOMA.}\label{fig_sec5}
\end{figure}
\else
\begin{figure}[t]
\centering
\subfigure[The average number of successfully decoded bits.]{\label{fig5a}\includegraphics[width=0.45\textwidth]{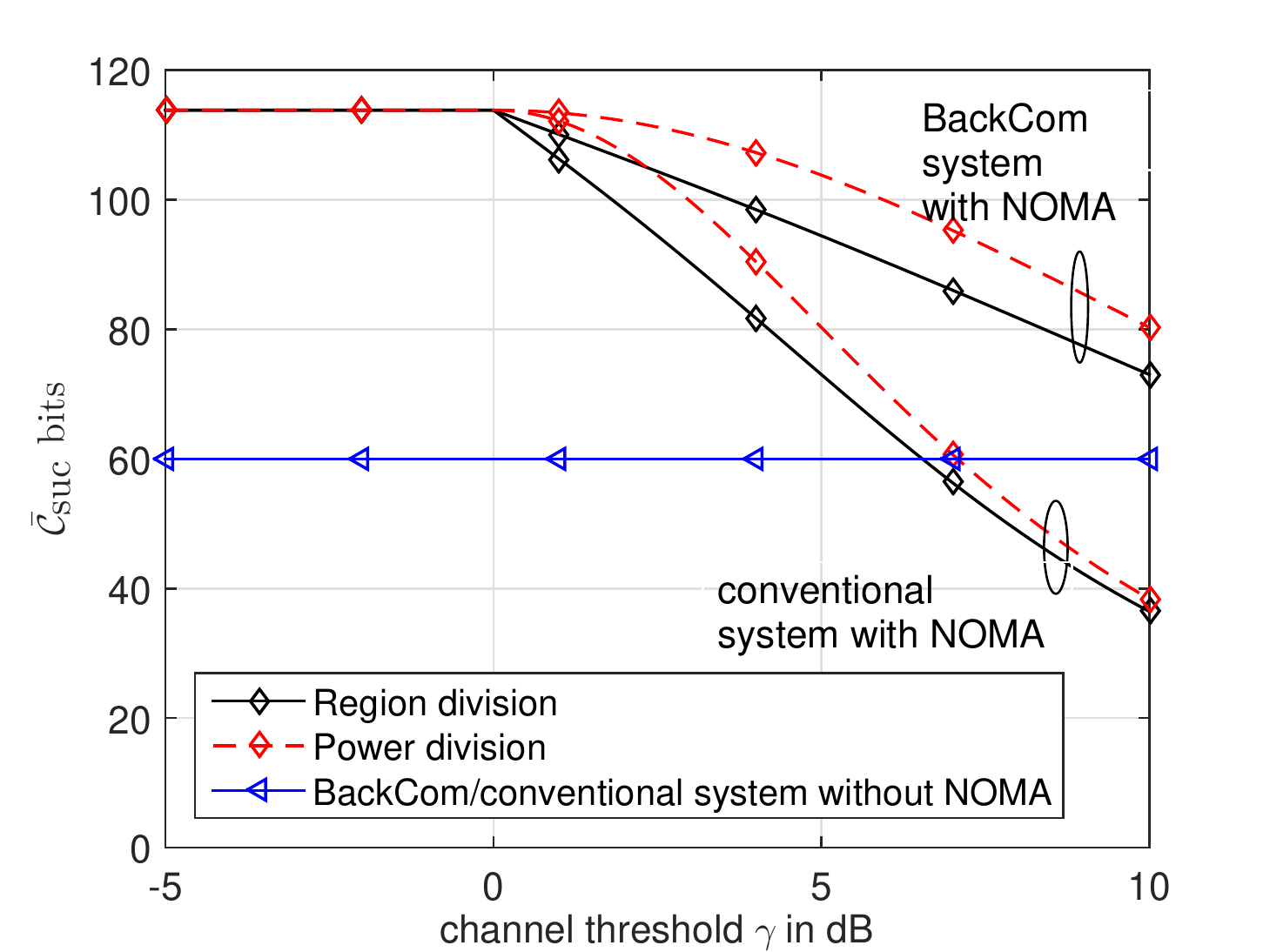}}\\
\subfigure[The ratio of $\bar{\mathcal{C}}_{\suc}$ for BackCom system with/without NOMA.]{ \label{fig5c}\includegraphics[width=0.45\textwidth]{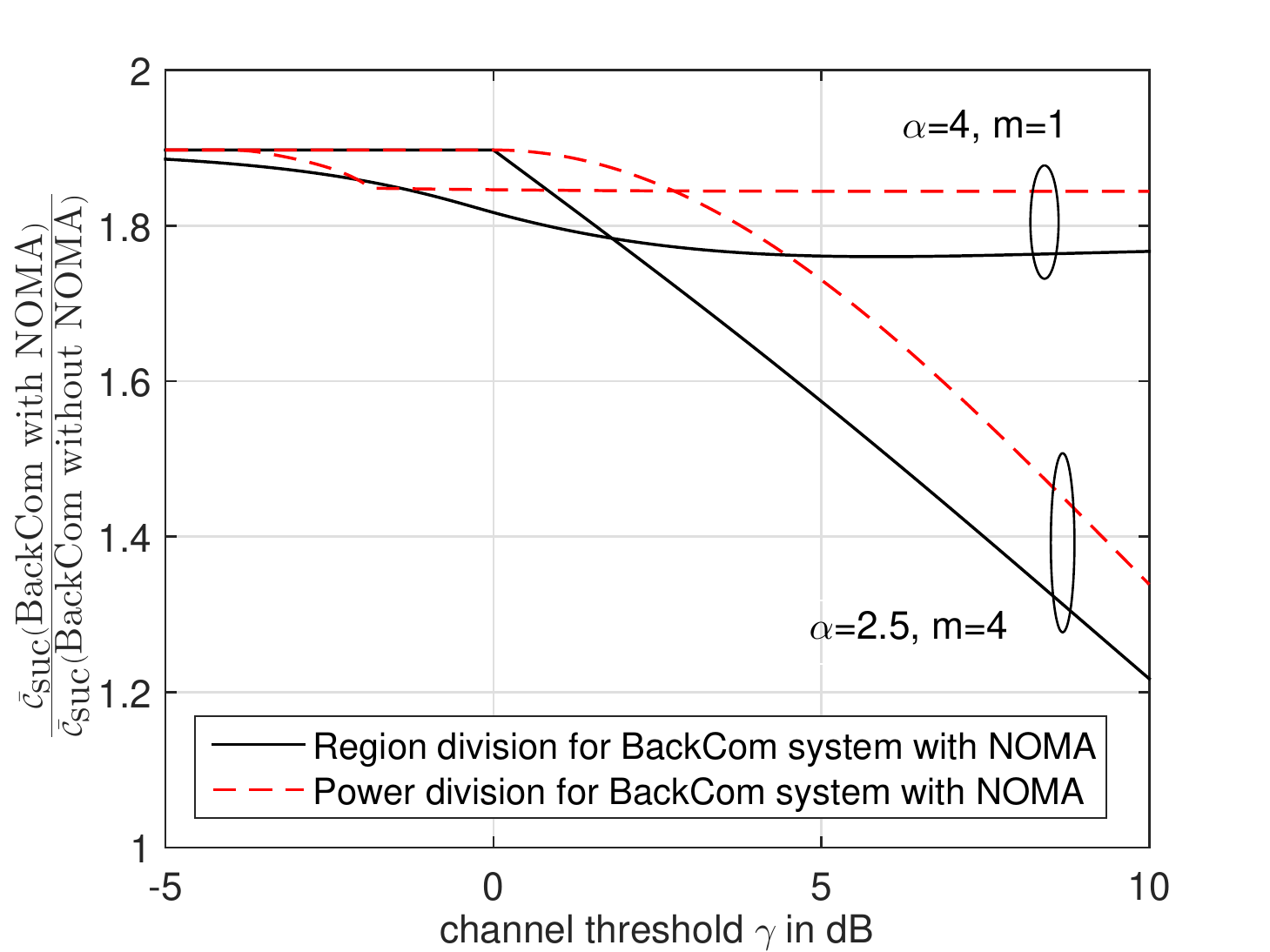}}

\caption{Channel threshold $\gamma$ versus (a) the average number of successfully decoded bits $\bar{\mathcal{C}}_{\suc}$ under $\alpha=2.5$, $m=4$ and (b) the ratio of $\bar{\mathcal{C}}_{\suc}$ for BackCom system with/without NOMA.}\label{fig_sec5}
\end{figure}
\fi
\ifCLASSOPTIONpeerreview
\else
   \begin{figure*}[!t]
\centering
\includegraphics[width=0.7 \textwidth]{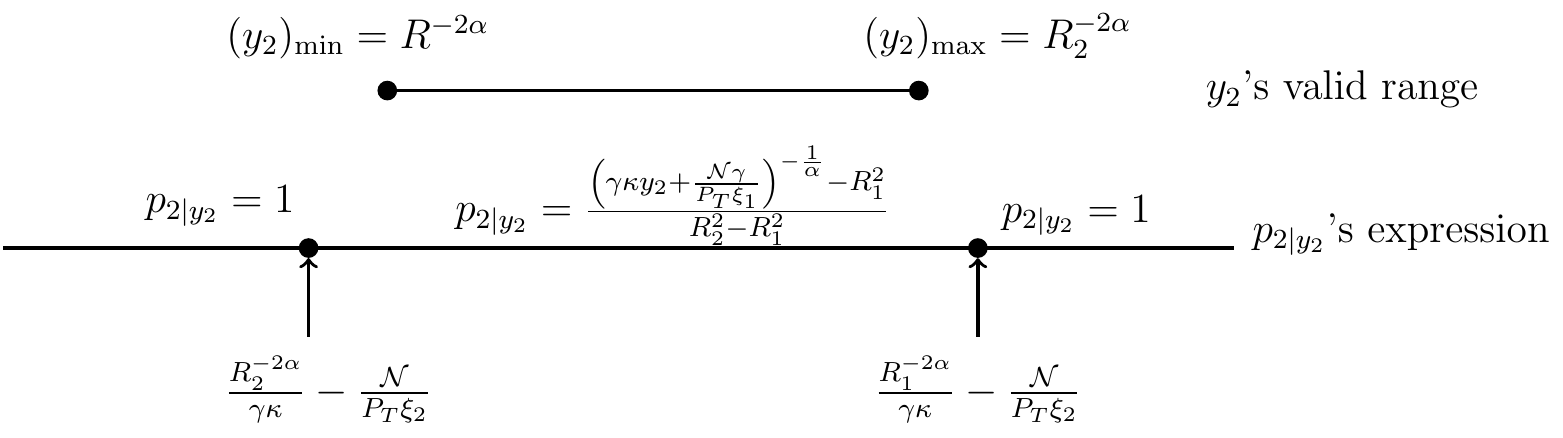}
\vspace{-5mm}
        \caption{Illustration of expressions of $p_{2|y_2}$ and the valid range of $y_2$ when $\frac{R_1^{-2\alpha}}{\gamma\kappa}-\frac{R_2^{-2\alpha}}{\gamma\kappa}\geq R_2^{-2\alpha}-R^{-2\alpha}$.}
        \label{fig_1ppendix2}
\vspace*{4pt}
\vspace{-0.05 in}
\end{figure*}
\fi
In this subsection, we evaluate the performance gain achieved by adopting NOMA into the BackCom system. For the purpose of comparison, we present the numerical results for the benchmark systems (i.e., the conventional system with/without NOMA and the BackCom system without NOMA). For the \textit{conventional communication system with NOMA}, the transmitting nodes are active devices and they use powered transceiver for the uplink communication. For the \textit{system without NOMA}, the BNs or conventional nodes access the reader in the pure TDMA fashion, e.g., only one BN/conventional node is scheduled to transmit signal to the reader per mini-slot lasting $\frac{\mathcal{L}}{M}$ seconds. The analytical results for these benchmark system can be derived using our analysis in this work. For the sake of brevity, we omit them here. Additionally, for a fair comparison among different communication systems, we assume that $\xi=0.7$ for all BNs and the transmit power for all conventional nodes are set to same, i.e., $20$ dBm.

Figs.~\ref{fig5a} plots the channel threshold $\gamma$ versus the average number of successfully decoded bits $\bar{\mathcal{C}}_{\suc}$. We first \textit{compare the BackCom system with the conventional system under NOMA scenario}. As shown in this figure, under the good channel condition (i.e., the channel condition tends to be LOS), the BackCom system has the larger average number of successfully decoded bits $\bar{\mathcal{C}}_{\suc}$ than the conventional system. This is mainly caused by the double attenuation of the received power at the reader for the BackCom system. This double attenuation effect can boost the performance of the BackCom system with NOMA under good channel condition. When the channel condition is good, the BNs are very likely to be successfully decoded alone. The double attenuation effect can make the channel gain between the stronger signal and weaker signal more distinguishable; hence, introducing the NOMA (i.e., bringing in the interference from the weaker signal) results in a small impact on the system.

We then \textit{compare the BackCom system with and without NOMA}, and we also plot the ratio of $\bar{\mathcal{C}}_{\suc}$ for these two systems in Fig.~\ref{fig5c}. From this figure, we can see that the BackCom system with NOMA generally leads to a better performance than the BackCom system without NOMA regardless of channel conditions. Under the case of same reflection coefficient, the system with NOMA allows two BNs to access the reader at the same time, which makes the reader experience the interference from the weaker signal when decoding the stronger signal. Hence, it is possible that less number of BNs can be successfully decoded when BNs are paired. However, in terms of the average number of successfully decoded bits, since the time on each mini-slot under NOMA is doubled, the BackCom system with NOMA can achieve larger $\bar{\mathcal{C}}_{\suc}$ than the system without NOMA. This illustrates why it is beneficial to apply NOMA to the BackCom system. In particular, by setting the proper reflection coefficients for the BackCom system with NOMA, the performance gain can be further improved.

\section{Conclusions}\label{sec:summary}
In this work, we have come up with a BackCom system enhanced by the power-domain NOMA, i.e., multiplexing the BNs located in different spatial regions or with different reflected power levels. Especially, the reflection coefficients for the BNs coming from different groups are set to be different such that the NOMA is fully utilized (i.e., increase the channel gain different for multiplexing BNs). In order to optimize the system performance, we provided the criteria for choosing the reflection coefficients for different groups of BNs. We also derived the analytical results for the average number of successfully decoded bits for two-node pairing case and the average number of successful BNs for the general multiplexing case. These derived results validated our proposed selection criteria. Our numerical results illustrated that NOMA generally results in the much better performance gain in the BackCom system than its performance gain in the conventional system. This demonstrated the significance of adopting NOMA with the BackCom system. Future work can consider the multiple readers scenario and the BNs powered by power beacons or ambient RF signals.
\section*{Appendix A}
\begin{proof}
\ifCLASSOPTIONpeerreview
\begin{figure}
\centering
\includegraphics[width=0.7 \textwidth]{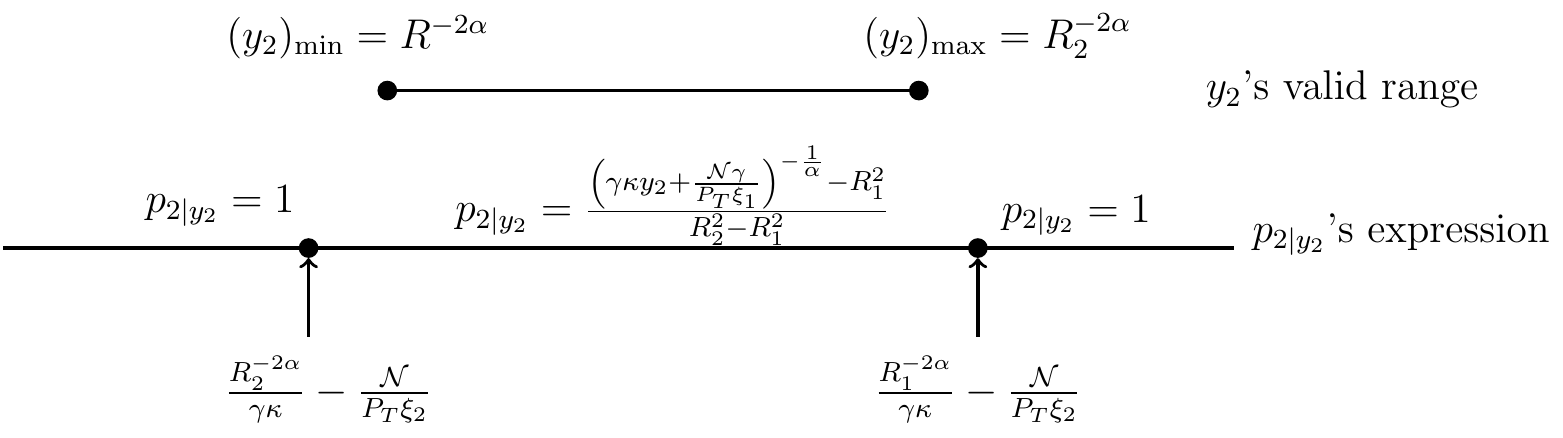}
\vspace{-5mm}
        \caption{Illustration of expressions of $p_{2|y_2}$ and the valid range of $y_2$ when $\frac{R_1^{-2\alpha}}{\gamma\kappa}-\frac{R_2^{-2\alpha}}{\gamma\kappa}\geq R_2^{-2\alpha}-R^{-2\alpha}$.}
        \label{fig_1ppendix2}
\end{figure}
\fi
Since we consider the $\xi_1\geq\xi_2$ scenario, the decoding order is always from the near BN to the far BN. The probability that both BNs are successfully decoded is given by
\ifCLASSOPTIONpeerreview
\begin{align}
p_2&=\Pr\left(\frac{P_T\xi_1 r_1^{-2\alpha}}{P_T\xi_2 r_2^{-2\alpha}+\mathcal{N}}\geq\gamma\,\,\&\& \,\,\frac{P_T\xi_2 r_2^{-2\alpha}}{\mathcal{N}}\geq\gamma\right)\nonumber\\
&=\Pr\left(y_1\geq\gamma\kappa y_2+\frac{\mathcal{N}\gamma}{P_T\xi_1}\,\,\&\&\,\,y_2\geq\frac{\mathcal{N}\gamma}{P_T\xi_2}\right)\nonumber\\
&=\left\{ \begin{array}{ll}
       0,\quad\quad\quad\quad\quad\quad\quad\quad\quad\quad\quad\quad\quad\quad\quad\,\,\,{\frac{\mathcal{N}\gamma}{P_T\xi_2}\geq R_2^{-2\alpha};} \\
       \mathlarger{\int}_{\textrm{min}\left\{\frac{\mathcal{N}\gamma}{P_T\xi_2},R^{-2\alpha} \right\}}^{R_2^{-2\alpha}}p_{2|y_2}f_{y_2}(y_2)\textup{d}y_2,\quad\quad\quad\quad{\frac{\mathcal{N}\gamma}{P_T\xi_2}<R_2^{-2\alpha};}\\
                    \end{array} \right.
\end{align}
\else
\begin{align}
p_2&=\Pr\left(\frac{P_T\xi_1 r_1^{-2\alpha}}{P_T\xi_2 r_2^{-2\alpha}+\mathcal{N}}\geq\gamma\,\,\&\& \,\,\frac{P_T\xi_2 r_2^{-2\alpha}}{\mathcal{N}}\geq\gamma\right)\nonumber\\
&=\Pr\left(y_1\geq\gamma\kappa y_2+\frac{\mathcal{N}\gamma}{P_T\xi_1}\,\,\&\&\,\,y_2\geq\frac{\mathcal{N}\gamma}{P_T\xi_2}\right)\nonumber\\
&=\left\{ \begin{array}{ll}
       0,\quad\quad\quad\quad\quad\quad\quad\quad\quad\quad\quad\quad\quad\quad\,{\frac{\mathcal{N}\gamma}{P_T\xi_2}\geq R_2^{-2\alpha};} \\
       \mathlarger{\int}_{\textrm{min}\left\{\frac{\mathcal{N}\gamma}{P_T\xi_2},R^{-2\alpha} \right\}}^{R_2^{-2\alpha}}p_{2|y_2}f_{y_2}(y_2)\textup{d}y_2,\quad\,\,{\frac{\mathcal{N}\gamma}{P_T\xi_2}<R_2^{-2\alpha};}\\
                    \end{array} \right.
\end{align}
\fi
\noindent where $y_1\triangleq r_1^{-2\alpha}$ with PDF $f_{y_1}(y_1)=\frac{y_1^{-\frac{1}{\alpha}-1}}{\alpha(R_2^2-R_1^2)}$ and $y_1\in\left[R_2^{-2\alpha},R_1^{-2\alpha}\right]$, $y_2\triangleq r_2^{-2\alpha}$ with PDF $f_{y_2}(y_2)=\frac{y_2^{-\frac{1}{\alpha}-1}}{\alpha(R^2-R_2^2)}$ and $y_2\in\left[R^{-2\alpha},R_2^{-2\alpha}\right]$, and $p_{2|y_2}$ is the conditional probability of $p_2$.

We first consider the case of $\frac{\mathcal{N}\gamma}{P_T\xi_2}<R^{-2\alpha}$, which implies that the weaker signal can be always successfully decoded given that the stronger signal is successfully decoded. Note that when $\gamma\kappa y_2+\frac{\mathcal{N}\gamma}{P_T\xi_1}\leq \left(y_1\right)_{\textrm{min}}=R_2^{-2\alpha}$ (e.g., $y_2\leq \frac{R_2^{-2\alpha}}{\gamma\kappa}-\frac{\mathcal{N}}{P_T\xi_2}$), the conditional probability $p_{2|y_2}$ is always equal to one. When $\gamma\kappa y_2+\frac{\mathcal{N}\gamma}{P_T\xi_1}\geq \left(y_1\right)_{\max}=R_1^{-2\alpha}$ (e.g., $y_2\geq \frac{R_1^{-2\alpha}}{\gamma\kappa}-\frac{\mathcal{N}}{P_T\xi_2}$), $p_{2|y_2}$ is always equal to zero. For the remaining range of $y_2$, $p_{2|y_2}=\int_{\gamma\kappa y_2+\frac{\mathcal{N}\gamma}{P_T\xi_1}}^{R_1^{-2\alpha}}\frac{y_1^{-\frac{1}{\alpha}-1}}{\alpha(R_2^2-R_1^2)}\textup{d}y_1=\frac{\left(\gamma\kappa y_2+\frac{\mathcal{N}\gamma}{P_T\xi_1}\right)^{-\frac{1}{\alpha}}-R_1^2}{R_2^2-R_1^2}$.

Based on the expressions of $p_{2|y_2}$ and $y_2$'s valid range, when $\frac{R_1^{-2\alpha}}{\gamma\kappa}-\frac{R_2^{-2\alpha}}{\gamma\kappa}\geq R_2^{-2\alpha}-R^{-2\alpha}$, we can plot a diagram in Fig.~\ref{fig_1ppendix2} to help finding the integration limits. From Fig.~\ref{fig_1ppendix2}, we obtain the final expression of $p_2$ as
\begin{itemize}
\item $\gamma\leq\frac{R_2^{-2\alpha}}{\kappa R_2^{-2\alpha}+\frac{\mathcal{N}}{P_T\xi_1}}$: $p_2=1$;

\item $\gamma\geq \frac{R_1^{-2\alpha}}{\kappa R^{-2\alpha}+\frac{\mathcal{N}}{P_T\xi_1}}$: $p_2=0$;

\item Other range:

 $p_2=\mathlarger{\int}_{R^{-2\alpha}}^{\max\left\{\frac{R_2^{-2\alpha}}{\gamma\kappa}-\frac{\mathcal{N}}{P_T\xi_2},R^{-2\alpha} \right\}}f_{y_2}(y_2)\textup{d}y_2+\mathlarger{\int}_{\max\left\{\frac{R_2^{-2\alpha}}{\gamma\kappa}-\frac{\mathcal{N}}{P_T\xi_2},R^{-2\alpha} \right\}}^{\min\left\{\frac{R_1^{-2\alpha}}{\gamma\kappa}-\frac{\mathcal{N}}{P_T\xi_2},R_2^{-2\alpha} \right\}}\frac{\left(\gamma\kappa y_2+\frac{\mathcal{N}\gamma}{P_T\xi_1}\right)^{-\frac{1}{\alpha}}-R_1^2}{R_2^2-R_1^2} f_{y_2}(y_2)\textup{d}y_2$.
\end{itemize}

We note the above expressions of $p_2$ also hold for $\frac{R_1^{-2\alpha}}{\gamma\kappa}-\frac{R_2^{-2\alpha}}{\gamma\kappa}< R_2^{-2\alpha}-R^{-2\alpha}$. For other cases, we can adopt the similar steps to work out $p_2$. After further computation and simplification, we arrive at the result in~\eqref{eq:nofading:p2}.
\end{proof}

 \bibliographystyle{IEEEtran}

\end{document}